\documentclass[11pt,letterpaper]{article}
\usepackage[doublespacing]{setspace} 

\usepackage{lmodern}
\usepackage[utf8]{inputenc}
\usepackage[italian,english]{babel}
\usepackage{geometry}
\usepackage{booktabs}
\usepackage{amsmath,amssymb,amsfonts,amsthm}
\usepackage{dsfont}
\usepackage{cancel}
\usepackage{bm}
\usepackage{enumitem}
\usepackage{mathtools}
\usepackage{changepage}
\usepackage{verbatim}
\usepackage{graphicx}
\usepackage{emptypage}
\usepackage{newlfont}
\usepackage[all,cmtip]{xy}
\usepackage[bottom]{footmisc}
\usepackage[titletoc,title]{appendix}
\usepackage{chngcntr}
\usepackage{apptools}
\usepackage{listings}
\usepackage{color}
\usepackage{sgame, tikz} 
\usepackage{kpfonts}

\usepackage{titling}
\usepackage{caption}
\usepackage{subcaption}
\usepackage{float}
\usepackage{xr-hyper}
\usepackage[colorlinks=true,linkcolor=blue, allcolors=blue]{hyperref}
\usepackage{sgamevar}

\newtheoremstyle{bfnote}%
{}{}%
{\itshape}{}%
{\bfseries}{.}%
{ }%
{\thmname{#1}\thmnumber{ #2}\thmnote{ (#3)}}
\theoremstyle{bfnote}

\newtheorem{theorem}{Theorem}
\newtheorem*{theorem*}{Theorem}
\newtheorem{corollary}{Corollary}
\newtheorem{remark}{Remark}
\newtheorem{lemma}{Lemma}
\newtheorem{example}{Example}
\newtheorem{proposition}{Proposition}
\newtheorem{definition}{Definition}

\newcommand{\rrparen}{\mathclose{)\mkern-3mu)}}

\reversemarginpar

\usepackage[ruled]{algorithm2e}

\usepackage{natbib}
\usepackage[bibliography=common]{apxproof}

\usepackage[normalem]{ulem}

\newtheoremrep{proposition}[]{Proposition}[]
\newtheoremrep{lemma}[]{Lemma}[]
\newtheoremrep{theorem}[]{Theorem}[]
\newtheoremrep{corollary}[]{Corollary}[]
\newtheoremrep{remark}[]{Remark}[]
\usepackage{multirow,array}

\DeclareMathOperator*{\argmin}{arg\,min}
\DeclareMathOperator*{\argmax}{arg\,max}

\usepackage[absolute,overlay]{textpos}

\title{ \vspace{-1.4cm}
Equilibria under Dynamic Benchmark Consistency \\\vspace{-0.2cm} in Non-Stationary Multi-Agent Systems\thanks{We are thankful to Itai Ashlagi, Mohsen Bayati, Kostas Bimpikis, Dan Iancu, Giacomo Mantegazza, Ilan Morgenstern, Seyeon Kim, Daniela Saban, Gabriel Weintraub, Weijie Zhong, Andrea Zorzi, and seminar participants at Stanford University. }
}

\author{[Authors’ Names and Institutions Removed for Peer Review]\vspace{1.6cm}}
\author{Ludovico Crippa\thanks{Graduate School of Business, Stanford University. Email: lcrippa@stanford.edu}\and Yonatan Gur\thanks{Netflix. Email: ygur@netflix.com} \and Bar Light\thanks{Business School and Institute of Operations Research and Analytics, National University of Singapore, Singapore. Email: barlight@nus.edu.sg } }

\begin{document}
{\vspace{-5em}
\maketitle}\vspace{-0.8cm}
\begin{abstract}\small\setstretch{1.0}\small
\noindent
We formulate and study a general time-varying multi-agent system where players repeatedly compete under incomplete information. Our work is motivated by scenarios commonly observed in online advertising and retail marketplaces, where agents and platform designers optimize algorithmic decision-making in dynamic competitive settings. In these systems, no-regret algorithms that provide guarantees relative to \emph{static} benchmarks can perform poorly and the distributions of play that emerge from their interaction do not correspond anymore to static solution concepts such as coarse correlated equilibria. Instead, we analyze the interaction of \textit{dynamic benchmark} consistent policies that have performance guarantees relative to \emph{dynamic} sequences of actions, and through a novel \textit{tracking error} notion we delineate when their empirical joint distribution of play can approximate an evolving sequence of static equilibria. In systems that change sufficiently slowly (sub-linearly in the horizon length), we show that the resulting distributions of play approximate the sequence of coarse correlated equilibria, and apply this result to establish improved welfare bounds for smooth games. On a similar vein, we formulate internal dynamic benchmark consistent policies and establish that they approximate sequences of correlated equilibria. 
Our findings therefore suggest that in a broad range of multi-agent systems where non-stationarity is prevalent, algorithms designed to compete with dynamic benchmarks can improve both individual and welfare guarantees, and their emerging dynamics approximate a sequence of static equilibrium outcomes. 
\end{abstract}
\normalsize

\setstretch{1.46}

\vspace{-0.2cm}
\section{Introduction}\label{Introduction}\vspace{-0.1cm}
\subsection{Motivation}\label{Motivation}\vspace{-0.1cm}

Sequential games of incomplete information arise in a broad array of application domains, including pricing, inventory management, and transportation networks. In recent years, the relevance of this game-theoretic framework has further increased with the rising popularity of online platforms in which self-interested agents compete under incomplete information. Prominent examples include retail platforms such as Amazon and eBay, and advertising platforms such as Google's Double Click and Bing Ads. In these instances, sellers (or advertisers) repeatedly interact with one another, while lacking real-time information on their competitors and their actions, and might have limited knowledge of their own payoff structure. The growing complexity of such platforms has led to decision-making becoming increasingly algorithmic, often relying on automated algorithms, implemented by the competing agents or,  on their behalf, by the platform to navigate the dynamically evolving competitive environment.\footnote{Notable examples include the central bidding tools offered by Meta ads, and dynamic pricing tools offered by platforms such as Airbnb and Amazon.} 

As in these settings the limited information often prevents agents from identifying in real time actions that best respond to the competitive environment, a common approach taken by agents and platform designers to identifying effective algorithms has been to compare the performance these algorithms achieve relative to a benchmark endowed with the benefit of \emph{hindsight}, i.e., with complete knowledge of the payoff structure and the strategies chosen by the other players. This approach has been particularly appealing as, in many practical settings, the best performance in hindsight can be computed using available historical data, and requires no behavioral assumptions (see discussion in \citealt{talluri2006theory}, Ch. 11.3). Specifically, a policy is said to be ``\emph{no-regret}" (also referred to as \emph{Hannan consistent}, see \citealt{Hannan}, \citealt{cesa2006prediction}) if, for all possible profiles of policies of the other players, it guarantees an average payoff at least as high as any \emph{static} action that can be selected with the benefit of hindsight.

Nevertheless, no-regret policies can fall short in dynamic environments, such as the ones that typically emerge in contemporary online marketplaces, where competitors regularly enter and exit, consumer demand patterns shift over time, and underlying technologies evolve. In such settings, even the best static action can perform poorly, which in turn, can lead to suboptimal decisions for the agents and low overall welfare for the platform. 

Moreover, another challenge platforms encounter in non-stationary systems is maintaining stability and predictability of the dynamics that arise when individual agents deploy no-regret algorithms. One appealing feature of these algorithms in repeated games is that the empirical joint distributions of play they induce correspond to the static solution concept of coarse correlated equilibria (e.g., \citealt{HMSC3} and \citealt{NisaRougTardVazi07}). However, 
in non-stationary systems static equilibrium notions become ambiguous, as different equilibria can emerge in different time periods, and the distribution of play that arise may no longer correspond to~such~equilibria.

These limitations raise fundamental questions about the nature of learning and equilibria in non-stationary multi-agent systems. In particular, what are the possible empirical joint distributions of play that can emerge when platforms, or their participants, adopt algorithms designed to handle dynamically changing conditions? Can these distributions be approximated by some static equilibrium concept, that are therefore implementable by the platform, while ensuring that outcomes remain efficient even when the environment fluctuates? These questions are critical for both platform designers and participants. 

To address them, we introduce and study a general non-stationary multi-agent system to examine the dynamics that can emerge from the interaction of policies that are responsive to changing conditions, and particularly ones with performance guarantees against dynamic \emph{sequences} of actions. We focus on characterizing the outcomes these policies produce and quantifying their welfare properties. Our analysis reveals when and how the deployment of such policies can lead to predictable dynamics and improve both individual and collective performance guarantees. 

To further motivate this study, we first illustrate the limitations of using static benchmarks in non-stationary multi-agent systems, both from the agent and from the platform point of view. We then turn to describe the main contributions of the paper.

\subsection{Limitation of static benchmarks in non-stationary multi-agent systems}

From an agent's perspective, two sources of non-stationarity can reduce the effectiveness of any static action as a benchmark. First, as common in multi-agent systems, even when the stage game remains fixed, changes in competitors' actions (for example, due to exploration or beliefs' update) can significantly affect the payoff an agent receives, potentially altering which actions are most preferred. Second, the game itself often changes over time: For instance, in digital advertising markets, advertisers' valuations for impressions depend on user demographics or the keywords being searched, and these valuations evolve over time. Similarly, in retail markets, conditions such as demand and supply elasticities often change. The following simple example highlights the limitations of static actions in such changing environments. 

\begin{example}[Evolving market conditions]\label{pricing game example}
Consider a selling season of even length $T \geq 2$ where two sellers compete over a market with $N$ customers having logit demand. For $i\in\{1,2\}$, given prices $p^i$ and~$p^{-i}$, the expected demand for seller~$i$ has the form:

\[ 
d^i(p;N,(\alpha^j)_{j=1}^{2},(\beta_t^j)_{j=1}^{2}) = N \frac{\exp(\alpha^i-\beta_t^i p^i)}{1+\exp(\alpha^i-\beta_t^i p^i) + \exp(\alpha^{-i}-\beta_t^{-i} p^{-i})}.
\]

\noindent For simplicity, we normalize production costs to zero, assume there are only two available prices, $p_l =  1$ and $p_h =  2 $, and focus on a symmetric instance with $\beta_t^i = \beta_t$ and $\alpha^i =~\alpha$, for each player $i$. Let $\alpha = 4$ and suppose that at the beginning of the selling season customers have a low price-sensitivity, $\beta_t = 3/4$, which increases to $\beta_t = 7/4$ in the second half. For each seller, $p_h$ turns out to be the dominant price in the first half of the season, while $p_l$ becomes dominant in the second half. Therefore, a dynamic benchmark that switches prices once can always charge the dominant price in each stage game. In contrast, any static action chosen in hindsight is suboptimal for part of the season, and it incurs linear loss $\Theta(T)$ compared to a dynamic benchmark that plays $p_h$ in the first half of the season and then switches to $p_l$ in the second.
\end{example}

While Example~\ref{pricing game example} illustrates the limitations of no-regret algorithms from the agents' perspective, it also reveals that the distributions of play that emerge from their interaction may no longer correspond to the coarse correlated equilibria of the stage games. Notably, in the simple setting of Example~\ref{pricing game example}, the set of coarse correlated equilibria reduces to the unique Nash equilibrium of the stage games, changing from~($p_h,p_h$) in the first half of the horizon to~($p_l,p_l$) in the second half.\footnote{Formally, coarse correlated equilibria are joint distributions. Therefore, in the context of Example \ref{pricing game example}, they correspond to $\delta_{p_h,p_h}$ and $\delta_{p_l,p_l}$. In the rest of this discussion we stick to $(p_h,p_h)$ and $(p_l,p_l)$ to ease the exposition.} Figure \ref{fig:comparison_Exp3Intro} illustrates a scenario where, in the setting of Example~\ref{pricing game example}, each seller is employing Exp3 (\citealt{auer2002nonstochastic}),  a classical no-regret algorithm designed for adversarial environments. The figure depicts the $\ell_2$-distance between the distribution of play generated by the Exp3 algorithms and the Nash equilibria of the stage games in the first and second half of the selling season.

\begin{figure}[H]
    \centering
\includegraphics[width=0.6\linewidth]{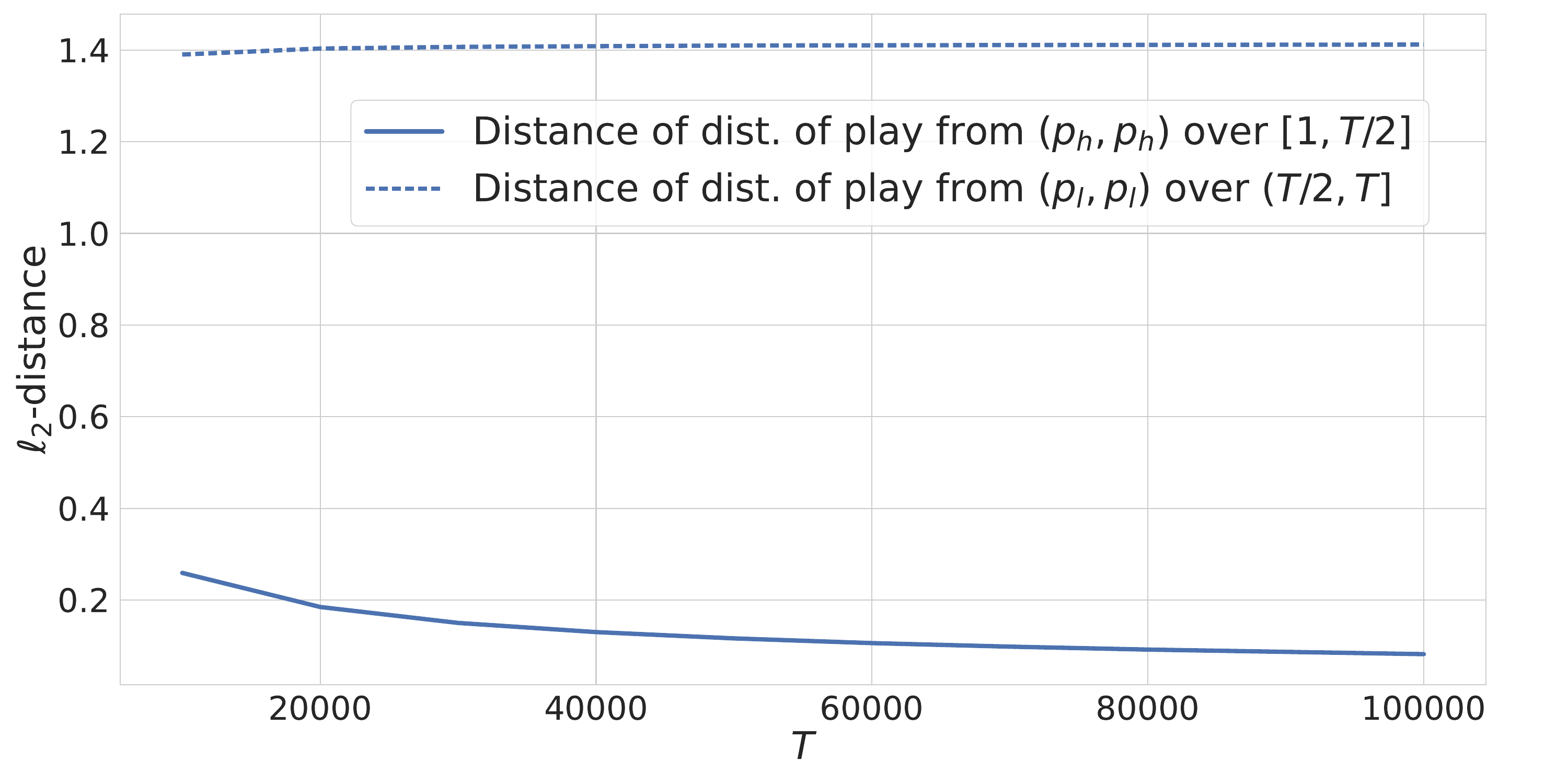}\vspace{-0.2cm}
    \caption{\footnotesize Distance between the average empirical joint distribution of play and the coarse correlated equilibrium of each stage game over the first and second halves of the season in Example \ref{pricing game example}. Sellers employ the Exp3 algorithm tuned with exploration parameter $\gamma_T =~\sqrt{2 \ln(2)/((e-1)T)}$. The figure shows the~$\ell_2$-distance of the empirical joint distribution of play in the first and second half from $(p_h,p_h)$ and $(p_l, p_l)$ respectively. The empirical joint distribution of play is estimated by averaging the results over $750$ simulations.}
    \label{fig:comparison_Exp3Intro}
\end{figure}\vspace{-0.3cm}
\noindent
Over the first half of the season, the distribution of play converges to $(p_h,p_h)$, which the unique coarse correlated equilibrium during that time segment.~However, in the second half, following the shift in customers' price sensitivity, the distribution of play does not converge to the new coarse correlated equilibrium,~$(p_l, p_l)$. We will revisit this example in \S4, after characterizing a broad property, termed \textit{dynamic benchmark consistency}, under which the joint distribution of play is guaranteed to track the sequence of stage game coarse correlated equilibria. We next discuss the main contributions of the paper.

\subsection{Main contributions}\label{Main contributions}

The main contribution of the paper is the formulation and analysis of a general time-varying multi-agent system where players address the changing nature of the environment by relying on policies with performance guarantees relative to dynamic sequences of actions. We characterize the outcomes that can emerge, and identify when the agents' interaction can lead to predictable dynamics and improve both individual and collective performance guarantees. Our analysis provides new insights into the dynamics and efficiency of time-varying multi-agent systems and offers a better understanding of competitive interactions in changing environments. In more detail, the contribution of the paper is along the following dimensions.

\textbf{Formulation.} We formulate a flexible time-varying multi-agent system where a general-sum stage game can change an arbitrary number of times along a horizon of $T$ periods, and denote its number of changes by $V_T$. Notably, our formulation is a generalization of the classical repeated games setting, which is retrieved in the special case when $V_T = 0$. 

In the context of the dynamic systems we consider, static equilibrium concepts become ambiguous. We therefore propose a notion of \emph{tracking error} that extends static equilibrium notions (including correlated ones) to time-varying systems. For any profile of policies deployed by the players and some class of static equilibria (e.g., Nash equilibria, correlated equilibria, or coarse correlated equilibria), our notion of tracking error sums, over the time segments in which the stage game does not change, the distance between the empirical distribution of play that arises in these segments, and the corresponding equilibria, weighted by the number of time periods in that segment. If for some profile of policies the tracking error grows sublinearly with the time horizon~$T$, then, on average, their joint distribution of play tracks the sequence of equilibria.

We focus on a broad class of policies that are designed for competitive changing environments, in the sense of satisfying the following criterion. Consider a sequence of time periods~$1,2,\ldots,T$, and a non-negative number $C_T$. A policy is said to be \emph{dynamic benchmark consistent} with respect to $C_T$ if, for all strategies of the other players, the per-period expected regret it exhibits relative to the best dynamic sequence of actions that can be selected in hindsight, and that has at most~$C_T$ action changes, diminishes to zero as the horizon length~$T$ grows large.

\textbf{Equilibrium analysis.} We first analyze the classical repeated game setting~($V_T = 0$), where only the endogenous changes in the actions of the players can cause changes to the payoff each of them receives. In this baseline case, we provide a \emph{tight} equilibrium characterization for the joint distribution of play generated by the interaction of dynamic benchmark consistent policies. We establish that, despite providing stronger performance guarantees compared to no-regret algorithms, asymptotically, the set of admissible empirical joint distributions of play that can emerge under the interaction of dynamic benchmark consistent policies remains the one of coarse correlated (Hannan) equilibria. While a direct analysis is non-trivial due to the dynamic nature of the benchmarks, we are able to establish this characterization by showing that, when the horizon length grows large, the set of possible joint distributions of play that can emerge is \textit{independent} of the number of changes allowed in the benchmark. 

In particular, consider any pair $\tilde{C}_T$ and~$C_T$ of upper bounds on the number of action changes allowed in the benchmark sequence. We show that, as the horizon length grows, any joint distribution of play that can emerge when players deploy policies that are dynamic benchmark consistent subject to $\tilde{C}_T$, can be approximated arbitrarily well when players deploy policies that are dynamic benchmark consistent with respect to~$C_T$. As this result also holds in the case in which~$\tilde{C}_T = 0$, it includes no-regret policies as well. By leveraging this approximation, we show that any equilibrium in the set of coarse correlated equilibria can emerge under the interaction of dynamic benchmark consistent policies. This establishes that when $V_T = 0$, the tightest characterization one may provide for the distribution of play generated by dynamic benchmark consistent policies coincides with the set of coarse correlated equilibria. 

Then, when $V_T \geq 1$, we delineate through the novel notion of tracking error when the distribution of play that emerges under the deployment of these policies can approximate an evolving sequences of coarse correlated equilibria, leading to predictable system dynamics. Specifically, our main result establishes that, when the underlying games are changing sufficiently slowly, that is, $V_T = o(T)$, and enough action changes are allowed in the benchmark, that is, $C_T \geq V_T$, mild conditions on the class of dynamic benchmark consistent policies guarantee sub-linear tracking error. We demonstrate that the required conditions are satisfied by typical dynamic benchmark consistent policies, including, for example, the Exp3S algorithm (see \citealt{auer2002nonstochastic}), as well as any suitably restarted no-regret policy. Therefore, in slowly changing multi-agent systems, on average, the distribution of play induced by dynamic benchmark consistent policies can track the evolving sequence of coarse correlated equilibria. 

We further analyze the important class of games where the players' payoff functions have a single best-reply to the actions of the other agents, which applies to many practical settings of algorithmic decision-making, including pricing (e.g., Example \ref{pricing game example}), and truthful mechanisms (e.g. second-price auctions). For this class of games, we show that the tracking guarantees extend to \emph{all} dynamic benchmark consistent policies with~$C_T \geq V_T$ and $C_T = o(T)$. Moreover, we provide the following converse to the former result for sequences of single best-reply payoff functions. Suppose that player $i$ is using a dynamic benchmark consistent policy with number of action changes strictly below the number of changes in the sequence of their payoff functions. Then, there exists a sequence of stage games for which the tracking error is linear in the horizon length~$T$, even if the other players are playing their single best-reply at every period. 

We demonstrate the breadth of our formulation and techniques by extending our analysis to no-internal regret policies (\citealt{blum2007internalregret}), which have been extensively studied in repeated games and were shown to induce distributions of play that converge to the set of correlated equilibria. We introduce the property of internal dynamic benchmark consistency (IDB) by incorporating time-dependent action changes in the internal regret benchmark, and we show that when~$V_T \leq C_T$ and $C_T = o(T)$, the empirical joint distribution of play induced by IDB-consistent policies exhibits diminishing tracking error relative to sequences of correlated equilibria.

\textbf{Welfare analysis.} Our results imply that whenever the stage game may change, a critical distinction between traditional no-regret and dynamic benchmark consistent policies emerges: Even in simple settings, the deployment of no-regret algorithms can result in \emph{linear} tracking error. By using a (dynamic) Price of Anarchy metric to quantify the worst-case efficiency of the possible outcomes, we show that this distinction has important implications in terms of system \emph{welfare}. When the stage game is fixed, we show that, asymptotically, the worst-case social welfare resulting from the interaction of dynamic benchmark consistent policies is independent of the number of action changes in the benchmark sequence, and it is always equivalent to the worst-case social welfare achievable under coarse correlated equilibria. Therefore, the tightness of the existing price of anarchy bounds, such as the ones obtained for $(\lambda,\mu)$-smooth games \citep{robustpoa}, remains as the ones achievable for traditional no-regret policies. 

On the other hand, when the stage game may change, we establish that dynamic benchmark consistent policies can provide \emph{better} welfare guarantees compared to traditional no-regret algorithms. In particular, when players face sequences of time-varying $(\lambda,\mu)$-smooth games, the fraction of the optimal social welfare that can be guaranteed scales with the number of changes that are allowed in the benchmark sequences, and thus dynamic benchmark consistent policies can provide better welfare guarantees compared to traditional no-regret algorithms.

Our findings have implications for the design and operations of online marketplaces, particularly in retail and digital advertising, where algorithmic agents repeatedly compete under incomplete information in a changing environment. Our results indicate that algorithms designed to compete with dynamic benchmarks can be deployed by agents (or by the platform on their behalf) to improve both individual and welfare guarantees. Furthermore, the dynamics that emerge from the interaction of dynamic benchmark consistent algorithms approximate a sequence of static equilibrium outcomes, which are therefore implementable by the platform.

\section{Literature Review}

Our paper is related to several strands of literature. We discuss each of them separately below.

\textbf{Single-agent learning.} Several literature streams have considered single-agent settings where in each period a decision maker takes an action, and has access to some information feedback that depends on the selected action. When the payoff function is fixed throughout the horizon, these settings include the stochastic approximation framework (e.g., \citealt{robbins1951stochastic}, \citealt{nemirovskij1983problem}) as well as the stochastic multi-armed bandit (MAB) problem (e.g., \citealt{lai1985asymptotically}, \citealt{ auer2002finite}). Other settings allow the payoff function to arbitrarily change over time, including the online convex optimization and the adversarial MAB settings (e.g., \citealt{zinkevich2003online}, \citealt{auer2002nonstochastic}).
The limitations of a single best action as a yardstick to identify effective policies and the possibility of achieving long-run average performance that asymptotically approaches the best sequence of expected rewards has been analyzed in various single-agent settings. Examples include adversarial settings (\citealt{auer2002nonstochastic}, \citealt{bubeck2010}) where dynamic benchmarks are constrained to a fixed number of action changes, the non-stationary stochastic settings (\citealt{besbes2015, besbes2019} and \citealt{cheung2022hedging}) where changes in the environment are controlled by a global variation budget, as well as the adaptive regret framework (e.g., \citealt{hazan2007adaptive}, \citealt{daniely2015strongly}). While these studies are concerned with the design of algorithms in single-agent settings, our work focuses on a multi-agent framework and analyzes the dynamics that emerge when agents use policies with guarantees relative to dynamic benchmarks.

\textbf{Learning in games.} A rich literature stream has been considering learning in game-theoretic settings. Our paper is closely related to the study of \emph{no-regret dynamics}, originating with the pioneering works by \cite{blackwell1956} and \cite{Hannan}. A key result in this line of work is that when all players adopt no-regret policies, the empirical joint distribution of play converges to the Hannan set, also called the set of coarse correlated equilibria (see, e.g., \citealt{HMSC3}). A related notion is the one of \emph{internal} regret (\citealt{foster1997calibrated},  \citealt{HMSC1, HMSC2, HMSCreinforcement}). When all players deploy no-internal regret policies, the corresponding distribution of play converges to the set of correlated equilibria. Specific algorithms that, when jointly deployed by all players, achieve near-optimal convergence rates to coarse correlated and correlated equilibria have been established respectively by \cite{daskalakis2021near} and \cite{anagnostides2022near}. In specific settings, no-regret algorithms where shown to converge to stronger equilibrium notions such as Nash, Bayes-Nash, and Dynamic-Nash; see, e.g., \cite{balseiro2019learning}, \cite{bichler2023computing}, \cite{ba2021doubly} and references therein.

No-regret dynamics have also been studied in specific classes of time-varying games; for instance, \cite{duvocelle2022multiagent} and \cite{yan2023fast} analyze the possibility of tracking changing sequences of Nash equilibria in monotone and strongly monotone time-varying games. A similar line of research is considered for classes of time-varying concave-continuous games (\citealt{mertikopoulos2021equilibrium}), and zero-sum games (\citealt{cardoso2019competing}, \citealt{zhang2022no}, \citealt{fiez2021online}, and \citealt{feng2024last}). In addition, \cite{foster2016learning} and \cite{lykouris2016learning} study games with changing populations where players can be substituted at each round with exogenous probability $p$, and provide bounds on the worst-case welfare of the outcomes. Our analysis in \S \ref{sect Welfare in time-varying games} can be viewed as extending their findings to a broader variation metric for sequences of games and not only random replacement of the agents. Recently, \cite{anagnostides2024convergence} consider time-varying games where optimistic mirror descent dynamics are studied under a set of learning agents that is artificially augmented, by jointly considering the learning problem of a central mediator who recommends actions to the game participants. In contrast with the above works, the current paper does not aim at studying in-game dynamics of specific algorithmic or game structure, but rather to provide a general framework to assess equilibria and welfare consequences of deploying a broad class of algorithms with guarantees relative to dynamic benchmarks.

\textbf{Online marketplaces}. In the context of online marketplaces, there has been a recent surge in the study of algorithmic decision-making in single and multi agent settings. A prominent example is digital advertising, where automated bidding strategies, or auto-bidders, have become prevalent in digital advertising platforms. Motivated by the practical features of no-regret algorithms in these settings, individual performance and welfare guarantees of these algorithms, along with related equilibrium properties in digital advertising context, have been extensively studied in the recent operations literature (e.g., \citealt{balseiro2019learning}, \citealt{balseiro2021budget}, \citealt{conitzer2022pacing},  \citealt{gaitonde2022budget}, \citealt{aggarwal2024no}, \citealt{han2024optimal} and \citealt{kumar2024strategically}). In addition, dynamic pricing algorithms couched in incomplete information settings are widely used in ride-sharing platforms, rental sharing platforms, and e-commerce platforms, and have been analyzed in various application domains (e.g., \citealt{banerjee2015pricing},  \citealt{papanastasiou2017dynamic}, \citealt{varma2023dynamic}, \citealt{shin2023dynamic}, \citealt{carnehl2024pricing}, and \cite{perakis2024dynamic} to name a few). Our paper contributes to this literature by addressing the inherently competitive and time-varying nature of these settings, where demand structure and product valuations typically change over time. We provide a general framework that can accommodate such dynamic environments and analyze the possible outcomes achievable by algorithms designed to perform well under these time-varying conditions.

\section{Model}\label{Model and Preliminary Results}
\subsection{Primitives and notations}

We consider a time-varying multi-agent system comprised of $N$ players and a horizon of length~$T$. The set of available actions for player $i$ is denoted by $A^{i} := \{1,...,K^i\}$ for some~$K^i \in \mathbb{N}$, and we assume that the action set remains fixed throughout the horizon. We denote by~$A := \prod_{i = 1}^{N} A^{i}$ the (multi-agent) outcome space, and endow it with the lexicographical order~$(<)$. For each time $t$, we denote by $u^i_t: A \rightarrow [-M,M]$ player $i$'s payoff function, where $M$ is some positive real number.  
We denote by $\Delta(A)$ the set of all joint distributions on $A$. Often, with some abuse of notation, we identify joint distributions with vectors, that is, $q =~(q_1,\ldots, q_{|A|})$, such that~$q_1$ is the probability mass on the first outcome, $q_2$ on the second one, and so on until $q_{|A|}$. Distances between distributions are induced by an arbitrary $p$-norm, where $p \geq 1$ or~$p =~\infty$. 

For any joint distribution $q \in \Delta(A)$ and subset of distributions~$\mathcal{D} \subseteq \Delta(A)$, we denote the distance between $q$ and the set $\mathcal{D}$ by $\mathrm{d}_p\left(q,\mathcal{D} \right) := \inf_{\tilde{q} \in \mathcal{D}} \|q - \tilde{q} \|_p $. As customary in the literature, we write $a^{-i}$ to denote the profile of actions of all players except~$i$, that is, $a^{-i} =~(a^{j} )_{j \neq i}$, and, to avoid trivialities, we assume that each player has at least two actions, that is, $K^{i} \geq 2$ for all $i$. For any integers $a \leq b$, we use the notation $\llbracket a,b \rrbracket := \{a,\ldots,b\}$ and $\llbracket a \rrbracket := \llbracket 1 , a \rrbracket$.

As both the number of players $N$ and the outcome space $A$ are fixed, we identify the $t$-th stage game $\Gamma_t$ as a vector of payoff functions, that is, $\Gamma_t := (u^i_t)_{i=1}^N$, where~$u^{i}_t$ is the payoff function for player $i$ at time $t$. The corresponding sequence of stage games is denoted by $\mathcal{G}_T := (\Gamma_t : t \in \llbracket T \rrbracket)$, and we write $\mathcal{U}_T^i := (u^i_t : t \in \llbracket T \rrbracket)$ to refer to the sequence of player $i$'s payoff functions. 

We allow the underlying stage game to change an arbitrary number of times throughout the horizon, and track its cumulative temporal variation as follows: 
\[ 
V(\mathcal{G}_T) := \sum_{t=1}^{T-1} \mathbf{1}(\Gamma_t \neq \Gamma_{t+1}),
\]
where $\Gamma_t \neq \Gamma_{t+1}$ if and only if $\| u^{i}_{t} - u^{i}_{t+1} \|_{\infty} > 0$ for some player $i$; 
that is, whenever the payoff of a player changes between two consecutive periods, we count it as a change of the stage game. When clear from the context, we replace $V(\mathcal{G}_T)$ with the simpler notation~$V_T$. Notably, the traditional setting of \emph{repeated} games is retrieved whenever $V_T = 0$. To deal with asymptotic results, we define~$\mathcal{G} := (\mathcal{G}_T : T \in \mathbb{N}) $, and, for each player $i$, $\mathcal{U}^i := (\mathcal{U}_T^i : T \in \mathbb{N})$.

To ease exposition, we assume that the horizon length~$T$ is known to the players (standard techniques such as the doubling-trick can be employed to obtain regret bounds even when $T$ is unknown; see \S\ref{subsect: dynamic benchmark consistency} for the definition of regret). In each period $t$, players independently select actions, and then observe only their own payoff (often referred to as \emph{bandit feedback}). Players do not observe the actions chosen by the other players, the number and the identity of the players who participate in the game, nor the payoff they would have obtained by playing differently. Let
\[ 
\mathcal{P}_{T}^i := \big\{f: \cup_{t=1}^{T-1}\left(A^i \times [-M,M]\right)^{t-1} \to \Delta(A^i)\big\}.
\]
\noindent
be the set of (non-anticipating) functions that prescribe at each time $t \in \llbracket T \rrbracket$ a probability distribution over $A^i$ that depends only on the previous (up to $t-1$) sequence of action-payoff~pairs. A policy for player~$i$, denoted by $\pi^i$, is a sequence of functions~$\pi_T^i$, such that $\pi_T^i$ is an element of~$\mathcal{P}_{T}^i$. 

We denote by~$\mathcal{P}^i$ the collection of all player $i$'s policies, and we similarly denote by $\mathcal{P} := \prod_{i=1}^{N} \mathcal{P}^i$ the space of all profiles of policies (one for each player). Elements of $\mathcal{P}$ are denoted by $\pi$. Let $\delta_{\bar{a}}$ denoting the Dirac measure at $\bar{a}$. Then, the (expected) empirical distribution of play induced by a profile of policies $\pi$ over the horizon~$\llbracket T \rrbracket$ is defined by 
\[ 
\bar\delta_T^{\pi} := \mathbb{E}^{\pi}\left[\frac{1}{T} \sum_{t = 1}^{T} \delta_{a^{t}}\right] \in \Delta(A).
\]

\subsection{Dynamic benchmark consistency}\label{subsect: dynamic benchmark consistency}

We measure the performance of a policy relative to dynamic sequences of actions selected in hindsight. For a non-negative number $C_T$, we denote the set of player $i$'s possible action sequences of length~$T$ with at most $C_T$ action changes~by
\[
\mathcal{S}^i(C_T) := \left \{ (x_t)_{t=1}^{T} \in (A^i)^{T} \;: \; \sum_{t=1}^{T-1} \mathbf{1}(x_{t+1} \neq x_t) \leq C_T \right\}. 
\]
Then, for a sequence of outcomes, $(a_t: t \in \llbracket T \rrbracket)$, the regret incurred by player $i$ relative to a comparative sequence $(x_t: t \in \llbracket T \rrbracket)$ taken from the set $\mathcal{S}^i(C_T)$ is defined as the difference between the performance generated by the actions $(a_t: t \in \llbracket T \rrbracket)$ and the performance generated by the modified sequence of actions~$((x_t,a_{t}^{-i}): t \in \llbracket T \rrbracket)$, that is,
\[
\mathrm{Reg}^i(\mathcal{U}_T^i,(x_t),(a_t)) :=  \sum_{t=1}^{T}  u^i_t(x_t,a_{t}^{-i}) - u^i_t(a_{t}^{i},a_{t}^{-i}). \vspace{0.05cm}
\]

Let $C := (C_T : T \geq 1)$ be a non-decreasing sequence.\footnote{It is without loss of generality to assume that $C$ is non-decreasing, because for each $T \geq 1$, $C_T$ is only an upper bound on the possible number of changes that are allowed in the benchmark over $\llbracket T \rrbracket$.} Then, dynamic benchmark consistency requires a policy to guarantee, as the horizon grows large, the performance of the best dynamic sequence of actions that can be selected in hindsight with number of changes constrained by $C$. 

\begin{definition}[Dynamic benchmark consistency]\label{V-Consistency Definition} 
A policy $\pi^i \in \mathcal{P}^i$ is dynamic benchmark consistent with respect to the sequence~$C$ if there exists a function $g: \mathbb{N} \to \mathbb{R}_{+}$ with $g(T) = o(T)$,\footnote{We can assume without loss of generality that $g$ is non-decreasing and that $g(T) \leq M T$, for all~$T \in \mathbb{N}$.} such that for all~$T \geq 1$, and for all (potentially correlated) strategies of the other players $\sigma^{-i}: \cup_{t=1}^{T} A^{t-1} \to \Delta(A^{-i})$, the regret of policy $\pi^i$ is bounded by
\[ \max_{(x_t) \in \mathcal{S}^i(C_T)} \mathbb{E}^{\pi^i,\sigma^{-i}}\left[ \mathrm{Reg}^i(\mathcal{U}_T^i,(x_t),(a_t)) \right] \leq g(T),\]
\vskip 0.01cm
{\setlength{\parindent}{0cm}
where $ \mathbb{E}^{\pi^i,\sigma^{-i}}$ denotes the expectation with respect to the probability measure over $A^{T}$ induced by the interaction of the policy $\pi^i$ (specified over a horizon of length $T$) and $\sigma^{-i}$.\footnote{Up to polylogarithmic factors, typical regret bounds for dynamic benchmark consistent policies are of the form~$g(T) = T^{\alpha} (C_T)^{1-\alpha}$, for $\alpha \geq 1/2$.
}
}
\end{definition}

\noindent
We denote by~$\mathcal{P}_{DB}^i(C)$ the set of all player $i$'s dynamic benchmark consistent policies subject to~$C$.

The sequence $C$ determines the strength of the benchmark by limiting the number of potential action changes in the comparative sequences used for evaluating policies. Crucially, $C$ places \emph{no restrictions on the policy itself}, on the actual actions selected by any player, or on the frequency at which the stage game changes. If a policy~$\pi^i$ satisfies Definition \ref{V-Consistency Definition} for all $\mathcal{U}^i$, we say that this policy is dynamic benchmark consistent for all payoff functions. To ease exposition, we use as short-hand DB($C$) to refer to a dynamic benchmark subject to the sequence~$C$. We denote by $\mathcal{P}_{DB}^N(C) := \prod_{i = 1}^{N} \mathcal{P}_{DB}^i(C)$ the set of all possible profiles of dynamic benchmark consistent policies.

The traditional notion of Hannan consistency is retrieved when $C$ is the constant sequence at~zero. Thus, dynamic benchmark consistency refines the concept of no-regret relative to static actions, by making the performance requirement that a policy needs to satisfy more~stringent. 
Nevertheless, in Appendix~\ref{structure of DB policies}, we show that when $C$ scales sub-linearly with $T$, i.e., $C_T = o(T)$, restarting at a suitable rate \emph{any} no-regret policy suffices to achieve DB($C$)-consistency. This mapping between traditional no-regret algorithms and dynamic benchmark consistent policies establishes that whenever $C_T = o(T)$, the set of dynamic benchmark consistent policies $\mathcal{P}^{i}_{\textnormal{DB}}(C)$ is in fact ``large'', containing as many policies as the classical no-regret class. A broader discussion about the algorithmic structure and the abundance of these policies is provided in Appendix~\ref{structure of DB policies}.

\subsection{Equilibria and tracking error}

Given the dynamic nature of the sequence of games faced by the $N$ players, it is not clear how to define meaningful equilibrium notions, especially considering the limited information available to the agents. For instance, standard notions of equilibrium (e.g., Nash, and correlated or coarse correlated equilibria), remain meaningful only if the sequence of games eventually stabilizes. 

We propose extending static equilibrium notions to time-varying settings by introducing the following \emph{tracking error}. For a sequence of games $\mathcal{G}_T$, define recursively $t_0 := 0$ and, for $k \in \llbracket V_T+1 \rrbracket$, $t_k := \max\{ t \in \llbracket T \rrbracket : \Gamma_t = \Gamma_{t_{k-1}+1}\}$. Then, by construction, over $\llbracket t_{k-1}+1,t_{k} \rrbracket$, the stage game does not change. To ease notation, we denote $\mathcal{T}_{(k)} := \llbracket t_{k-1}+1,t_{k} \rrbracket$, and refer to this interval as the $k$-th batch.

\begin{definition}[Tracking error]\label{tracking error def}
For each batch $\mathcal{T}_{(k)}$, let~$\mathcal{E}_{(k)} \subseteq \Delta(A)$ be some set of (potentially correlated) equilibria of interest. Then,
the tracking error of a profile of policies $\pi \in \mathcal{P}$ is
\[
err_p(\mathcal{G}_T,\pi,\mathcal{E}) := \sum_{k =1}^{V(\mathcal{G}_T)+1} |\mathcal{T}_{(k)}|\textnormal{d}_p(\bar{\delta}_{(k)}^{\pi},\mathcal{E}_{(k)}),
\]
where $\bar{\delta}_{(k)}^{\pi} := \mathbb{E}^{\pi}\left[\frac{1}{|\mathcal{T}_{(k)}|} \sum_{t \in \mathcal{T}_{(k)}} \delta_{a_t} \right]$ denotes the restriction of the empirical distribution of play to the $k$-th batch and $\mathcal{E} := (\mathcal{E}_{(k)} : k \in \llbracket V(\mathcal{G}_T) + 1 \rrbracket)$ is the sequence of equilibria outcomes of interest.
\end{definition}

The essence of Definition \ref{tracking error def} is the following: to assess how well the empirical distribution of play tracks the sequence of evolving equilibria of interest, we sum, over the batches where the stage game remains constant, the distance between the empirical distribution of play within each batch $\bar{\delta}^{\pi}_{(k)}$ and the corresponding equilibrium set $\mathcal{E}_{(k)}$, weighted by the length of that batch. If for some profile of policies $\pi$, the tracking error is sub-linear, then, on average, the distribution of play induced by $\pi$ tracks the evolving sequence of equilibrium sets $\mathcal{E}$. In particular, if the game is constant throughout the horizon, that is, $V_T = 0$, then the tracking error boils down to the distance of the (unrestricted) empirical distribution of play $\bar{\delta}^{\pi}_T$ from the set of equilibria $\mathcal{E}_{(1)}$. 

The analytical approach we develop is to first focus on the case of repeated games with $V_T=0$, and provide a tight characterization for the joint distributions of play that emerge by the interaction of dynamic benchmark consistent policies as a static equilibrium concept. We next include this equilibrium concept in our notion of tracking error, and move to consider non-stationary systems ($V_T \geq 1$). In particular, we delineate when the distribution of play can exhibit a diminishing tracking error relative to the sequence of equilibria that emerge when $V_T\geq 1$. 

\subsection{Discussion and extensions}


\paragraph{Random payoffs.}{External randomness in the payoff functions can be readily accommodated within our framework and all our results extend immediately. In many applications of interest (e.g., pricing or auctions) it is natural to define the payoffs by \( u^i(a, \epsilon) \), where $\epsilon$ is a random component drawn from some distribution $F$, that is, \( \epsilon \sim F \). For example, in auction settings, bidder valuations may be random, but the distribution of valuations typically remains stable over multiple auctions, and it changes only following more major shifts in market conditions. 

In such cases, we define \( u^i_t(a) := u^i(a, \epsilon_t) \),~where~\( \epsilon_t \sim F_t \), allowing the distribution \( F_t \) to evolve over time. One may then re-define $V_T$ as\vspace{-0.2cm}
\[
V_T := \sum_{t=1}^{T-1} \mathbf{1}(F_{t+1} \neq F_t).\vspace{-0.2cm}
\]
Thus, only changes in the underlying distribution \( F_t \) contribute to the variation budget, even though the realized stage game may fluctuate every period due to randomness in \( \epsilon_t \). All our results extend to this setting by incorporating the additional sources of randomness into the corresponding expectations,\footnote{The results of \S\ref{sect: dbc and tracking error} hold under the repeated game setting, which is retrieved by setting $F_t = F$, for all $t \in [T]$ with $F$ a degenerate distribution.} where expectations of regret and the distribution of play would be taken with respect to \( \pi_i \), \( \sigma_{-i} \), and the distributions \( (F_1, \ldots, F_T) \).

\vspace{-0.5cm}
\paragraph{Piece-wise stationary measures.}{



The variation and tracking-error definitions in our formulation adopt a piecewise-stationary view of the environment, where the horizon is partitioned into discrete phases during which the underlying payoff functions are effectively constant. 
While small period-to-period fluctuations are common in practice, in many markets the underlying game is nearly constant for extended stretches, with major changes occurring only infrequently, which makes a piecewise-stationary perspective natural. For example, an ad exchange may run hundreds of auctions per hour while the estimated distribution of bidders’ valuations is updated only occasionally. Assuming that the stage game remains approximately unchanged for many plays and experiences major shifts only at discrete points allows tractability while capturing first-order dynamics that are critical to the analysis of such complex systems, and underlies the development of equilibrium analyses and their successful applications in this domain \citep{aggarwal2024no}. 


We note that while smoothly-varying non-stationary models (e.g., \citealt{besbes2015}) provide a feasible path for extending our analysis, it is less clear how to define a meaningful tracking-error notion under such variation types. 

}

}

\section{Characterizing joint distributions of play and their tracking error}\label{sect: dbc and tracking error}

\subsection{Repeated games}\label{Dynamic Hannan Set Section}

We begin by analyzing the classical setting of repeated games where the stage game remains fixed and therefore $V_T = 0$. In this case, from the players' perspective, non-stationarity can arise only because of the changes in the actions of the other agents. We study the outcomes that can emerge when players use dynamic benchmark consistent policies, and provide an asymptotic characterization of the distributions of play that can emerge. This allows us to identify the static equilibrium concept that we are going to include in the tracking error notion. The results of this subsection (and their corollary in \S\ref{PoA Section}) assume for simplicity that the payoffs~$u^i$ have rational values, that is~$u^i(a) \in [-M,M] \cap \mathbb{Q}$, for all $i$ and $a$.

Differently from traditional no-regret policies, for which a characterization of the empirical distribution of play in repeated games follows directly from their definition, the dynamic nature of the benchmarks considered in this work makes such characterization non-trivial. Indeed, dynamic benchmark consistency requires that, as the horizon length $T$ grows large, for any possible partition of the horizon $\llbracket T \rrbracket$ (with cardinality at most $C_T+1$), the achieved performance is at least as good as the average performance achieved by best replying to the empirical distribution of play over \emph{each interval} of that partition. Therefore, rather than pursuing a direct approach, we rely on a two-steps procedure, by first studying how the number of action changes allowed in the benchmark at hand affects the behavior of the distribution of play. 

As dynamic benchmark consistent policies also satisfy the (static) no-regret criterion, an empirical distribution of play can emerge only if it is also possible under traditional no-regret policies. Nevertheless, a \emph{tight} characterization of the distributions of play that can be generated by dynamic benchmark consistent policies requires identifying which of the distributions that are possible under no-regret policies are ruled out by the stricter benchmark considered in this work. To this end, consider any pair of sub-linear sequences $\tilde{C}$ and~$C$. The following theorem establishes that, as the horizon length grows large, \emph{any} distribution of play, generated by the interplay of policies that are dynamic benchmark consistent subject to~$\tilde{C}$, can be approximated by the  distribution of play induced by the interplay of policies that are dynamic benchmark consistent subject to $C$.

\begin{proposition}[Approximation of joint distribution of play]\label{Direction 1}
Let $C$ and $\tilde{C}$ be two non-decreasing sub-linear sequences, and $\tilde{\pi}  \in \mathcal{P}_{DB}^N(\tilde{C})$ be a profile of~DB($\tilde{C}$) consistent policies. Take any convergent sub-sequence of the joint distribution of play~$\bar\delta_T^{\tilde{\pi}}$ and denote its limit~by $\tilde{q} \in \Delta(A)$. Then, for any~$\epsilon >0$, there exists a profile of DB($C$) consistent policies~$\pi \in \mathcal{P}_{DB}^N(C)$ that generates a convergent distribution of play~$\bar\delta_T^{\pi}$ whose limit $q \in \Delta(A)$ is $\epsilon$-close in the $p$-norm to $\tilde{q}$, that is, $\| q - \tilde{q} \|_p \leq \epsilon$.\footnote{To simplify exposition, the statement of Proposition \ref{Direction 1} assumes that all the players' policies are dynamic benchmark consistent with respect to the same sequence of upper bounds~$C$. However, the proof of Proposition~\ref{Direction 1} shows that our result holds also when players deploy different sequences of upper bounds.}
\end{proposition}

\noindent \textbf{Main ideas of the proof.}{ While the full proof of the proposition appears in Appendix \ref{Dynamic Hannan Set Section}, we next outline its main ideas. Let $C$ and $\tilde{C}$ be sub-linear sequences. Fix a profile $\tilde{\pi}$ of DB($\tilde{C}$)-consistent policies, and let $\tilde{q} \in \Delta(A)$ be a limit point of $\bar{\delta}_T^{\tilde{\pi}}$, the empirical joint distribution of play under $\tilde{\pi}$. Assume, for simplicity, that $\tilde{q}$'s support contains only two outcomes, $a_{(1)}$ and $a_{(2)}$, with rational probabilities, i.e., $\tilde{q}(a_{(1)}) = m_1/m$ and $\tilde{q}(a_{(2)}) = m_2/m$, for integers $m_1, m_2$ and $m = m_1 + m_2$.

Fix $T \geq m$. For each player we construct a trigger-type policy $\pi^i$, which prescribes a cooperative schedule unless a deviation is detected, in which case it switches to an underlying DB($C$)-consistent policy. For each player $i$, the cooperative schedule is designed to cycle through $(a_{(1)}^i, a_{(2)}^i)$ based on $\tilde{q}$'s probabilities, in repeating cycles of $m_1$ periods of $a_{(1)}^i$ followed by $m_2$ periods of $a_{(2)}^i$. The detection test identifies deviations from this schedule, while the underlying DB($C$)-consistent policy ensures that if a deviation is detected, we can still guarantee sub-linear regret relative to a dynamic benchmark subject to the sequence~$C$. On the other hand, if no deviation is detected, we show that the cooperative schedule is followed by all players, and establish that the regret incurred by each player relative to a dynamic benchmark subject to the sequence~$C$ must be sub-linear too. Combining the analysis for these two cases establishes that $\pi^i$ is DB($C$)-consistent. Finally, when jointly deployed, the distribution of play generated by $\pi = (\pi^i)_i$ converges to $\tilde{q}$. Notably, the $\epsilon$ in the statement of the theorem arises when $\tilde{q}$ has irrational probability masses; In such cases, $\tilde{q}$ is approximated by a distribution with rational coordinates. \hfill $\blacksquare$}

\bigskip
\noindent The possibility to approximate the distribution of play generated by any profile of dynamic benchmark consistent policies via no-regret policies (that is, with $C \equiv 0$) allows one to identify the static equilibrium concept associated with the notion of dynamic benchmark consistency when the stage game is fixed. When all players deploy no regret policies, the empirical joint distribution of play converges to the set of coarse correlated (or Hannan) equilibria of the underlying stage game \citep{HMSC3}. This set, often referred to as the Hannan set, is defined as follows:

\begin{equation}\label{Hannan Set Definition}
    \mathcal{H}(\Gamma) := \Bigg \{ q \in \Delta(A) \; : \; \forall i \in \llbracket N \rrbracket, \; \forall x^{i} \in A^i, \underset{a \sim q}{\mathbb{E}} \left[ u^i(x^i,a^{-i}) - u^i(a) \right] \leq 0 \Bigg \}.
\end{equation}

\noindent
The natural interpretation of $\mathcal{H}(\Gamma)$ requires the existence of a mediator who samples profiles of actions~$a$ from a joint distribution $q $, and then, privately, recommends to each player~$i$ to play~$a^i$. A distribution $q$ is a coarse correlated equilibrium if any ex-ante deviation by a player is, on average, not profitable when the other players stick to their recommendations.\footnote{
When the right-hand-side of the inequality in \eqref{Hannan Set Definition} is replaced by~$\epsilon$, we refer to the corresponding equilibria as $\epsilon$-approximated coarse correlated equilibria.} Building on Proposition~\ref{Direction 1}, we can now state the main result of \S \ref{Dynamic Hannan Set Section}: any equilibrium in $\mathcal{H}(\Gamma)$ can emerge as a distribution of play induced by some profile of dynamic benchmark consistent policies. 

\begin{theorem}[Joint distributions of play under dynamic benchmark consistency]\label{Hannan Set Survives}
Let $C$ be any non-decreasing sub-linear sequence. The Hannan set $\mathcal{H}(\Gamma)$ is the smallest closed set of distributions that can be reached by any profile of DB($C$) consistent policies; i.e., if a closed set of distributions $\mathcal{D} \subseteq \Delta(A)$ is such that, for all~$\pi  \in  \mathcal{P}^{N}_{DB}(C)$,  $ \mathrm{d}_p\left( \bar\delta^{\pi}_T,\mathcal{D}\right) \rightarrow 0 $, as $T \rightarrow \infty$, then it must be the case that $\mathcal{D} \supseteq \mathcal{H}(\Gamma)$. 

\end{theorem}
\noindent 
The proof of this result appears in Appendix \S \ref{Proofs of Dynamic Hannan Set Section}. Theorem~\ref{Hannan Set Survives} establishes that the class of equilibria associated with DB-consistency remains the set of coarse correlated equilibria, which can also emerge under traditional no-regret policies.~We next turn to analyze time-varying games (i.e.,$V_T \geq 1$), and apply our notion of tracking error to sequences of coarse correlated equilibria.

\subsection{Time-varying games}\label{sect: time varying games}

When $V_T = 0$, changes in the payoff functions arise only because of the evolving action choices of the other players. On the other hand, when $V_T \geq 1$, also exogenous changes in the stage game can appear, such as shifts in customers' price sensitivity (e.g., Example~\ref{pricing game example}). Having identified coarse correlated equilibria as the reference equilibrium class under dynamic benchmark consistency, we turn to study how well the players' empirical distribution of play generated by dynamic benchmark consistent policies can track the \emph{sequence} of evolving sets of coarse correlated equilibria. Therefore, throughout this section, when using the term ``tracking error'', we refer to the tracking error with respect to the coarse correlated equilibria sets. To ease exposition of results in this subsection (and their corollaries in \S \ref{sect Welfare in time-varying games}), we assume that stage games are selected from an arbitrarily large but finite set of possible games; We defer discussion of weaker conditions to Appendix~\ref{proofs of sect: sect time varying games}.

\subsubsection{Single best-reply games}\label{subsect: single best-reply games}

Turning to analyze the case of $V_T \geq 1$, we begin by focusing on the important setting in which, for any game in $\mathcal{G}$, each player $i$ has a single best-reply to all the actions of the other players~$a^{-i}$.\footnote{Formally, we say that a payoff function $u^i_t$ has the single best-reply property if there exists an action $a^i \in A^i$ such that $a^i \in \argmax \{ u^i_t(x,a^{-i}) : x \in A^i\}$, for all $a^{-i} \in A^{-i}$. A stage game $\Gamma_t$ is termed single best-reply, if the payoff function~$u^i_t$ has the single best-reply property for all $i \in \llbracket N \rrbracket$.} We refer to these games as single best-reply games. Such a condition is met, for example, when the game has a Nash equilibrium in (weakly) dominant actions, a scenario that often emerges in real-world applications of algorithmic decision-making, including common price competitions (e.g., Example \ref{pricing game example}), and games induced by truthful mechanisms (e.g., second-price auctions). 
Notably, since each stage game has a singe best-reply, a player's preferred action is not affected by the actions selected by the competitors. Instead, only the exogenous changes in the stage game itself can alter the identity of a player's best-reply. In this sense, along the spectrum of possible non-stationarities a player can face, time-varying single best-reply games lie on the opposite side compared to general-sum repeated games.

We denote by $\Gamma^{\textnormal{sbr}}$ a game in which each player has a single best-reply to the actions of the other players, and we similarly write $\mathcal{G}^{\textnormal{sbr}}_T$ to refer to a sequences of stage games with the single best-reply property. Importantly, in these games, as long as $C_T\geq V_T$, the dynamic benchmark can always choose the sequence of best-replying actions, and therefore it becomes the \emph{optimal} dynamic benchmark that can be selected with the benefit of hindsight. We next provide necessary and sufficient conditions to obtain diminishing tracking error in single best-reply games.\footnote{With some abuse of notation, we write $err(\mathcal{G}_T,\pi,\mathcal{H})$ to denote the tracking error of a profile of policies $\pi$, relative to the sequence of equilibrium sets $\mathcal{E}_{(k)} = \mathcal{H}(\Gamma_{(k)})$ for each stage game~$\Gamma_{(k)}$ along the sequence of games $\mathcal{G}_T$.}

\begin{theorem}[Sufficient conditions for a diminishing tracking error]\label{suff result}
Let $\mathcal{G}^{\textnormal{sbr}}$ be a sequence of single best-reply games. If each player $i$ employs a dynamic benchmark consistent policy $\pi^i$ with $C^i_T \geq V_T$ for all $T$, and $C^i_T = o(T)$, then the tracking error is sub-linear, that is, $err_{p}(\mathcal{G}^{\textnormal{sbr}}_T, \pi, \mathcal{H}) = o(T)$.
\end{theorem}
 
The proof of Theorem \ref{suff result} appears in Appendix~\S\ref{proof of single-best reply games}. This result implies that, when the variation in the single best-reply games is bounded by a slowly changing sequences of action changes, deploying dynamic benchmark consistent policies is sufficient to achieve sub-linear tracking error. Therefore, the joint distribution of play that can emerge approximates the sequence of Hannan sets (if a game is a strictly dominant action game, then the Hannan set reduces to a unique Nash equilibrium; that is, $\mathcal{H}(\Gamma) = \{\delta_{a^{*}}\}$, where $a^{*}$ is the Nash equilibrium). 

Conversely, we next establish that achieving dynamic benchmark consistency with respect to benchmark sequences that have a number of changes at least as large as the variation in the stage games is also \emph{necessary} to obtain a diminishing tracking error. 

\begin{theorem}[Necessary conditions for a diminishing tracking error]\label{imp result}
Let $\mathcal{U}^i$ be single best-replying payoff functions for player $i$ such that $S_T := \sum_{t=1}^{T-1} \mathbf{1}(u^i_{t+1} \neq u^i_t )$ is sub-linear in $T$, that is, $S_T = o(T)$. If player~$i$'s policy $\pi^i$ is not $\textnormal{DB}$-consistent relative to the sequence $(S_T)$, then there exists a sequence of single best-reply games $\mathcal{G}^{\textnormal{sbr}}$ with variation $V(\mathcal{G}^{\textnormal{sbr}}_T) = S_T$ for all $T$, and $\textnormal{DB}$-consistent policies~$\pi^j$ (for all players $j \neq i$) relative to the optimal dynamic benchmark in hindsight (i.e., $C_T = T-1$) that generate linear tracking error, that is, $err_p(\mathcal{G}^{\textnormal{sbr}}, \pi,\mathcal{H}) = \Theta(T)$. 
\end{theorem}

\noindent\textbf{Main ideas of the proof.} While the complete proof of the theorem appears in Appendix~\ref{proof of single-best reply games}, we next outline its main ideas. 
Let $T \in \mathbb{N}$, and suppose that $\pi^i$ is not $\textnormal{DB}$-consistent with respect to the sequence ($S_T$). As a preliminary step, we show that a linear lower bound on the regret incurred by player $i$
can be established by restricting attention, without loss of generality, to deterministic competitors' strategies.
Fix a profile of deterministic strategies~$\sigma^{-i} = (\sigma^j)_{j \neq i}$ maximizing $i$'s regret.

We construct a sequence of single-best reply games $\mathcal{G}_{T}^{\textnormal{sbr}}$ with the following properties: (i) The payoff functions of player $i$ coincide with $\mathcal{U}_T^i$; (ii) For each competitor $j \neq i$, the payoff function is constant, that is, $u^j_t = u^j$; (iii) For each competitor $j \neq i$, the payoff function is fully determined by the actions of the other players, that is, for all $(x,y) \in A^j \times A^j$, one has $u^j(x,a^{-j}) = u^j(y,a^{-j})$ for all~$a^{-j} \in A^{-j}$; and (iv) For each competitor $j \neq i$ and action $x \in A^j$, the payoff $u^j$ is injective in the other players' actions, that is, for all $a^{-j} \neq \hat{a}^{-j}$ one has $u^j(x,a^{-j}) \neq u^j(x,\hat{a}^{-j})$.

Properties (i) and (ii) ensure that $V(\mathcal{G}^{\textnormal{sbr}}_T) = S_T$ for all $T$, and, together with (iii), that the stage games are single best-reply. Further, properties (ii) and (iii) imply that, for each player $j \neq i$, any policy satisfies dynamic benchmark consistency, even relative to the \emph{optimal} benchmark that can be selected in hindsight. Importantly, property (iv) allows each agent to identify the actions selected by the other players. In particular, each player $j \neq i$ can mimic the actions selected by $\sigma^{j}$ via a policy~$\pi^j$, which does not have access to the previous history of outcomes. To conclude the proof, we observe that because the stage games are all single-best reply, the actions selected in hindsight by player $i$'s benchmark coincide with the sequence of her best-replying actions. Given this, we show that the tracking error generated by $\pi$ turns out to be linear in player $i$'s regret and in $T$. \hfill $\blacksquare$

\bigskip
\noindent Theorem \ref{imp result} has two important implications. First, it identifies the necessary conditions that policies must meet to guarantee diminishing tracking error in single best-reply games. Second, it demonstrates that if player $i$ does not employ a DB-consistent policy relative to $(V(\mathcal{G}_T))$, we cannot expect diminishing tracking error, even if we allow the other players $-i$ to deploy~DB-consistent policies against sequences $C^j$ with $C^j_T > V_T$, for all $T$ and $j \neq i$.

Together, these findings imply a separation between traditional no-regret and dynamic benchmark consistent policies that is critical from a game theoretic standpoint: performance guarantees \emph{only} relative to static actions can lead to large tracking error, including simple and practical settings. While the construction behind Theorem \ref{imp result} relies on a multi-agent system where only player~$i$ is facing a non-trivial decision problem, we demonstrate in Appendix~\ref{Simulations for Example 2} that the same conclusions can arise also in practical settings, including the one described in Example~\ref{pricing game example}. 

\subsubsection{General-sum games}\label{subsect: general-sum games}

Moving beyond single best-reply games, we continue to analyze the interaction of dynamic benchmark consistent policies in general-sum games where non-stationary can arise from changes in competitors' actions as well as variation in the stage game itself. The main result of this section, Theorem \ref{thm: restarting = sublinear track error}, establishes that under mild conditions, the empirical distribution of play that emerge under the interaction of dynamic benchmark consistent policies exhibits diminishing tracking error with respect to the evolving sequence of Hannan sets.

\begin{theorem}[Tracking the sequence of coarse correlated equilibria]\label{thm: restarting = sublinear track error}
Let $\mathcal{G}$ be a sequence of general-sum games with variation $V(\mathcal{G}_T) = o(T)$. Let $\pi \in \mathcal{P}$ be a profile of policies such that, for each player~$i$
\begin{equation}\label{clipped dbc}
      \sum_{k=1}^{V(\mathcal{G}_T) + 1} \max_{x \in A^i} \mathbb{E}^{\pi} \left[ \sum_{t \in \mathcal{T}_{(k)}}  \left(u_{(k)}^{i}(x,a^{-i}_t) - u^{i}_{(k)}(a_t^{-i}) \right)\right]_{+} = o(T),
\end{equation}
\noindent
Then, the tracking error generated by $\pi$ is sub-linear, that is, $err_p(\mathcal{G}_T,\pi,\mathcal{H}) = o(T)$. 
\end{theorem}
\noindent
Notably, Condition \ref{clipped dbc} is slightly stronger than dynamic benchmark consistency, as, if for some batch $\mathcal{T}_{(k)}$, the regret incurred by $\pi^i$ is negative, its contribution to the overall regret is truncated to zero. In Appendix \ref{dbc and tracking error}, we show that dynamic benchmark consistency alone is not sufficient to guarantee a diminishing tracking error for general-sum games. However, we argue in the proof of Theorem \ref{thm: restarting = sublinear track error} that Condition \ref{clipped dbc} is actually satisfied by natural DB-consistent policies, including:
\begin{enumerate}[label = (\roman*)]
    \item Restarting any Hannan consistent policy every $\Delta^i_T$ periods with $\Delta^i_T = o(T/(V(\mathcal{G}_T) + 1))$, and $\Delta^i_T \to \infty$, as $T \to \infty$ (see Proposition \ref{prop: restarting = sublinear external-internal dynamic regret} in Appendix \ref{structure of DB policies}). 
    \item The Exp3S policy with $\gamma^i_T = o(1)$, $\alpha^i_T/\gamma^i_T = o(1)$ and $\ln(K^i/\alpha^i_T)/\gamma^i_T = o(T/(V(\mathcal{G}_T)+1))$ (see \citealt{auer2002nonstochastic} and Appendix \ref{structure of DB policies}). 
\end{enumerate} 

\noindent \textbf{Main ideas of the proof}. While the full proof of Theorem~\ref{thm: restarting = sublinear track error} appears in Appendix~\ref{proofs time-varying general sum-games}, we next outline its main ideas. The proof builds on a key analytical result (Lemma \ref{lemma: regret-distance} in Appendix \ref{proofs of sect: sect time varying games}), connecting the regret incurred by agents to the distance between the distribution of play and the set of coarse correlated equilibria in repeated games. Specifically, fix a general-sum game~$\Gamma = (u^i)_{i}$. We show that there exists a constant, $\textnormal{const}(\Gamma)$, such that for any~$\epsilon \geq 0$, $ T \in \mathbb{N}$, and a profile of policies $\pi \in \mathcal{P}$, if, for all players $i$, 
\[
\max_{(x^t) \in \mathcal{S}_i(0)} \mathbb{E}^{\pi} \left[ \mathrm{Reg}_i\left( (\underbrace{(u^i,\ldots,u^i)}_{T \text{ times}}), (x^t), (a^t) \right) \right] \leq \epsilon T,
\]
then it follows that $\textnormal{d}_p(\bar{\delta}^{\pi}_{T},\mathcal{H}) \leq \min \{ \textnormal{const}(\Gamma) \epsilon , 2^{1/p}\}$, where a characterization of $\textnormal{const}(\Gamma)$ appears in the proof of Lemma \ref{lemma: regret-distance} and is obtained by bounding, for the stage game $\Gamma$, the Hausdorff distance between its Hannan set and its~$\epsilon$-approximated Hannan set.\footnote{We note that, by definition, if the largest average regret relative to a static benchmark is at most $\epsilon$ for all players, then the distribution of play is an~$\epsilon$-Hannan equilibrium. However, this \emph{does not} imply that the distance of the joint distribution of play from the Hannan set is at most~$\epsilon$. For instance, in Example~\ref{small regret but big distance} of Appendix \ref{proofs of sect: sect time varying games} we show a simple game where for some distributions in the $\epsilon$-approximated Hannan set, the closest distribution in the Hannan set has distance $2^{1/p}$ in the $p$-norm (i.e., the largest distance that can be obtained), even for small values of $\epsilon$. Establishing that the distance of $\bar{\delta}^{\pi}_T$ from $\mathcal{H}$ can be bounded from above by a linear function of $\epsilon$ is a key part of our analysis.}

To conclude the proof we leverage Lemma~\ref{lemma: regret-distance} to connect policies' performance and distance of the empirical distribution of play from the Hannan set in batches where the stage game is constant. In particular, we show that 
\[ err_p(\mathcal{G}_T,\pi,\mathcal{H}) \leq \textnormal{const}(\mathcal{G}_T) \sum_{k=1}^{V(\mathcal{G}_T) + 1} \max_{x \in A^i} \mathbb{E}^{\pi} \left[ \sum_{t \in \mathcal{T}_{(k)}}  \left(u_{(k)}^{i}(x,a^{-i}_t) - u^{i}_{(k)}(a_t^{-i}) \right)\right]_{+},\]
\noindent
where $\textnormal{const}(\mathcal{G}_T) = \max_{t \in \llbracket T \rrbracket} \textnormal{const}(\Gamma_t)$.
\hfill $\blacksquare$

\subsection{Revisiting Example \ref{pricing game example}}

We conclude this section by revisiting the setting of Example \ref{pricing game example} to illustrate the connection between the empirical distribution of play and the sequence of stage game equilibria. We now assume that both sellers are deploying the Exp3S policy \citep{auer2002nonstochastic}, tuned in such a way to satisfy DB($1$)-consistency (see Appendix \ref{structure of DB policies} for more details). 

Figure \ref{fig:comparison_Exp3SMain} shows the $\ell_2$-distance of the emerging distribution of play from the set of coarse correlated (in the setting of Example \ref{pricing game example}, also Nash) equilibria in the first and second half of the horizon together with the resulting (per-period) tracking error, denoted by $err_2(\mathcal{G}_T^{Ex1},(Exp3S,Exp3S),\mathcal{H})/T$.

\begin{figure}[H]
    \centering
\includegraphics[width=0.6\linewidth]{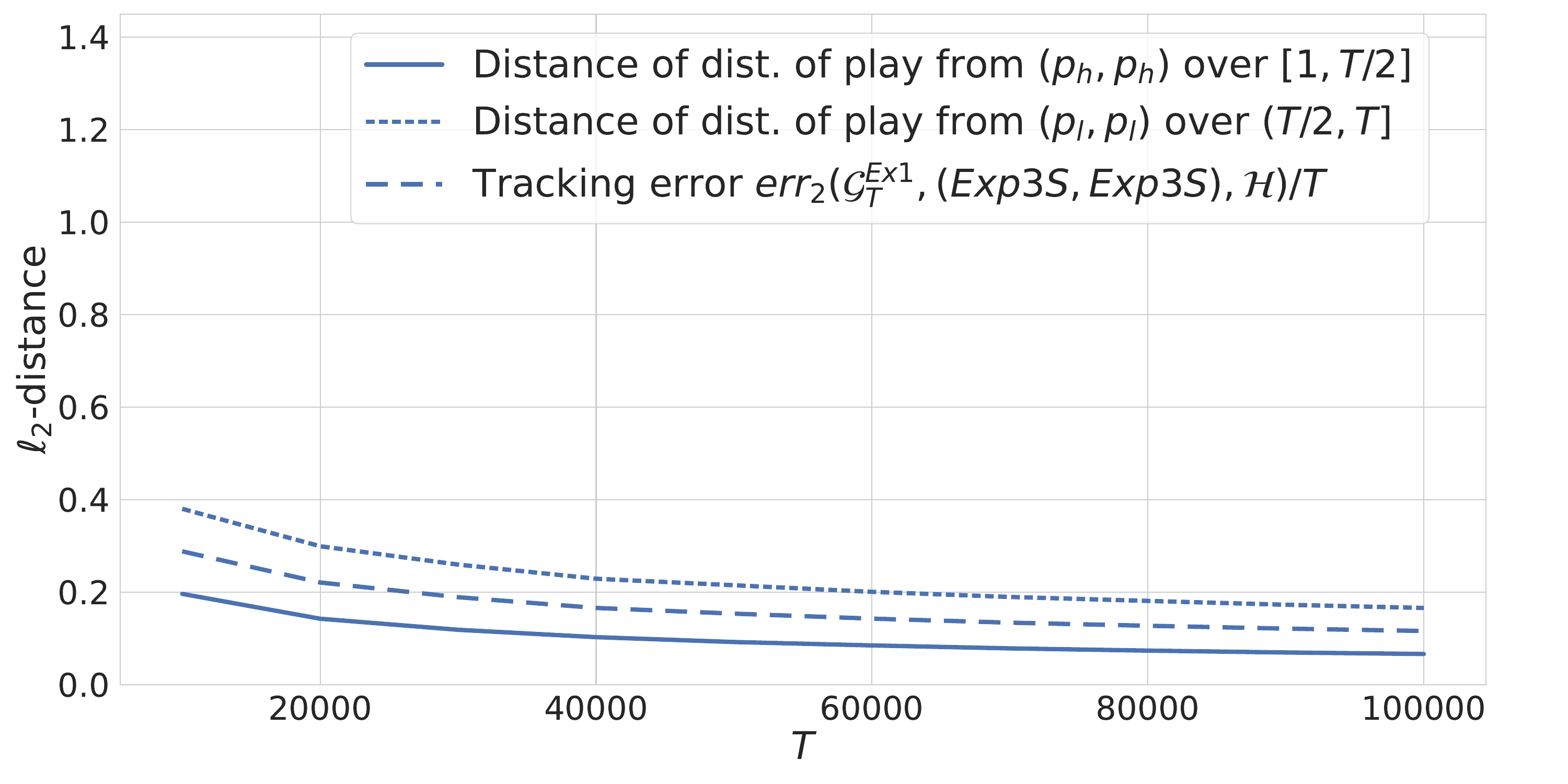}\vspace{-0.2cm}
    \caption{\footnotesize Distance between the empirical joint distribution of play and the coarse correlated equilibrium of each stage game over the first and second halves of the season in Example \ref{pricing game example} together with resulting tracking error. Sellers employ the Exp3S \citep{auer2002nonstochastic} algorithm tuned with exploration parameter $\gamma_T = \sqrt{4 \ln(T)/ T}$ and sharing factor $\alpha_T = 2/T$. The figure shows the~$\ell_2$-distance of the distribution of play in the first and second half from $(p_h,p_h)$ and $(p_l, p_l)$ respectively. The distribution of play is estimated by averaging the results over $750$ simulations. }
    \label{fig:comparison_Exp3SMain}
\end{figure}\vspace{-0.2cm}

Recall that the setting of Example \ref{pricing game example} the stage games are single-best reply with one change in the payoff functions taking place at period $T/2$. In contrast with Figure~\ref{fig:comparison_Exp3Intro}, Figure~\ref{fig:comparison_Exp3SMain} implies that the distribution of play that emerge from the interaction of the Exp3S policies converges asymptotically in each half of the selling horizon to the corresponding coarse correlated (Nash) equilibrium. Thus, in line with Theorem~\ref{thm: restarting = sublinear track error}, the tracking error is sub-linear in $T$, and this error diminishes to zero in a long-run average sense.  

We note that, in general, a diminishing tracking error does not guarantee convergence of the empirical joint distribution of play to the coarse correlated equilibrium over segments of the horizons where the stage game does not change. However, Figure~\ref{fig:comparison_Exp3SMain} illustrates that in the setting of Example \ref{pricing game example} there is an equivalence between diminishing tracking error and convergence in each such segment. Indeed, when all segments of time periods over which the stage game does not change have length that is linear in $T$, that is,~$|\mathcal{T}_{(k)}| = \Theta(T)$ for all $k$, the tracking error is sub-linear if and only if for each batch $\mathcal{T}_{(k)}$, the distribution of play over that batch~$\bar{\delta}^{\pi}_{(k)}$ converges to $\mathcal{H}(\Gamma_{(k)})$.

\noindent
We proceed in the following section to study the welfare implications of our tracking error guarantees, by comparing the social efficiency of the outcomes that can arise under traditional no-regret policies with those that can arise under dynamic benchmark~consistency.
\section{Social efficiency of outcomes}\label{PoA Section}\vspace{-0.1cm}

In this section we demonstrate the applicability of the equilibrium analysis detailed in \S \ref{sect: dbc and tracking error} to study the welfare efficiency of the outcomes resulting from the interaction of dynamic benchmark consistent policies. We establish that, when the stage game is fixed, the worst-case social welfare that can emerge is independent of the number of action changes in the benchmark. In contrast, when the stage game varies over time, we show that it is possible to bound the fraction of the optimal social welfare that can be guaranteed via a non-decreasing function of the number of action changes~$C_T$ allowed in the benchmark. This establishes that dynamic benchmark consistent policies can provide better welfare guarantees relative to traditional no-regret policies in non-stationary systems. Throughout the section, we assume without loss of generality that~the payoff functions are non-negative for each player $i$.

\subsection{Welfare in repeated games}\label{Welfare in repeated games}

We begin by first addressing the case where $V_T=0$, in which the stage game does not change. We focus on the worst-case social welfare that can be generated by analyzing the price of anarchy. This metric is commonly used in the literature for this purpose and it is defined as follows. 
\begin{definition}[Price of anarchy]\label{Def:2}
Let $\Gamma$ be a stage game and $\mathcal{D} \subseteq \Delta(A)$ some distributions of interest. Consider a welfare function $W: A \to \mathbb{R}_{+}$. The price of anarchy (PoA) with respect to the class~$\mathcal{D}$ is the ratio between the optimal centralized social welfare and the worst-case social welfare that can be generated within the set $\mathcal{D}$. Formally,

\[ 
\mathrm{PoA}\left(\Gamma, W, \mathcal{D} \right) := \frac{\underset{\tilde{q} \in \Delta(A)}{\max} \underset{a \sim \tilde{q}}{\mathbb{E}} \left[ W(a) \right]}{\underset{q \in \mathcal{D}}{\inf} \;  \underset{a \sim q}{\mathbb{E}} \left[ W(a) \right]},
\]
where we interpret $w/0 = \infty$ for all non-zero $w$, and $0/0 = 1$. Typical welfare functions of interest are the additive welfare $W(a) = \sum_{i \in \llbracket N \rrbracket} u^i(a)$, and the minimum welfare $W(a) = \min_{i \in \llbracket N \rrbracket} u^i(a).$
\end{definition}
\noindent
While the traditional choice for the class $\mathcal{D}$ is the set of product distributions, whose profile of marginals constitutes a Nash equilibrium, the price of anarchy has been analyzed for other distribution classes, including coarse correlated and correlated equilibria. By applying Theorem~\ref{Hannan Set Survives}, the next corollary establishes that when the stage game remains fixed, as the horizon grows large, the price of anarchy obtained under dynamic benchmark consistent policies \emph{coincides} with the PoA of coarse correlated equilibria, and so it does not scale with the number of action changes allowed in the benchmark at hand. The proof of this result appears in Appendix \S \ref{Proofs of welfare in repeated games}.

\begin{corollary}[PoA equivalence]\label{Cor:1} Let $C$ be a non-decreasing sub-linear sequence, and $W: A \to \mathbb{R}_{+}$ be a welfare function. Then, the following equivalence holds
    \begin{equation}\label{PoTA}
    \underset{\pi = (\pi^i)_{i = 1}^{N} \in \mathcal{P}_{DB}^{N}(C)}{\sup} 
    \left\{\frac{\underset{\tilde{q} \in \Delta(A)}{\max} \; \underset{a \sim \tilde{q}}{\mathbb{E}} \left[ W(a) \right] \;}{\underset{T \rightarrow +\infty}{\liminf} \; \frac{1}{T} \sum_{t=1}^{T} \mathbb{E}^{\pi} \left[ W(a_t) \right] }\right\} = \mathrm{PoA}(\Gamma, W,\mathcal{H}(\Gamma)), 
    \end{equation}
and so, the worst-case social welfare generated by the interaction of DB($C$) consistent policies does not scale with the number of action changes in the benchmark at hand. 
\end{corollary}

The worst-case equivalence in Corollary~\ref{Cor:1} implies that the \emph{tightness} of PoA-bounds provided for settings that consider static benchmarks remains in the present framework. For instance, an extensively used  technique to establish bounds on the price of anarchy is the so-called smoothness analysis proposed by \cite{robustpoa}, and defined as follows for any welfare function $W$.
\begin{definition}[\citealt{robustpoa}]\label{Smoothness definition}
A game $\Gamma$ is said to be $(\lambda,\mu)$-smooth for some $\lambda \geq 0$ and $\mu > -1$, if, for all outcomes $a,a' \in A$,
\begin{equation}\label{smoothness condition}
 \sum_{i=1}^{N} u^i((a')^i,a^{-i}) \geq \lambda  W(a') - \mu W(a).
\end{equation}
\end{definition}

That is, the $(\lambda,\mu)$-smoothness property \eqref{smoothness condition} constrains the cumulative change in the social welfare that can be generated by a profile of unilateral deviations, and, as such, can be exploited to provide bounds on the PoA. In particular, it is known that \eqref{smoothness condition} together with $W(a) \geq \sum_i u^i(a)$, imply that any Hannan equilibrium achieves a fraction of at least $\lambda/(1+\mu)$ of the social optimum. 

Further, there are classes of games for which the lower bound on the social welfare implied by~(\ref{smoothness condition}) was shown to be tight under the interaction of Hannan consistent policies, including simultaneous second-price auctions, valid utility games (\citealt{vetta2002nash}), and congestion games with affine cost functions.\footnote{In cost-minimization games, the payoffs $u^i$ and welfare function $W$ are replaced respectively by cost-function $\textnormal{Cost}^i$ and welfare-cost function $\textnormal{Cost}$, and the definition of smooth-game is adjusted accordingly (see \citealt{robustpoa}). Further, for simultaneous SPA, the employed smoothness condition is indeed a relaxation of Inequality \eqref{smoothness condition}. In particular, it requires the existence of a social welfare optimal outcome $a'$ such that, for all $a$, Inequality \eqref{smoothness condition} holds.} For these games, our results imply that the lower bound derived by (\ref{smoothness condition}) remains tight also under the interaction of dynamic benchmark consistent policies.

\subsection{Welfare in time-varying games}\label{sect Welfare in time-varying games}

We next continue to study the efficiency of outcomes when $V_T\geq 1$ and the stage game can change. Let $\mathcal{W}_T$ be a sequence of welfare functions. We denote $\mathcal{W} := (\mathcal{W}_T : T \in \mathbb{N})$, and assume that $\sup_{T} \max_{W_t \in \mathcal{W}_{T}} \| W_t \|_{p} \leq \bar{W}$. We extend the static price of anarchy notion for coarse correlated equilibria to non-stationary systems with $V_T\geq 1$ in the following natural way: 

\[ \mathrm{dynPoA}(\mathcal{G}, \mathcal{W},\mathcal{H}) := \frac{\underset{T \to \infty}{\limsup} \; \frac{1}{T} \sum_{t=1}^{T} \underset{\tilde{q}_t \in \Delta(A)}{\max} \; \underset{a \sim \tilde{q}_t}{\mathbb{E}} \left[ W_t(a) \right]}{\underset{T \rightarrow \infty}{\liminf} \; \frac{1}{T} \sum_{t=1}^{T} \min_{q_t \in \mathcal{H}(\Gamma_t)} \underset{a \sim q_t}{\mathbb{E}} \left[ W_t(a) \right]}.\]
\noindent
Applying the diminishing tracking error guarantees established in Theorem \ref{thm: restarting = sublinear track error} yields the following corollary, whose proof appears in Appendix \S \ref{Proofs of welfare in time-varying games}.

\begin{corollary}[Dynamic price of anarchy]\label{dyn poa without smoothness}
Let $\pi$ be a profile of dynamic benchmark consistent policies such that, for each player~$i$, $C^i_T \geq V(\mathcal{G}_T)$ for all $T$, and $C^i_T = o(T)$. Then, it holds

\begin{equation}\label{dyn poa simple bound}
    \frac{\underset{T \to \infty}{\limsup} \; \frac{1}{T} \sum_{t=1}^{T} \underset{\tilde{q}_t \in \Delta(A)}{\max} \; \underset{a \sim \tilde{q}_t}{\mathbb{E}} \left[ W_t(a) \right]}{\underset{T \rightarrow \infty}{\liminf} \; \frac{1}{T} \sum_{t=1}^{T} \mathbb{E}^{\pi} \left[ W_t(a_t) \right] } \leq \mathrm{dynPoA}(\mathcal{G},\mathcal{W},\mathcal{H}). 
    \end{equation}

\end{corollary}

By Corollary \ref{dyn poa without smoothness}, when the number of action changes in the benchmark bounds from above the variation in the stage games, that is, $C_T \geq V_T\geq 1$, the dynamic price of anarchy that emerges from the interaction of dynamic benchmark consistent policies can be bounded by the dynamic price of anarchy associated with the sequences of coarse correlated equilibria. When $V_T = 0$, we demonstrated above (Corollary \ref{Cor:1}) that the bound is tight. We further conjecture that \eqref{dyn poa simple bound} remains tight when $V_T \geq 1$, but leave the proof of this conjecture for future research.

Moreover, by adapting standard smoothness arguments, we next show that it is possible to bound the left-hand side of \eqref{dyn poa simple bound} even when the number of action changes is below the variation of the stage games, that is, $C_T <V_T$. In particular, denote by
\[ \beta(C_T) := \frac{\max_{(x_t)_{t=1}^{T} \in A^{T}: \forall i, \; (x_t^i)_{t=1}^{T} \in \mathcal{S}^i(C_T)} \sum_{t=1}^{T} W_t(x_t)}{  \sum_{t=1}^{T} \max_{y_t \in A} W_t(y_t)}\]
\noindent
the largest fraction of the optimal social welfare that can be implemented by a central planner through sequences of outcomes with at most $C_T$ action changes for each player. Observe that, for each $T$, the function $\beta(\cdot)$ is non-decreasing in $C_T$, and that, whenever the variation in the sequence of stage games does not exceed the number of changes allowed in the dynamic benchmark (i.e.,~$C_T \geq V_T$) then $\beta(C_T)$ equals one. However, it is not difficult to find a sequence of games~$\mathcal{G}$, for which it holds $\liminf_{T} \beta(C_T) > 0$ even when~$C_T < V_T$ for all $T$. The following result provides a lower bound for the social welfare generated by dynamic benchmark consistent policies that explicitly depends on the number of allowable action changes,~$C_T$.

\begin{proposition}[Welfare lower bound in smooth games]\label{DynPoA} Let $C$ be a sub-linear sequence and $\mathcal{G}_T$ be a sequence of $(\lambda,\mu)$-smooth games.\footnote{ Note that if the games in $\mathcal{G}_T$ have different smoothness parameters, we can set $\mu := \max_{t \in [T]} \mu_t$ and $\lambda := \min_{t \in [T]} \lambda_t$. Then, all games in $\mathcal{G
}_T$ are $(\lambda,\mu)$-smooth.} Let  $\pi$ be a profile of~DB($C$) consistent policies, and denote by $g$ any function such that
 \[ \max_{i \in \llbracket N \rrbracket} \max_{(x_t) \in \mathcal{S}^i(C_T)} \mathbb{E}^{\pi} \left[ \mathrm{Reg}^i(\mathcal{U}^i_T,(x_t),(a_t)) \right] \leq g(T). \]
 Suppose that $W_t(a) \geq \sum_i u^i_t(a)$ for all $a$ and $t$. Then, the welfare generated by $\pi$ is bounded from below as follows:
\begin{equation}\label{dyn PoA}
\sum_{t=1}^{T}  \mathbb{E}^{\pi}\left[ W_t(a_{t})\right] \geq \frac{\lambda \beta(C_T)}{1+\mu} \textnormal{OPT-SW}(\mathcal{G}_T) - \frac{N}{1+\mu} g(T),
\end{equation}
\noindent
where $\textnormal{OPT-SW}(\mathcal{G}_T) := \sum_{t=1}^{T} \max_{y_t \in A} W_t(y_t)$ is the optimal centralized social welfare.
\end{proposition}

\noindent The proof of Proposition~\ref{DynPoA} appears in Appendix~\ref{Proofs of welfare in time-varying games}. Dividing by $T$ both sides of Inequality \eqref{dyn PoA}, when the function $g$ is sub-linear, the right-hand side can be approximated for $T$ large enough by  
\begin{equation*}
\frac{\lambda \beta(C_T)}{1+\mu} \frac{\textnormal{OPT-SW}(\mathcal{G}_T)}{T}.
\end{equation*}
\noindent
This implies that the share of the (centralized) optimal social welfare that can be guaranteed in fact scales with the number of action changes allowed in the benchmark sequence.

\section{Internal dynamic benchmark consistency and correlated equilibria}\label{no-internal regret section}
In this section, we show how our techniques and analysis directly extend to other equilibrium notions beyond the coarse correlated equilibrium concept discussed so far. While the dynamic benchmarks considered above replace all the actions ($a^i_t$) selected by player $i$ with some sequence of actions ($x_t$) taken from $S^i(C_T)$, the single-agent learning literature also includes studies of other methods to select benchmark sequences. 

A well-studied approach (see \citealt{lehrer2003wide}, \citealt{blum2007internalregret}) is to identify a pair of actions $(x,y) \in A^i \times A^i$ and replace in the sequence $(a^i_t)$ all occurrences of action $x$ with action~$y$. A specific action pair ($x,y$) is then selected to maximize the (expected) difference between the payoff player~$i$ accrues under $(a^i_t)$ and the payoff they would have obtained under a modified sequence, in which action $x$ is replaced by action $y$. The corresponding regret notion is typically referred to as internal regret and it is formalized as follows.

\begin{definition}[No-internal regret, \citealt{blum2007internalregret}]\label{no-internal regret static}
A policy $\pi^i$ satisfies no-internal regret if there exists a non-negative function  $g: \mathbb{N} \to \mathbb{R}_{+}$ with~$g(T) = o(T)$, such that, for all~$T \geq 1$ and all (potentially correlated) strategies of the other players $\sigma^{-i}: \bigcup_{t \in \llbracket T \rrbracket} A^{t-1} \rightarrow \Delta(A^{-i})$, the internal regret of policy $\pi^i$ is bounded by
\[ \max_{(x,y) \in A^i \times A^i} \mathbb{E}^{\pi^{i},\sigma^{-i}} \left[ \sum_{t \in \llbracket T \rrbracket} \mathbf{1}(a_{t}^i = x) \left(u_{t}^{i}(y,a^{-i}_t) - u^{i}_{t}(x,a_t^{-i}) \right)\right] \leq g(T).\]

\end{definition}

\noindent
In the repeated games setting, \cite{HMSC1} show that when all players deploy policies with no-internal regret, the distribution of play that emerges converges to the set of correlated equilibria, which is a subset of coarse correlated equilibria, and it is defined as,
\begin{equation}\label{Correlated Equilibrium Definition}
    \mathcal{CE}(\Gamma) := \Bigg \{ q \in \Delta(A) \; : \; \forall i \in \llbracket N \rrbracket, \; \forall (x,y) \in A^i \times A^i, \; \underset{a \sim q}{\mathbb{E}}\left[ \mathbf{1}(a_i = x) \left(u^i(y,a^{-i}) - u^i(x,a^{-i}) \right) \right] \leq 0 \Bigg \}.
\end{equation}
\noindent
The interpretation of the set of correlated equilibria $\mathcal{CE}(\Gamma)$ is similar to that of $\mathcal{H}(\Gamma)$, with the main distinction that in correlated equilibria players' deviations are considered ex-post, that is, with knowledge of the action sampled by the mediator through the joint distribution $q \in \Delta(A)$.

Importantly, in the internal regret benchmark a single action modification (from $x$ to $y$) is selected and employed throughout the sequence of actions played by the agent. Analogously to our approach in \S \ref{sect: dbc and tracking error}, we next study the equilibria that can emerge when 
a dynamic sequence of action modifications is employed to construct the internal benchmark at hand. To this end, we denote by~$\mathcal{I}_{C_T}$ the collection of all possible partitions of the time horizon with at most~$C_T + 1$ elements; that is, 
\[ \mathcal{I}_{C_T} := \bigcup_{c = 1}^{C_T+1} \left\{ 
 (I_k)_{k=1}^{c} : I_k = \llbracket a_k, b_k \rrbracket \; \text{for some } a_k \leq b_k, \; I_{k} \cap I_{k'} = \emptyset \; \textnormal{for } k \neq k', \textnormal{ and } \sqcup_{k=1}^{c} I_k = \llbracket T \rrbracket \right\}.\]
 
\noindent
We define internal dynamic benchmark consistent policies in a similar fashion to Definition~\ref{V-Consistency Definition}, extending the no-internal regret property to dynamic sequences of action modifications.

\begin{definition}[Internal dynamic consistency]\label{no-internal regret}
A policy $\pi^i$ is internally dynamic benchmark consistent with respect to the sequence~$C$ if there exists a non-negative function  $g: \mathbb{N} \to \mathbb{R}_{+}$ with~$g(T) = o(T)$, such that, for all~$T \geq 1$ and all (potentially correlated) strategies of the other players $\sigma^{-i}: \bigcup_{t \in \llbracket T \rrbracket} A^{t-1} \rightarrow \Delta(A^{-i})$, its internal regret is bounded by
\[ \max_{\boldsymbol{I} \in \mathcal{I}_{C_T}} \sum_{k=1}^{|\boldsymbol{I}|} \max_{(x,y) \in A^i \times A^i} \mathbb{E}^{\pi^{i},\sigma^{-i}} \left[ \sum_{t \in I_k} \mathbf{1}(a_{t}^i = x) \left(u_{t}^{i}(y,a^{-i}_t) - u^{i}_{t}(x,a_t^{-i}) \right)\right] \leq g(T).\]

\end{definition}

\noindent To simplify exposition, we use as short-hand IDB($C$) to refer to an internally dynamic benchmark subject to the sequence~$C$, and we denote the class of policies that are internally dynamic consistent relative to the sequence~$C$ by $\mathcal{P}^i_{IDB}(C)$.

The existence of policies that are IDB(C) consistent when~$C \equiv 0$ is guaranteed by \cite{HMSCreinforcement} and \cite{blum2007internalregret}. In the appendix, we show that restarting at a suitable rate any IDB($0$) consistent policy suffices to guarantee IDB($C$) consistency (see Proposition \ref{restarting proposition for internal dbc} in Appendix \ref{proofs of no-internal regret section}), when $C_T = o(T)$. This establishes that $\mathcal{P}^i_{IDB}(C)$ is non-empty.

In the repeated games setting, that is, when $V_T = 0$, the same ideas and techniques developed in the proofs of Proposition \ref{Direction 1} and Theorem~\ref{Hannan Set Survives} apply to the internal regret framework. This yields the following asymptotic characterization of the set of possible distributions of play that can emerge under IDB(C) consistent policies as the set of correlated equilibria of the stage game.

\begin{theorem}[Joint distributions of play under internal dynamic benchmark consistency]\label{Correlated equilibria Survive}
Let $C$ be any non-decreasing sub-linear sequence. Fix any stage game $\Gamma$ with rational payoff functions. Then, the set of correlated equilibria $\mathcal{CE}(\Gamma)$ is the smallest closed set of distributions that can be reached by any profile of IDB($C$) consistent policies. That is, if a closed set of distributions $\mathcal{D} \subseteq \Delta(A)$ is such that, for all~$\pi  \in  \mathcal{P}^{N}_{IDB}(C)$,  $ \mathrm{d}_p\left( \bar\delta^{\pi}_T,\mathcal{D}\right) \rightarrow 0 $, then $\mathcal{D} \supseteq \mathcal{CE}(\Gamma)$. 
\end{theorem}

\noindent Therefore, Theorem \ref{Correlated equilibria Survive} establishes that the set of correlated equilibria is a tight characterization for the distribution of play induced by internal dynamic benchmark consistent policies. 

Building on this characterization, we move to analyze time varying multi-agent systems. We leverage the tracking error notion to evaluate the extent to which the joint distribution of play that can emerge under profiles of IDB($C$) policies can approximate the sequence of correlated equilibria sets. Suppose for simplicity that the games in $\mathcal{G}$ are taken from an arbitrarily large but finite set of possible games. Then, we establish that the empirical joint distribution of play that can emerge under any profile of internally dynamic benchmark consistent policies yields sub-linear tracking error, whenever the number of elements~ in the partition $\mathcal{I}_{C_T}$ is sub-linear and bounds from above the number of changes in the stage game, $V(\mathcal{G}_T)$. 

\begin{theorem}[Tracking the sequence of correlated equilibria]\label{thm: restarting = sublinear track error for correlated equilibria}
Let $\pi$~be a collection of internally dynamic benchmark consistent policies relative to a sequence $C$ with $V(\mathcal{G}_T) \leq C_T$ for all $T$, and $C_T =~o(T)$. Then, the tracking error relative to the sequence of correlated equilibrium sets is sub-linear, i.e., 
\vspace{-0.1cm}
\[ \sum_{k=1}^{V(\mathcal{G}_T) + 1} |\mathcal{T}_{(k)}| \textnormal{d}_p(\bar{\delta}^{\pi}_{(k)}, \mathcal{CE}(\Gamma_{(k)})) = o(T).\]
\end{theorem}
\noindent
Notably, no restriction on the class of internal dynamic benchmark consistent policies is needed in Theorem \ref{thm: restarting = sublinear track error for correlated equilibria}. This distinction relative to DB-consistent policies appears because, when $V(\mathcal{G}_T) \leq C_T$, over each batch $\mathcal{T}_{(k)}$, the internal regret is always non-negative. This non-negativity property prevents bad instances where over some batches a policy might perform poorly (that is, incurring regret that linear in the length of the batch), because it is possible to offset this loss by outperforming the benchmark in other batches. 

Finally,~we observe that the welfare guarantees provided in \S\ref{PoA Section} can be directly extended to internal dynamic benchmark consistent policies.~Analogous results to Corollary~\ref{Cor:1} and \ref{dyn poa without smoothness} can be derived, with the reference equilibrium notions replaced by correlated equilibria, instead of coarse correlated equilibria.~Proposition \ref{DynPoA} continues to hold in the internal regret framework, as any IDB($C$)-consistent policy is also DB($C$)-consistent. We omit the details for brevity.

\vspace{-0.0mm}
\section{Concluding remarks}\label{Concluding Remarks}\vspace{-0.0mm}

\textbf{Summary and implications.} In this paper, we formulated and studied a general time-varying multi-agent system where players adopt policies with performance guarantees relative to dynamic sequences of actions. We characterized the empirical distributions of play that can emerge from their interaction, and showed this joint distribution exhibits a diminishing tracking error relative to evolving sequences of stage-game equilibria.  Furthermore, we applied this characterization to quantify the resulting social efficiency. 

Specifically, we showed that when there is no variation in the sequence of games, the empirical distributions of play that can emerge under these policies asymptotically \emph{coincides} with the set of coarse correlated equilibria. This finding implies that the worst-case social welfare, as captured by the price of anarchy, does not scale with the number of action changes allowed in the benchmark, as it is always equal to the PoA of the underlying set of coarse correlated equilibria. 

Contrarily, when the stage game changes, a key separation emerges relative to traditional no-regret policies. First, we showed that the distribution of play induced by DB-consistent policies approximates the evolving sequences of Hannan sets (while traditional no-regret algorithms might fail to do so). Second, we provided refined bounds on the social welfare guaranteed by DB-consistent policies. These bounds are non-decreasing in the number of action changes allowed in the benchmark at hand. Further, we showed that our techniques can be extended to cover other equilibrium notions beyond coarse correlation: by including time-dependent action changes in an internal benchmark, we established that it is possible to track sequences of correlated equilibria.

Through the formulation of the equilibrium tracking error and the characterization of a broad class of learning algorithms under which the tracking error is guaranteed to diminish, our study establishes a strong connection between single-agent learning algorithms that are designed to perform well and non-stationary systems, and time-varying equilibria in these systems. 

Our results have implications on the design and operations of online marketplaces such as ad and retail platforms where non-stationarity is prevalent, and on the design of algorithms for competition in these markets. Our findings suggest that dynamic benchmark consistent policies can provide improved individual and welfare guarantees relative to traditional no-regret policies. Therefore, these policies can be a natural candidate for implementation by competing agents as well as by a market designer that provides algorithms for the participating agents.

\paragraph{Future research directions.}{
Our study provides grounds for further research on few fronts. First, we conjecture that Condition \eqref{clipped dbc} used in Theorem \ref{thm: restarting = sublinear track error} is satisfied by all policies that are dynamic benchmark consistent. Second, in Appendix \ref{a.s. guarantess}, we show that dynamic benchmark consistency can be guaranteed (asymptotically) even with probability one, which implies that Proposition \ref{Direction 1} and Theorem \ref{Hannan Set Survives} can be adapted to obtain an almost surely counterpart, and we leave the investigation of whether the tracking guarantees can be provided with probability one to future research. Finally, we believe it is an interesting and important direction to understand how to modify the techniques used in this paper to cover convex action sets.}

\small
\bibliography{bibliography.bib}

\normalsize 

\renewcommand{\thesection}{\Alph{section}}
\newpage
\appendix

\section*{Appendix}

\section{Proofs of \S\ref{sect: dbc and tracking error}}

\subsection{Proofs of \S \ref{Dynamic Hannan Set Section}} \label{Proofs of Dynamic Hannan Set Section}

\begin{proof}[Proof of Proposition \ref{Direction 1}]
The argument is divided into three steps. In \textbf{Preliminaries}, we maintain assumptions over the payoff structure of each player that will be leveraged in the remaining two steps and we argue that these assumptions are without loss of generality. In \textbf{policy design}, we design a profile of policies, one for each player, whose analysis is then provided in \textbf{Analysis}.

\bigskip\noindent \textbf{Part I: Preliminaries.} Fix player $i \in \llbracket N \rrbracket$. We claim that it is without loss of generality to assume that, for each action $a^i \in A^i$, player $i$'s payoff $u^i(a^i,\cdot)$ is injective in her opponents' actions. This makes the detection step devised in Part II~\emph{infallible}. Let~$\beta^i := (\beta^i(a^{-i}))_{a^{-i} \in A^{-i}} \in \mathbb{R}^{|A^{-i}|}$ be a collection of real numbers and $\alpha^i > 0$ a positive real number. Define a new payoff function $\tilde{u}^i(a^i,a^{-i}) := \alpha^i \left[ u^i(a^i,a^{-i})  + \beta^i(a^{-i}) \right]$. Then, for any belief~$\mu^i \in \Delta(A^{-i})$, player $i$'s preferences remain the same; formally, for actions $a^i,\bar{a}^{i} \in A^i$, it holds $\mathbb{E}^{\mu^i}\left[u^i(a^i,a^{-i})\right] \geq \mathbb{E}^{\mu^i}\left[u^i(\bar{a}^{i},a^{-i}) \right]$ if and only if $\mathbb{E}^{\mu^i}\left[\tilde{u}^{i}(a^i,a^{-i})\right] \geq  \mathbb{E}^{\mu^i}\left[\tilde{u}^{i}(\bar{a}^{i},a^{-i})\right]$. We claim that there exist~$\alpha^i > 0$ and $\beta^i \in \mathbb{R}^{|A^{-i}|}$, such that, for all $a_i \in A_i$, the $a_i$-section of~$\tilde{u}^i$, i.e.,~$\tilde{u}^i(a^i,\cdot)$, is injective. Formally, for each action 
$ a^i \in A^i$, and $h,j  \in A^{-i}$,
\[ h \neq j \implies \tilde{u}^i(a^i,h) \neq \tilde{u}^i(a^i,j).\] 

Indeed, collect the payoffs of player~$i$ into a $|A^i| \times |A^{-i}|$ matrix, denoted by $U^{i}$. Recall that, by assumption,~$|u^i| \leq M$, for all players. Set $\beta^{i}_{0} = -2M$ and define, recursively, $\beta^i_s = \beta^{i}_{s-1} + 3M = \beta^{i}_0 + s(3M)$, for all~$s \in \llbracket|A^{-i}|\rrbracket$. For each $(r,s) \in A^i \times A^{-i}$, let $\hat{U}^{i}_{r,s} = U^{i}_{r,s}+\beta^i_s$. Then, by construction, it holds that, for any $r \in \llbracket |A_i| \rrbracket$
\begin{equation*}
\begin{aligned}
&\hat{U}^{i}_{r,1} &&\in  [0,2M] \\
&\hat{U}^{i}_{r,2} &&\in  [3M,5M] \\
&\vdots \\
&\hat{U}^{i}_{r,|A^{-i}|} &&\in  [(3|A^{-i}|-3)M,(3|A^{-i}|-1)M].
\end{aligned}
\end{equation*}
Thus, for $r \in \llbracket |A_i| \rrbracket$ and $s \neq \tilde{s} \in \llbracket |A^{-i}| \rrbracket$, we have $\hat{U}^{i}_{r,s} \neq \hat{U}^{i}_{r,\tilde{s}}$. Finally, set $\alpha^i = (3|A^{-i}|-1)M$, and define $\tilde{U}^{i}_{r,s} = \alpha^i \hat{U}^{i}_{r,s}$. Given the above arguments, we can assume, without loss of generality, that the sections of the payoff functions of all players in the underlying stage game $\Gamma$ are injective. 

\bigskip\noindent \textbf{Part II: Policy Design.} Observe that, for each~$T \in \mathbb{N}$, the larger $C_T$ is, the stricter the requirement imposed by a dynamic
benchmark endowed with $C_T$ changes. Thus, it suffices to establish the result for $\tilde{C}$ equal to the constant sequence at zero. Let $C$ be any sequence of non-decreasing real numbers such that $C_T = o(T)$. Fix any profile of Hannan consistent policies~$\tilde{\pi} = (\tilde{\pi}^i)_{i = 1}^{N}$. By applying the Bolzano-Weierstrass theorem for bounded sequences, the sequence of empirical distribution of play, $(\bar{\delta}^{\tilde{\pi}}_T: T \in \mathbb{N})$ admits a convergent sub-sequence, whose limit we denote by $\tilde{q} \in \Delta(A)$. As usual, we identify distributions with vectors. 

Note that, by construction, for all $i \in \llbracket N \rrbracket$ and $x \in A_i$,

\begin{equation}\label{non-neg eq}
    \sum_{a \in A}(u^i(x,a^{-i}) - u^i(a))\tilde{q}(a) \leq 0.
\end{equation}

Denote the ordered support of $\tilde{q}$ (under the lexicographical order) by $(a_{(1)},...,a_{(S)})$. We suppose for now that the corresponding probability masses are rational numbers, and we express them as~$\tilde{q}_s =~\frac{m_s}{m}$ where $m_s,m \in~\mathbb{N}$, for all $s \in \llbracket S \rrbracket$.  We construct the policy chosen by player $i$, denoted by $\pi^i$, as a \emph{trigger}-type policy that consists of three elements: a \emph{cooperative} policy, an underlying DB($C$) consistent policy,~$\hat{\pi}_i$, and a detection test.  The cooperative schedule is designed to ensure convergence of the distribution of play to the target distribution $\tilde{q}$, when all players are cooperating. The detection test serves to identify deviations by the other players from the cooperative outcome, and the underlying DB($C$) consistent policies ensures robust guarantees also in case of defection by the other players. Specifically, the cooperative schedule prescribes player $i$ to cycle through her part in the elements of the support of $\tilde{q}$, i.e., $(a_{(1)}^i,...,a_{(S)}^i)$, with a per-cycle frequency that equals to the target distribution~$\tilde{q}$. When cooperating, player $i$ continues to play the following cycle until reaching $T$:\vspace{-0.1cm}
\[\underbrace{\underbrace{a_{(1)}^i,...,a_{(1)}^i}_{m_1 \; \text{times}},\underbrace{a_{(2)}^i,...,a_{(2)}^i}_{m_2 \; \text{times}},\;... \;,\underbrace{a_{(S)}^i,...,a_{(S)}^i}_{m_S \; \text{times}}}_{\sum_{s=1}^{S} m_s = m \; \text{times}} \; \mid \; \underbrace{\underbrace{a_{(1)}^i,...,a_{(1)}^i}_{m_1 \; \text{times}},\underbrace{a_{(2)}^i,...,a_{(2)}^i}_{m_2 \; \text{times}},\;... \;,\underbrace{a_{(S)}^i,...,a_{(S)}^i}_{m_S \; \text{times}}}_{\sum_{s=1}^{S} m_s = m \; \text{times}} \; \mid ...\vspace{-0.1cm}
\]

\noindent
Denote by $\mathcal{W}_T^s$ all periods in which the profile $a_{(s)}^i$ is scheduled to be played when the horizon has length $T$. The detection test compares the payoff obtained at the end of each period $t$ with the one player~$i$ would have obtained if all other players except $i$ were cooperating. For $t \in \llbracket 2,T \rrbracket$, denote by $D_{T}^i(t)$ the collection of all periods before time $t-1$ in which the other players cooperated, i.e., 
\[ D_{T}^i(t) := \left \{ \tau \in \llbracket t-1 \rrbracket \; : \; \sum_{s=1}^{S} \mathbf{1}(\tau \in \mathcal{W}^T_s) \left( \underbrace{u^i(a_{(s)}^i,a^{-i}_{\tau})}_{\text{feedback $i$ receive at $\tau$}} - \underbrace{u^i(a_{(s)}^i,a_{(s)}^{-i})}_{\text{input of the policy}} \right) = 0\right\}.\]
\noindent
Therefore, if $D_{T}^i(t) = \llbracket t-1 \rrbracket$, player $i$ infers that the other players $-i$ cooperated until time~$t-1$. Finally, let $\hat{\pi}_i$ be any DB($C$) consistent policy that is adopted by player $i$ in case of detected defection by the other players. This policy exists by Proposition \ref{prop: restarting = sublinear external-internal dynamic regret}. Then, the trigger-type policy~$\pi^i$ prescribes to cooperate until a deviation is detected. If a deviation is detected, $\pi^i$ forgets all information acquired before the deviation time and immediately switches to the underlying~DB($C$) consistent policy $\hat{\pi}_{i}$, which is played until the end of the horizon~$\llbracket T \rrbracket$.



{\bigskip\noindent \textbf{Part III: Analysis.} We show that the trigger-type policies designed above are DB($C$) consistent and, when implemented by all players, produce a distribution of play that converges to the target distribution $\tilde{q}$. First, note that as all players start being cooperative, none of them detects a deviation by the other players, and so they continue to cooperate in the future. Therefore, by construction, $ \bar{\delta}^{\pi}_T \to \tilde{q}$ as $T \to \infty$. It remains to show that $\pi^i$ is DB($C$) consistent. Fix $T \in \mathbb{N}$. We assume without loss of generality that~$T \geq m$, and we let $\sigma^{-i} : \cup_{t=1}^{T} A^{t-1} \to \prod_{j \neq i} \Delta(A^{-i})$ be any (potentially correlated) strategies by the other players $-i$.
There are two cases to be discussed. 

\noindent
\textbf{Case 1}: Assume first that $\mathbb{P}^{\pi^i,\sigma^{-i}}\left( \forall t \in \llbracket T \rrbracket, \; a_{t}^{-i} = \sum_{s=1}^{S} \mathbf{1}(\tau \in \mathcal{W}_{T}^s)a_{(s)}^{-i} \right) = 1$, that is the other players~$-i$ do not deviate from their cooperative policy. Therefore, by using the notation~$c_t := (c_{t}^i,c_{t}^{-i})$ to refer to the cooperative actions played at time $t$, we obtain 
\begin{equation*}
    \begin{aligned}
     \underset{x \in \mathcal{S}^i(C_T)}{\text{max}} \sum_{t=1}^{T} \mathbb{E}^{\pi^i,\sigma^{-i}}\left[u^i(x_t,a^{t}_{i}) - u^i(a^{t})\right]  
    &=  \underset{x \in \mathcal{S}^i(C_T)}{\text{max}} \sum_{t=1}^{T} \left( u^i(x_t,c_{t}^{-i}) - u^i(c^{t}_{i},c_{t}^{-i})\right). \\
    \end{aligned}
\end{equation*}

Next, we show that, under the above cycling structure,  
\begin{equation}\label{key equation} \underset{x \in \mathcal{S}^i(C_T)}{\text{max}} \sum_{t=1}^{T} \left( u^i(x_t,c_{t}^{-i}) - u^i(c_{t}^{i},c_{t}^{-i})\right) \leq h(T),
\end{equation}
\noindent
for some non-negative function $h = o(T)$.

Pick $(x_t)_{t=1}^{T} \in \mathcal{S}^i(C_T)$. Denote by $U_T$ the actual number of changes in the fixed sequence~$(x_t)_{t=1}^{T}$. Such a sequence generates  $U_T+1$ batches $\{\mathcal{T}_1 , \ldots, \mathcal{T}_{U_T+1}\}$ where the action is constant on each batch, i.e., for each $j$ we have $x^{t} =x$ for all $t$ in the batch $\mathcal{T}_j$ for some action $x$. Fix a batch $\mathcal{T}_j$ and suppose that $x_t = x_j$, for all $t \in \mathcal{T}_j$. Denote by $|\mathcal{T}_{j}|$ the cardinality of $\mathcal{T}_{j}$, and by~$\lfloor x \rfloor$ the integral value lower than $x$.  Recall that $m$ is the size of a cooperative cycle. Suppose that there are $k \geq 0$ cycles that are included in the batch  $\mathcal{T}_{j}$ and let $Q_{j} \subseteq \mathcal{T}_j$ be the set of periods that includes all the periods that belong to cycles in $\mathcal{T}_j$, i.e., $Q_{j} = \{ t_{1} , \ldots , t_{km} \}$ where $t_{1}$ starts the first cycle in $\mathcal{T}_{j}$ and $t_{km}$ ends the last cycle in $\mathcal{T}_{j}$. If $k=0$ then $W_{j}$ is the empty set. We have

\begin{equation} \label{Ineq:Bar1}
\begin{aligned}
 \sum_{t \in \mathcal{T}_j} u^i(x_j,c_{t}^{-i}) & =  \sum_{t \in Q_{j}} u^i(x_j,c_{t}^{-i}) +  \sum_{t \in \mathcal{T}_j \setminus Q_{j}} u^i(x_j,a_{t}^{-i} ) \\
 & = k \sum_{t =1 }^{m} u^i(x_j,c_{t}^{-i}) +  \sum_{t \in \mathcal{T}_j \setminus Q_{j}} u^i(x_j,c_{t}^{-i} ) \\
& \leq  \left \lfloor \frac{|\mathcal{T}_j|}{m} \right \rfloor \sum_{t=1}^{m} u^i(x_j,c_{t}^{-i}) + 2m,
\end{aligned}
\end{equation}
where the inequality follows from the facts that $k \leq \lfloor |\mathcal{T}_j| / m  \rfloor $, payoffs are bounded in $[0,1]$, and from noting that the cardinality of $\mathcal{T}_j \setminus Q_{j}$ is at most $2m$. Then, 

\begin{equation}\label{crucial ineq}
\begin{aligned}
\sum_{t \in \mathcal{T}_j} u^i(x_j,c_{t}^{-i}) &= \left \lfloor \frac{|\mathcal{T}_j|}{m} \right \rfloor \left (\sum_{t=1}^{m} u^i(x_j,c_{t}^{-i}) \right) + \left( \sum_{t \in \mathcal{T}_j} u^i(x_j,c_{t}^{-i}) -   \left \lfloor \frac{|\mathcal{T}_j|}{m} \right \rfloor \sum_{t=1}^{m} u^i(x_j,c_{t}^{-i})\right) \\
&\leq  |\mathcal{T}_j| \; \underset{y \in A_i}{\text{max}} \; \frac{1}{m} \;\sum_{t=1}^{m} u^i(y,c_{t}^{-i}) + 2m,
\end{aligned}
\end{equation}
\vspace{-0.1cm}
\noindent
where the inequality holds because $\lfloor x \rfloor \leq x$ and Inequality ($\ref{Ineq:Bar1}$). By summing over all elements of the partition, we obtain
\begin{equation}
\begin{aligned}
\sum_{j=1}^{U_T+1}\sum_{t \in \mathcal{T}_j} u^i(x_j,c_{t}^{-i}) &\overset{(a)}{\leq} \sum_{j=1}^{U_T+1} \left( |\mathcal{T}_j| \; \underset{y \in A_i}{\text{max}} \; \frac{1}{m} \;\sum_{t=1}^{m} u^i(y,c_{t}^{-i}) + 2 m \right)  \\
&=  T \; \underset{y \in A_i}{\text{max}} \;\frac{1}{m}\sum_{t=1}^{m} u^i(y,c_{t}^{-i}) + 2 m \left(U_T+1 \right) \\
&\overset{(b)}{\leq} T \; \underset{y \in A_i}{\text{max}} \;\frac{1}{m}\sum_{t=1}^{m} u^i(y,c_{t}^{-i}) + 2 m \left(C_T+1 \right), 
\end{aligned}
\end{equation}
\noindent
where ($a$) holds because of (\ref{crucial ineq}), ($b$) holds because $U_T \leq C_T$.
As the sequence of $(x_t)_{t=1}^{T} \in \mathcal{S}^i(C_T)$ was arbitrary, by using the assumption $C_T = o(T)$, we obtain
\begin{equation}\label{eq: small benchmark gap}
    \underset{x \in \mathcal{S}^i(C_T)}{\text{max}}  \sum_{t=1}^{T}  u^i(x_t,c_{t}^{-i}) - T \underset{y \in A_i}{\max} \;  \frac{1}{m}\sum_{t=1}^{m} u^i(y,c_{t}^{-i}) \leq 2m(C_T + 1) =: L(T) = o(T).
\end{equation}

On the other hand, by using Inequality \eqref{non-neg eq} and recalling that $(1/T) \sum_{t=1}^{T} \delta_{c_T} = \bar{\delta}^{\pi}_T \to \tilde{q}$ as $T \to \infty$, it follows 

\begin{equation}\label{R(T) = o(T)}
     R(T) := \left(T\underset{y \in A_i}{\max} \frac{1}{m} \sum_{t=1}^{m} u^i(y,c_{t}^{-i}) -  \sum_{t=1}^{T} u^i(c_{t}^{i},c_{t}^{-i}) \right)_{+} = o(T).
\end{equation}

\noindent
The following steps hold
\begin{equation*}
\begin{aligned}
  &\phantom{=}\underset{x \in \mathcal{S}^i(C_T)}{\text{max}}  \sum_{t=1}^{T}  \left( u^i(x_t,c_{t}^{-i}) - u^i(c_{t}^{i},c_{t}^{-i}) \right) \\
  &= \Bigg(\underset{x \in \mathcal{S}^i(C_T)}{\text{max}}  \sum_{t=1}^{T}  u^i(x_t,c_{t}^{-i}) - T\underset{y \in A_i}{\text{max}} \frac{1}{m}  \sum_{t=1}^{m}  u^i(y,c_{t}^{-i}) \Bigg) + \left(T\underset{y \in A_i}{\max} \frac{1}{m}  \sum_{t=1}^{m} u^i(y,c_{t}^{-i}) -  \sum_{t=1}^{T} u^i(c_{t}^{i},c_{t}^{-i}) \right)  \\
  &\leq \Bigg(\underset{x \in \mathcal{S}^i(C_T)}{\text{max}}  \sum_{t=1}^{T}  u^i(x_t,c_{t}^{-i}) - T\underset{y \in A_i}{\text{max}} \frac{1}{m}  \sum_{t=1}^{m}  u^i(y,c_{t}^{-i}) \Bigg) +  \left(T\underset{y \in A_i}{\max} \frac{1}{m}  \sum_{t=1}^{m} u^i(y,c_{t}^{-i}) -  \sum_{t=1}^{T} u^i(c_{t}^{i},c_{t}^{-i}) \right)_{+} \\
  &\leq L(T) + R(T) =: h(T) = o(T).
\end{aligned}
\end{equation*}

\noindent
For the following part of the proof, we assume without loss of generality that $h$ is non-decreasing. Indeed, if not, we can just replace the function $h$ with $\tilde{h}$ defined as $\tilde{h}(T) := \max_{\tau \in \llbracket T \rrbracket}h(\tau)$, and we still have $\tilde{h}(T) = o(T)$.

\noindent
\textbf{Case 2}: Suppose now that $\mathbb{P}^{\pi,\sigma^{-i}}\left(\forall t \in \llbracket T \rrbracket, \; a_{t}^{-i} = \sum_{s=1}^{S} \mathbf{1}(\tau \in \mathcal{W}_{T}^s)a_{(s)}^{-i}\right) <~1$. Denote by $D_T$ the (random) deviation time by the other players $-i$.
Recall that we use the notation $c_t = (c_{t}^i,c_{t}^{-i})$ to denote the cooperative actions played at time $t$. For each $(x_t: t \in \llbracket T \rrbracket) \in \mathcal{S}^i(C_T)$, we have

\begin{equation*}
\begin{aligned}
    \mathbb{E}^{\pi^i,\sigma^{-i}}\left[ \sum_{t=1}^{T} u^i(x_t,a_{t}^{-i}) - u^i(a_t)\right] &= 
    \sum_{d=1}^{T}\mathbb{E}^{\pi^i,\sigma^{-i}}\left[ \sum_{t=1}^{T} u^i(x_t,a_{t}^{-i}) - u^i(a_t) \; \Bigg| \; D_T = d \right]\mathbb{P}^{\pi^i,\sigma^{-i}}(D_T = d) \\
    &+ \mathbb{E}^{\pi^i,\sigma^{-i}}\left[ \sum_{t=1}^{T} u^i(x_t,a_{t}^{-i}) - u^i(a_t) \; \Bigg| \; D_T = \infty\right]\mathbb{P}^{\pi^i,\sigma^{-i}}(D_T = \infty).
\end{aligned}
\end{equation*}

In the following, to ease notation we define
\begin{equation}\label{finite time dev}
\mathbb{E}^{\pi^i,\sigma^{-i}}_{d}\left[ \sum_{t=1}^{T} u^i(x_t,a_{t}^{-i}) - u^i(a_t)\right] := \mathbb{E}^{\pi^i,\sigma^{-i}}\left[ \sum_{t=1}^{T} u^i(x_t,a_{t}^{-i}) - u^i(a_t) \; \Bigg| \; D_T = d \right].
\end{equation}

and similarly
\begin{equation}\label{infinite time dev}
\mathbb{E}^{\pi^i,\sigma^{-i}}_{\infty}\left[ \sum_{t=1}^{T} u^i(x_t,a_{t}^{-i}) - u^i(a_t)\right] := \mathbb{E}^{\pi^i,\sigma^{-i}}\left[ \sum_{t=1}^{T} u^i(x_t,a_{t}^{-i}) - u^i(a_t) \; \Bigg| \; D_T = \infty\right].
\end{equation}

Now, by using the same analysis for Case 1, we immediately conclude that
\[ \mathbb{E}^{\pi^i,\sigma^{-i}}_{\infty}\left[ \sum_{t=1}^{T} u^i(x_t,a_{t}^{-i}) - u^i(a_t)\right] \leq h(T).\]

On the other hand,

\begin{equation*}
\begin{aligned}
    \mathbb{E}^{\pi^i,\sigma^{-i}}_{d}\left[ \sum_{t=1}^{T} u^i(x_t,a_{t}^{-i}) - u^i(a_t)\right] &= \mathbb{E}^{\pi^i,\sigma^{-i}}_{d}\left[ \sum_{t=1}^{d} u^i(x_t,a_{t}^{-i}) - u^i(a_t)\right] + \mathbb{E}^{\pi^i,\sigma^{-i}}_{d}\left[ \sum_{t=d +1}^{T} u^i(x_t,a_{t}^{-i}) - u^i(a_t)\right] \\
    &\overset{(a)}{\leq}  \mathbb{E}^{\pi^i,\sigma^{-i}} _{d}\left[\sum_{t=1}^{d -1} u^i(x_t,c_{t}^{-i}) - u^i(c_t)  \right] + 1 + \mathbb{E}^{\pi^i,\sigma^{-i}}_{d}\left[ \sum_{t=d +1}^{T} u^i(x_t,a_{t}^{-i}) - u^i(a_t)\right],
\end{aligned}
\end{equation*}
\noindent
where ($a$) holds because, by definition, all players stick to their cooperative schedule up to $d - 1$, so that 
\[ \sum_{t=1}^{d -1}  \left(u^i(x_t,a_{t}^{-i}) - u^i(a_t) \right) = \sum_{t=1}^{d -1}  \left(u^i(x_t,c_{t}^{-i}) - u^i(c_t) \right),\]

and because, as payoffs are in $[0,1]$,
\[ \mathbb{E}^{\pi^i,\sigma^{-i}}_{d}\left[ u^i(x_{d},a_{d }^{-i})-u^i(a_{d }^i,a_{d}^{-i}) \right] \leq 1.\]

Further, the following holds

\begin{equation}
    \begin{aligned}
        \mathbb{E}^{\pi^i,\sigma^{-i}}_{d}\left[ \sum_{t=d+1}^{T} u^i(x_t,a_{t}^{-i}) - u^i(a_t)\right] 
        &\overset{(a)}{=} 
        \sum_{a^{-i} \neq a_{d}^{-i}} \mathbb{E}^{\hat{\pi}_i,\sigma^{-i} \mid a^{-i}}\left[ \sum_{t=1}^{T-d} u^i(x_t,a_{t}^{-i}) - u^i(a_t) \right]\mathbb{P}^{\pi^i,\sigma^{-i}}(a_{d}^{-i} = a^{-i} \mid D_T = d)   \\
        &\overset{(b)}{\leq} f(T-d) \overset{(c)}{\leq} f(T) = o(T),
    \end{aligned}
\end{equation}

where ($a$) holds because immediately after detecting a defection, player $i$ switches to~$\hat{\pi}_i$, and by defining $\sigma^{-i} \mid a^{-i}$ as follows: for $t \in \llbracket T-d \rrbracket$ and $(a^{\tau} : \tau \in \llbracket t \rrbracket)$,

\[ \sigma^{-i} \mid a^{-i}((a^{\tau} : \tau \in \llbracket t \rrbracket)) := \sigma^{-i}((c_1,\ldots,c_{d-1},(c_{d}^i,a^{-i})) \oplus (a^{\tau} : \tau \in \llbracket t \rrbracket)),\]

where $\oplus$ denotes the concatenation operator between two sequences. ($b$) holds because $\hat{\pi}_i$ is DB($C$) consistent; that is there exists a non-negative (and wlog) non-decreasing function $f$ such that ($b$) holds. Finally, ($c$) uses the monotonicity of $f$. Similarly, we have 

\begin{equation}
    \begin{aligned}
    \mathbb{E}^{\pi^i,\sigma^{-i}}_{d}\left[\sum_{t=1}^{d-1} u^i(x_t,c_{t}^{-i}) - u^i(c_t) \right] &=  \sum_{t=1}^{d-1} \left(u^i(x_t,c_{t}^{-i}) - u^i(c_t) \right) \\
    &\overset{(a)}{\leq} (d-1) \mathbf{1}(d \leq m) + h(d-1) \mathbf{1}(d \geq m+1).
    \end{aligned}
\end{equation}

\noindent
where ($a$) holds by bounding (as $u^i(a) \in [0,1]$ for all $a$), $\sum_{t=1}^{d-1} u^i(x_t,c_{t}^{-i}) - u^i(c_t) \leq d-1$, if $d \leq m$, and, with the same steps underlying Case $1$, by bounding
\[ \sum_{t=1}^{d-1} u^i(x_t,c_{t}^{-i}) - u^i(c_t) \leq h(d-1),\]
if $d \geq m+1$. Define $g(T) := m + f(T) + h(T)$. Then, on the one hand, $g(T) = o(T)$, and, on the other hand, as we can assume wlog that $h$ is non-decreasing,
\[ \mathbb{E}^{\pi^i,\sigma^{-i}}\left[ \sum_{t=1}^{T} u^i(x_t,a_{t}^{-i}) - u^i(a_t)\right] \leq g(T).\]

As the sequence $(x_t: t \in \llbracket T \rrbracket)$ was an arbitrary element of $\mathcal{S}^i(C_T)$, this shows that under Case~$2$, the policy $\pi^i$ satisfies no-regret relative to a dynamic benchmark with $C_T$ action changes (when the horizon has length $T$). Altogether, Case $1$ and Case $2$ establish that $\pi^i$ is DB($C$) consistent.

Finally, the generalization to irrational probability masses is established through standard density arguments. This concludes the proof.
}
\end{proof}

\begin{proof}[Proof of Theorem \ref{Hannan Set Survives}]
Let $C$ be any sequence of non-decreasing real-numbers such that $C_T = o(T)$ and $\pi = (\pi^i)_{i =1}^{N} \in \mathcal{P}^{N}_{DB}(C)$ be a profile of DB($C$) consistent policies. First of all, observe that, for any closed set $\mathcal{D} \subseteq \Delta(A)$,
\[  d(\bar{\delta}^{\pi}_T,\mathcal{D}) \rightarrow 0  \iff \text{for any convergent sub-sequence} \; (\bar{\delta}_{T_{K}}^{\pi})_{K}, \; \text{its limit belongs to }\mathcal{D}.\]

\noindent
In order to establish the result, let

\[ \mathcal{H}^{D}(C) :=  \bigcup_{\pi \in \mathcal{P}_{DB}^N(C)} \left\{ q \in \Delta(A) \; : \;  q \in \mathrm{Lim} \left( \bar{\delta}^{\pi}_T \right) \right\}  \]

\noindent
be the collection of all joint distributions that can be reached as limit by a convergent sub-sequence of the joint distribution of play generated by a profile of DB($C$) consistent policies. By applying the Bolzano-Weierstrass Theorem for bounded sequences, we conclude that $\mathcal{H}^{D}(C) \neq \emptyset$. 

Standard arguments (see for example \citealt{NisaRougTardVazi07}) imply that $ \mathcal{H}^{D}(C) \subseteq \mathcal{H} $. Denote by~$\Delta(A)_{\mathbb{Q}} := \left \{ q \in \mathbb{Q}^{|A|} : \forall a \in A, \; q(a) \geq 0, \; \text{and} \; \sum_{a \in A}q(a) = 1\right\}$ the collection of all joint distributions made only of rational components. 
As, by assumption, the payoff functions of the players have rational values, it follows that the vertices of $\mathcal{H}$ have rational coordinates; hence, $\mathcal{H} \cap \Delta(A)_{\mathbb{Q}} \neq \emptyset$ and $\text{cl}_{\Delta(A)} \left(\mathcal{H} \cap \Delta(A)_{\mathbb{Q}} \right) = \mathcal{H}$, where $\text{cl}_{\Delta(A)}(B)$ denotes the closure of a set~$B \subseteq \Delta(A)$. Let $q \in \mathcal{H} \cap \Delta(A)_{\mathbb{Q}}$. The trigger-type policies devised in Proposition \ref{Direction 1} are~DB($C$) consistent and can be used to ensure that the corresponding distribution of play converges to $q$. Therefore, it holds that~$\mathcal{H} \cap \Delta(A)_{\mathbb{Q}} \subseteq \mathcal{H}^{D}(C)$. We have established 
\[ \mathcal{H}^{D}(C) \subseteq \mathcal{H} \quad \textnormal{and} \quad \mathcal{H} \cap \Delta(A)_{\mathbb{Q}} \subseteq \mathcal{H}^{D}(C).\]
The proof is concluded by taking the closure in both inclusions.

\end{proof}

\section{Proofs of \S \ref{sect: time varying games}}\label{proofs of sect: sect time varying games}

\begin{lemma}\label{lemma: regret-distance}
Let $\Gamma = (u^i : i \in \llbracket N \rrbracket)$ be a stage game with corresponding Hannan set $\mathcal{H}$. Then, there exists a constant, $\textnormal{const}(\Gamma)$, such that for any~$\epsilon \geq 0$, $ T \in \mathbb{N}$, and any profile profile of policies $\pi \in \mathcal{P}$, if, 
$$
\max_{(x_t) \in \mathcal{S}^i(0)} \mathbb{E}^{\pi} \left[ \sum_{t=1}^{T} u^i(x_t,a^{-i}_t) - u^i(a_t) \right] \leq \epsilon T
$$
for each player $i$, then, the distance between the empirical distribution of play $\bar{\delta}^{\pi}_{T}$ and the Hannan set $\mathcal{H}$ is bounded by$$\textnormal{d}_p(\bar{\delta}^{\pi}_{T},\mathcal{H}) \leq \min \{ \textnormal{const}(\Gamma) \epsilon , 2^{1/p}\}.$$
\end{lemma}

\begin{proof}
To ease notation let $K := \sum_{i \in \llbracket N \rrbracket} K^i$ and $Z := \prod_{i \in \llbracket N \rrbracket} K^i$. Further, we recall that $\textnormal{d}$ is induced by some $p$-norm with $p \geq 1$ or $ p = \infty$, and that the outcome space $A$ is endowed with the lexicographical order. We write the ordered (wrt lexicographic order) sequence of outcomes as
\[ (a_{(1)}, a_{(2)},\ldots, a_{(Z)}).\]
By assumption, for each player $i$, 
$$\max_{(x_t) \in \mathcal{S}^i(0)} \mathbb{E}^{\pi} \left[ \sum_{t=1}^{T} u^i(x_t,a^{-i}_t) - u^i(a_t) \right] \leq \epsilon T.$$
\noindent 
By definition, it follows that $\bar{\delta}^{\pi}_T$ is an $\epsilon$-approximate Hannan equilibrium, that is, $\bar{\delta}^{\pi}_T \in \mathcal{H}_{\epsilon}$,  where
    \[ \mathcal{H}_{\epsilon} := \left\{q \in \Delta(A) : \forall i \in \llbracket N \rrbracket, \forall x \in A_i, \; \mathbb{E}_{a \sim q}\left[u^i(x,a^{-i}) - u^i(a)\right] \leq \epsilon \right\}. \]
For $x,y \in \mathbb{R}^{K}$, we use the notation $x \preceq y$ if $x_k \leq y_k$, for all $k$. Then, we conveniently rewrite the Hannan set and its $\epsilon$-approximate version in matrix form 
\[ \mathcal{H} = \left\{q \in \mathbb{R}^{Z} : Q q \preceq 0, 1^{\top} q = 1 \right\} \textnormal{ and } \mathcal{H}_{\epsilon} = \left\{q \in \mathbb{R}^{Z} : Q q \preceq \begin{bmatrix}
\epsilon 1 \\
0
\end{bmatrix}, 1^{\top} q = 1 \right\},\]

where $Q = \begin{bmatrix}
R \\
-I
\end{bmatrix}$ with $I$ the $Z \times Z$ identity matrix and $R$ a $K \times Z$ matrix containing the possible incentives to deviate for each player, that is, 

\[ R := \begin{bmatrix}
    u^{1}(1,a^{-1}_{(1)}) - u^{1}(a_{(1)})  & \cdots & u^{1}(1,a^{-1}_{(Z)}) - u^{1}(a_{(Z)}) \\
    \vdots & \vdots & \vdots \\
    u^{1}(K^1,a^{-1}_{(1)}) - u^{1}(a_{(1)})  & \cdots & u^{1}(K^1,a^{-1}_{(Z)}) - u^{1}(a_{(Z)}) \\
    \vdots & \vdots & \vdots \\
    u^{N}(1,a^{-N}_{(1)}) - u^{N}(a_{(1)})  & \cdots & u^{N}(1,a^{-N}_{(Z)}) - u^{N}(a_{(Z)}) \\
    \vdots & \vdots & \vdots \\
    u^{N}(K^N,a^{-N}_{(1)}) - u^{N}(a_{(1)})  & \cdots & u^{N}(K^N,a^{-N}_{(Z)}) - u^{N}(a_{(Z)})
\end{bmatrix}. \]

Clearly, $\textnormal{rank}(Q) = Z$. Since the profile of policies $\pi$ is arbitrary, we do not know exactly where~$\bar{\delta}^{\pi}_T$ lives inside~$\mathcal{H}_{\epsilon}$. Thus, we take a worst-case approach and provide a bound for the Hausdorff distance between the polytopes $\mathcal{H}$ and $\mathcal{H}_{\epsilon}$. Indeed, it holds $\textnormal{d}_p(\bar{\delta}^{\pi}_T, \mathcal{H}) \leq \textnormal{d}_{H}(\mathcal{H}, \mathcal{H}_{\epsilon})$, where 
\[ \textnormal{d}_{H}(\mathcal{H}, \mathcal{H}_{\epsilon}) = \max_{p \in \mathcal{H}_{\epsilon}} \min_{q \in \mathcal{H}} \textnormal{d}_p(p,q).\footnote{Since $\mathcal{H} \subseteq \mathcal{H}_{\epsilon}$, we wrote directly the one-side Hausdorff distance.}\]

By specializing Theorem 2.4 and Theorem 3.2 in \cite{li1993sharp} to our setting, we obtain 
\[ \textnormal{d}_{H}(\mathcal{H}, \mathcal{H}_{\epsilon}) \leq  \alpha(\Gamma) K^{1/p} \epsilon,\]

where 

\[ \alpha(\Gamma) := \max_{I \in \mathcal{M}} \sup_{u,v} \left\{  \textnormal{d}_p \left( \begin{bmatrix}
Q_{I,0} \\
1^{\top}
\end{bmatrix}^{+} \begin{bmatrix}
u \\
v
\end{bmatrix}, S_{I}\right) : \left\|\begin{bmatrix}
u \\
v
\end{bmatrix} \right\|_{p} = 1\right\},\]

and
\begin{itemize}
\item $Q_{I}$ denotes the matrix consisting of rows of $Q$ whose indices are in the index set $I \subseteq \llbracket K+Z \rrbracket$.
\item $P^{+}$ denotes the pseudo-inverse of matrix $P$.
\item $Q_{I,0}$ denotes the matrix obtained by replacing the rows
of $Q$ whose indices are not in $I$ with~$0$.
\item $\mathcal{M} := \left\{ I \subseteq \llbracket K+Z \rrbracket: |I| = Z-1, \textnormal{rank}\left(\begin{bmatrix}
Q_{I} \\
1^{\top}
\end{bmatrix}\right) = Z\right\}$.
\item $S_{I} := \{ x \in \mathbb{R}^{Z}: Q_{I} x \preceq 0, 1^{\top}x = 0\}$.
\end{itemize}

Finally, by setting $\textnormal{const}(\Gamma) := \alpha(\Gamma) K^{1/p}$ and using the fact that distributions belong to the unit simplex, which implies that $\|q_0-q_1\|_p \leq 2^{1/p}$ for $q_0,q_1 \in \Delta(A)$, we get the more refined bound 
\[\textnormal{d}_{H}(\mathcal{H}, \mathcal{H}_{\epsilon}) \leq \min \{ \textnormal{const}(\Gamma) \epsilon, 2^{1/p}\}. \]

\end{proof}

\begin{remark}\label{remark for constant}For a sequence of games $\mathcal{G}_T$, we denote by $\textnormal{const}(\mathcal{G}_T) := \max_{\Gamma \in \mathcal{G}_T} \textnormal{const}(\Gamma)$ the largest constant identified in Lemma \ref{lemma: regret-distance}, and similarly we denote $\textnormal{const}(\mathcal{G}) := \sup_{T \in \mathbb{N}} \textnormal{const}(\mathcal{G}_T)$. In all the results below we assume that $\textnormal{const}(\mathcal{G}) < \infty$.
\end{remark}

\begin{example}\label{small regret but big distance}
We provide a simple example where for some distributions in the $\epsilon$-approximated Hannan set, the closest distribution in the Hannan set has distance $\sqrt{2}$ in the $\ell_2$-distance (i.e., the largest distance that can be obtained), even for small values of $\epsilon$. Consider the following game

\begin{center}
    \setlength{\extrarowheight}{1pt}
    \begin{tabular}{cc|c|c|}
      & \multicolumn{1}{c}{} & \multicolumn{1}{c}{{c}}  & \multicolumn{1}{c}{d} \\\cline{3-4}
      & a & $\delta\hspace{0.1cm};\hspace{0.1cm}\delta$ & $1\hspace{0.1cm};\hspace{0.1cm}0$ \\\cline{3-4}
      & b & $0\hspace{0.1cm};\hspace{0.1cm}1$ & $0\hspace{0.1cm};\hspace{0.1cm} 0$ \\\cline{3-4}
    \end{tabular}
\end{center}
\noindent
for some $1 > \delta > 0$. Observe that the only Hannan equilibrium is $\mathcal{H} = \{(1,0,0,0)\}$, i.e., a degenerate distribution assigning probability mass equal to one to the unique Nash equilibrium $(a,c)$. Now, suppose we relax the Hannan equilibria as follows:
\[
\mathcal{H}_{\epsilon} = \left\{q \in \Delta(A) \mid \forall i \in [2], \forall x \in A_i, \; \mathbb{E}_{a \sim q}\left[u^i(x,a^{-i}) - u^{i}(a)\right] \leq \epsilon \right\},
\]
\noindent
for some $1 > \epsilon \geq \delta$. It is not difficult to verify that the degenerate distributions assigning probability mass equal to one to the outcomes $(a,d)$ and $(b,c)$ are $\epsilon$-approximate Hannan equilibria, that is, under our usual abuse notation for distributions, it holds $(0,1,0,0)  \in \mathcal{H}_{\epsilon}$ and $(0,0,1,0) \in \mathcal{H}_{\epsilon}$. However, we clearly have that $\|(0,1,0,0) -  (1,0,0,0)\|_2 = \sqrt{2}$.
\end{example}

\subsection{Proofs of \S \ref{subsect: single best-reply games}}\label{proof of single-best reply games}

\begin{proof}[Proof of Theorem \ref{suff result}]

We abuse notation and write the more compact $V_T$ instead of $V(\mathcal{G}^{\textnormal{sbr}}_T)$. We denote by $\mathcal{T}_{(k)}$ the batches where the stage game $\Gamma^{\textnormal{sbr}}_{(k)}$ is constant. Recall that $\bar{\delta}_{(k)}^{\pi}$ denotes the distribution of play restricted to the $k$-th batch induced by a profile of policies $\pi \in \mathcal{P}$.

Suppose that, each player $i$ is employing a DB($C^i$)-consistent policy, denoted by $\pi^i$, with~$C^i_T \geq V_T$, for all $T \geq 1$. Since the underlying games have the single best-reply property, and the benchmarks allow $C^i_T \geq V_T$ action changes, for each batch, they can select the corresponding single best-reply actions. Thus, the regret of any policy relative to these benchmarks is non-negative. 

For each batch $\mathcal{T}_{(k)}$, denote the static regret for player $i$ over that batch by

\[ \textnormal{R}_k^i := \max_{x \in A_i} \; \mathbb{E}^{\pi} \left[ \sum_{t \in \mathcal{T}_{(k)}} u_{(k)}^{i}(x,a^{-i}_t) - u_{(k)}^{i}(a_t) \right].\]
\noindent
Similarly, we denote by $\textnormal{R}_{k} := \max_{i \in \llbracket N \rrbracket} \textnormal{R}_{k}^i$ the largest regret over the $k$-th batch. Then, we have
\begin{equation}\label{tracking error from definition dominant}
    \begin{aligned}
    err_p(\mathcal{G}^{\textnormal{sbr}}_T, \pi,\mathcal{H}) = \sum_{k =1}^{V_T+1} |\mathcal{T}_{(k)}|\textnormal{d}_p(\bar{\delta}_{(k)}^{\pi},\mathcal{H}(\Gamma_{(k)})) 
    &\overset{(a)}{\leq}  \sum_{k \in \llbracket V_T+1 \rrbracket } \textnormal{const}(\Gamma^{\textnormal{sbr}}_{(k)}) \textnormal{R}_{k} \\
    &\leq \textnormal{const}(\mathcal{G}^{\textnormal{sbr}}) \sum_{i \in \llbracket N \rrbracket}\sum_{k \in \llbracket V_T+1 \rrbracket} \textnormal{R}_{k}^i \\
    &\overset{(b)}{=} \textnormal{const}(\mathcal{G}^{\textnormal{sbr}}) \sum_{i \in \llbracket N \rrbracket} \max_{(x_t) \in \mathcal{S}^{i}(C^i_T)} \mathbb{E}^{\pi}\left[ \textnormal{Reg}^i(\mathcal{U}^i_T,(x_t),(a_t))\right]. 
    \end{aligned}
\end{equation}
\noindent
where ($a$) holds by recalling that $R^{k} \geq 0$, and by applying Lemma \ref{lemma: regret-distance} with $\epsilon = \textnormal{R}_{k}/|\mathcal{T}_{(k)}|$. ($b$) holds because by assumption~$C^i_T \geq V_T$, and because $\Gamma^{\textnormal{sbr}}_{(k)}$'s are single best-reply games. 

Finally, as, by assumption, for each player $i$, $\pi^i$ is DB($C^i$)-consistent, we conclude that $$err_p(\mathcal{G}^{\textnormal{sbr}}_T, \pi,\mathcal{H}) = o(T).$$
\end{proof}

\begin{proof}[Proof of Theorem \ref{imp result}]
Throughout the proof, for a sequence of games $\mathcal{G}_T$, we use the notation $(\mathcal{G}_T)^i$ to refer to the sequence of player $i$'s payoff functions inside $\mathcal{G}_T$. Let $\mathcal{U}^i$ be single best-replying functions for player $i$ such that 
$$S_T := \sum_{t=1}^{T-1} \mathbf{1}(u^i_{t+1} \neq u^i_t ) = o(T).$$
Observe that (as in the proof of Theorem \ref{suff result}) for each $T \geq 1$, and for any sequence $(a_t) \in A^T$,
\[ \max_{(x_t) \in \mathcal{S}^{i}(S_T)} \textnormal{Reg}^i(\mathcal{U}_T^{i},(x_t),(a_t)) \geq 0,\]
as $\mathcal{U}_T^i$ are single best-reply payoff functions with $S_T$ changes. 

Therefore, in this context, a policy $\bar{\pi}^i$ is DB($(S_T : T\in \mathbb{N})$) consistent if and only if
\begin{equation}\label{dbc alt}
g(T) = \sup_{\sigma^{-i} : \cup_{t=1}^{T} A^{t-1} \to \Delta(A^{-i})} \max_{(x_t) \in \mathcal{S}^{i}(S_T)} \mathbb{E}^{\bar{\pi}^i,\sigma^{-i}}\left[ \textnormal{Reg}^i(\mathcal{U}_T^{i},(x_t),(a_t))\right] = o(T).    
\end{equation}

Observe that player $i$'s regret is a linear function in the (joint) probability of the actions selected by the other players. Therefore, no additional power is provided by allowing $\sigma^{-i}$ to choose distributions in $\Delta(A^{-i})$, instead of deterministic profiles of actions $a^{-i} \in A^{-i}$. In the following, we focus, without loss of generality, on deterministic joint strategies, that is, $\sigma^{-i} : \cup_{t=1}^{T} A^{t-1} \to A^{-i}$. Notably, this observation also implies that the sup in \eqref{dbc alt} can be replaced by a max over deterministic joint strategies, as for each $T$, there exist only finitely many deterministic joint strategies.

As by assumption $\pi^i$ is not DB($(S_T : T \in \mathbb{N})$) consistent, we conclude that, for all $T \geq 1$ there exists a (deterministic) function $\bar{\sigma}^{-i} : \cup_{t=1}^{T} A^{t-1} \to A^{-i}$ such that 
\begin{equation}\label{cond 1}
\max_{(x_t) \in \mathcal{S}^{i}(S_T)} \mathbb{E}^{\pi^i,\bar{\sigma}^{-i}}\left[ \textnormal{Reg}^i(\mathcal{U}_T^{i},(x_t),(a_t))\right] = \Omega(T).
\end{equation}

Importantly, note that, for each $T \geq 1$, $\bar{\sigma}^{-i}$ can be equivalently expressed by specifying $|A^{-i}|$ individual strategies $\bar{\sigma}^j : \cup_{t=1}^{T} A^{t-1} \to A^{j}$, such that, for all $h \in \cup_{t=1}^{T} A^{t-1}$, $\bar{\sigma}^{-i}(h) = \prod_{j \neq i} \bar{\sigma}^j(h)$.

For each $T \in \mathbb{N}$, we construct a sequence of single-best reply games $\mathcal{G}_{T}^{\textnormal{sbr}}$ as follows:
\begin{enumerate}[label = (\roman*)]
    \item The payoff functions of player $i$ coincide with $\mathcal{U}_T^i = (u^i_t : t \in \llbracket T \rrbracket)$, that is $(\mathcal{G}_{T}^{\textnormal{sbr}})^i = \mathcal{U}_T^i$.
    \item For each $j \neq i$, player $j$'s payoff function remains constant, that is, $u^j_t = u^j$, for all $t \in \llbracket T \rrbracket$.
    \item For each $j \neq i$, player $j$'s payoff function is fully determined by the actions of the other players, that is, for all $a^{-j} \in A^{-j}$,   $u^j(x,a^{-j}) = u^j(y,a^{-j})$, for all $(x,y) \in A^j \times A^j$. 
    \item For each player $j \neq i$, and for each action $x \in A^j$, the $x$-section of $u^j$ is injective in the other players' actions, that is, 
    \[ \forall x \in A^j, \forall (a^{-j}, \bar{a}^{-j}) \in A^{-j} \times A^{-j}, \; \left( a^{-j} \neq \bar{a}^{-j} \implies u^j(x,a^{-j}) \neq u^j(x,\bar{a}^{-j}) \right).\]
\end{enumerate}
As usual, we denote $\mathcal{G}^{\textnormal{sbr}} = (\mathcal{G}^{\textnormal{sbr}}_T : T \in \mathbb{N})$. Observe that assumptions (i) and (ii) imply that $V(\mathcal{G}^{\textnormal{sbr}}_T) = S_T$ for all~$T$.
Further, (i), (ii) and (iii) together with the assumption that $\mathcal{U}^i_T$ is a sequence of single best-reply payoff functions ensure that $\mathcal{G}_{T}^{\textnormal{sbr}}$ is a sequence of single best-reply games. In particular, assumption (iii) also implies that any policy $\pi^{j}$ employed by player $j \neq i$ is DB($(T-1 : T \in \mathbb{N})$) consistent for the payoffs~$((\mathcal{G}^{\textnormal{sbr}}_T)^j : T \in \mathbb{N})$. Finally, via assumption (iv), we can infer the action played by the other players by exploiting the injectivity of $u^j$. Therefore, for each $T$ and $j \neq i$, we can implement the strategies $\bar{\sigma}^j$ (that have knowledge of history of actions played before) via some policy $\bar{\pi}^{j}$ that does not have direct access to the previous history of actions. Denote the corresponding profile by $\bar{\pi}^{-i} = (\bar{\pi}^{j})_{j \neq i}$. By using \eqref{cond 1}, for each $(x_t) \in \mathcal{S}^i(S_T)$, we have
\begin{equation}\label{cond 2}
    \mathbb{E}^{\pi^i,\bar{\pi}^{-i}}\left[ \textnormal{Reg}^i((\mathcal{G}_{T}^{\textnormal{sbr}})^i,(x_t),(a_t))\right] = \mathbb{E}^{\pi^i,\bar{\sigma}^{-i}}\left[ \textnormal{Reg}^i(\mathcal{U}_T^{i},(x_t),(a_t))\right] = \Omega(T),
\end{equation}

that is,  $\limsup_{T \to \infty} \frac{1}{T} \max_{(x_t) \in \mathcal{S}^{i}(S_T)} \mathbb{E}^{\pi^i,\bar{\pi}^{-i}}\left[ \textnormal{Reg}^i((\mathcal{G}_{T}^{\textnormal{sbr}})^i,(x_t),(a_t))\right] >0.$ To shorten notation, denote by $L$ the value of this $\limsup$. By using Condition \eqref{cond 2} and the fact that payoffs are bounded, we conclude that there exists~$(T_r)_{r \in \mathbb{N}}$ (a subsequence of possible horizons) such that
\[ \frac{1}{T_r} \max_{(x_t) \in \mathcal{S}^{i}(S_T)} \mathbb{E}^{\pi^i,\bar{\pi}^{-i}}\left[ \textnormal{Reg}^i((\mathcal{G}_{T_r}^{\textnormal{sbr}})^i,(x_t)_{t=1}^{T_r},(a_t)_{t=1}^{T_r})\right] \rightarrow L, \quad \textnormal{as $r \to \infty$}.\]

Let $R \in (0,L)$. Then, there is $k(R) \geq 1$, such that, for all $r \geq k(R)$, it holds

\[ \frac{1}{T_r} \max_{(x_t) \in \mathcal{S}^{i}(S_{T_r})} \mathbb{E}^{\pi^i,\bar{\pi}^{-i}}\left[ \textnormal{Reg}^i((\mathcal{G}_{T_r}^{\textnormal{sbr}})^i,(x_t)_{t=1}^{T_r},(a_t)_{t=1}^{T_r})\right] \geq R.\]

Fix $r \geq k(R)$ and the horizon length $T_r$. Denote the batches where the stage game remains constant by $(\mathcal{T}_{(k)})_{k \in [V_T+1]}$, and denote the corresponding Hannan set and payoff of player $i$ by~$\mathcal{H}_k$ and $u_{(k)}^i$. Similarly, we denote by $\bar{\delta}^{(\pi^i,\bar{\pi}^{-i})}_{(k)}$ the distribution of play induced by $(\pi^i,\bar{\pi}^{-i})$ when restricted to $\mathcal{T}_{(k)}$. We also denote by 
$$\textnormal{dev}_{(k)} := \left( u_{(k)}^i(d_{(k)}^i,a^{-i}) - u_{(k)}^i(a) \right)_{a \in A}$$
the vectors of possible deviations to the single best-reply (for player $i$) $d_{(k)}^i$ over the batch $\mathcal{T}_{(k)}$. Finally, for each batch $\mathcal{T}_{(k)}$, let 
\[ q_{(k)} \in \textnormal{arg min}_{q \in \mathcal{H}_k} \textnormal{d}_p(\bar{\delta}^{(\pi^i,\bar{\pi}^{-i})}_{(k)}, \mathcal{H}_k),\]

be any distribution in the $k$-th Hannan set with minimum distance from the restricted distribution of play $\bar{\delta}^{(\pi^i,\bar{\pi}^{-i})}_{(k)}$. Let $Z := |A|$. The following holds 
\begin{equation}\label{linear tracking error}
    \begin{aligned}
    T_r R &\overset{(a)}{\leq} 
    \sum_{k = 1}^{S_{T_r} + 1} \sum_{t \in \mathcal{T}_{(k)}} \mathbb{E}^{\pi^i,\bar{\pi}^{-i}} \left[ u^i_t(d_{(k)}^i,a_{t}^{-i}) - u^i_t(a_t) \right] \\
    &= \sum_{k = 1}^{S_{T_r} + 1} |\mathcal{T}_{(k)}| \sum_{a \in A} \left( u_{(k)}^i(d_{(k)}^i,a^{-i}) - u_{(k)}^i(a) \right) \bar{\delta}^{(\pi^i,\bar{\pi}^{-i})}_{(k)}(a)  \\
    &= \sum_{k = 1}^{S_{T_r} + 1} |\mathcal{T}_{(k)}|\textnormal{dev}_{(k)}^\top \bar{\delta}^{(\pi^i,\bar{\pi}^{-i})}_{(k)}\\
    &\overset{(b)}{\leq} \sum_{k = 1}^{S_{T_r} + 1} |\mathcal{T}_{(k)}| \textnormal{dev}_{(k)}^\top \left( \bar{\delta}^{(\pi^i,\bar{\pi}^{-i})}_{k} - q_{(k)} \right) \\
    &\overset{(c)}{\leq} \sum_{k = 1}^{S_{T_r} + 1} |\mathcal{T}_{(k)}| \| \textnormal{dev}_{(k)} \|_{q} \|\bar{\delta}^{(\pi^i,\bar{\pi}^{-i})}_{(k)} - q_{(k)}\|_p \\
    &\overset{(d)}{\leq} 2 M Z^{1/q} \sum_{k = 1}^{S_{T_r} + 1} |\mathcal{T}_{(k)}| \|\bar{\delta}^{(\pi^i,\bar{\pi}^{-i})}_{(k)} - q_{(k)}\|_p \\
    &\overset{(e)}{=} 2 M Z^{1/q} \sum_{k = 1}^{V_{T_r} + 1} |\mathcal{T}_{(k)}| \textnormal{d}_p\left(\bar{\delta}^{(\pi^i,\bar{\pi}^{-i})}_{(k)}, \mathcal{H}_k \right) = 2M Z^{1/q} \; err_p(\mathcal{G}_{T_r}^{\textnormal{sbr}},(\pi^i,\bar{\pi}^{-i}),\mathcal{H}),
\end{aligned}
\end{equation}

where ($a$) holds by recalling that $d_{(k)}^i$ is the single best-reply for player $i$ over the $k$-th batch $\mathcal{T}_{(k)}$. Therefore, for the dynamic benchmark (with $S_{T_r}$ action changes), playing $x_t = d_{(k)}^i$, when $t \in \mathcal{T}_{(k)}$ maximizes player $i$'s cumulative regret. $(b)$ holds because $q_{(k)} \in \mathcal{H}_k$, so that $\textnormal{dev}_{(k)}^\top q_{(k)} \leq 0$, by definition of Hannan equilibrium. $(c)$ holds by applying H\"{o}lder's inequality. ($d$) holds by recalling that $u_{(k)}^i \in [-M,M]$, for $k \in \llbracket V_T + 1 \rrbracket$, and ($e$) holds by definition of $q_{(k)}$. 

The proof of the proposition is concluded by observing that \eqref{linear tracking error} implies that
\[ \underset{T \to \infty}{\limsup} \; \frac{err_p(\mathcal{G}^{\textnormal{sbr}}_T,(\pi^i,\bar{\pi}^{-i}),\mathcal{H})}{T} \geq \frac{R}{2 M Z^{1/q}} > 0.\]
\end{proof}

\subsection{Linear tracking error of Hannan consistent policies}\label{Simulations for Example 2}

In the context of Example \ref{pricing game example}, we illustrate that Hannan consistent policies can lead to linear tracking error even when the other players are playing their single best-reply actions at each time. Recall that in Example \ref{pricing game example} there is only on stage game change, that is, $V_T = 1$. We assume that seller one is using Exp3 (the prototypical Hannan consistent policy), while seller two always plays his dominant actions. We refer to player two's policy by $\pi^2$.

\begin{figure}[H]
    \centering
    \begin{minipage}[t]{0.48\textwidth}  
        \centering
        \includegraphics[width=\linewidth,height=5cm]{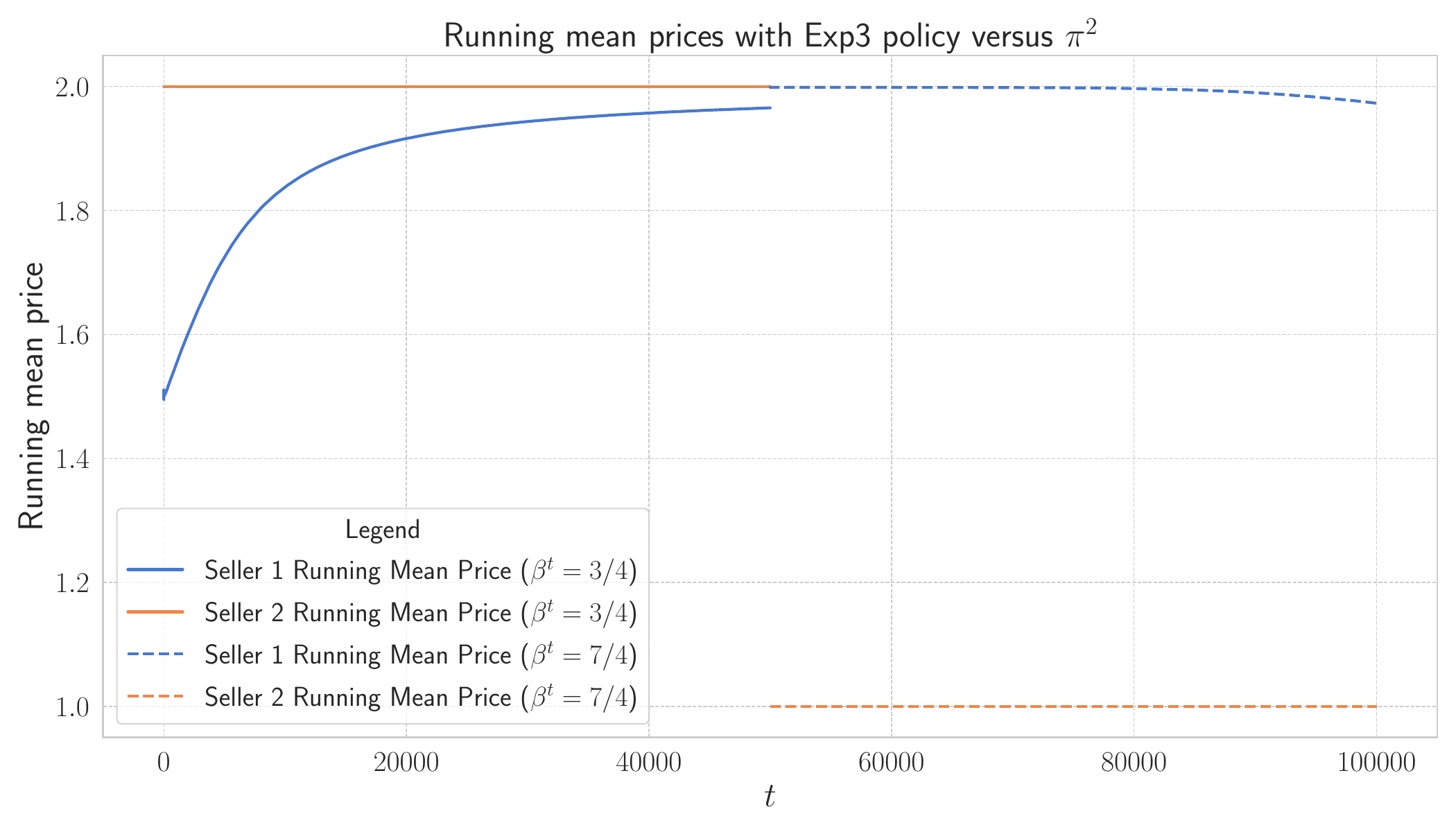}
        \caption*{(a) \footnotesize Running mean prices for the two sellers (restricted to each time batch where the game is constant) when $T =~10^5$.}
        \label{fig:rmp_Exp3}
    \end{minipage}
    \hfill
    \begin{minipage}[t]{0.48\textwidth}  
        \centering
        \includegraphics[width=\linewidth,height=5cm]{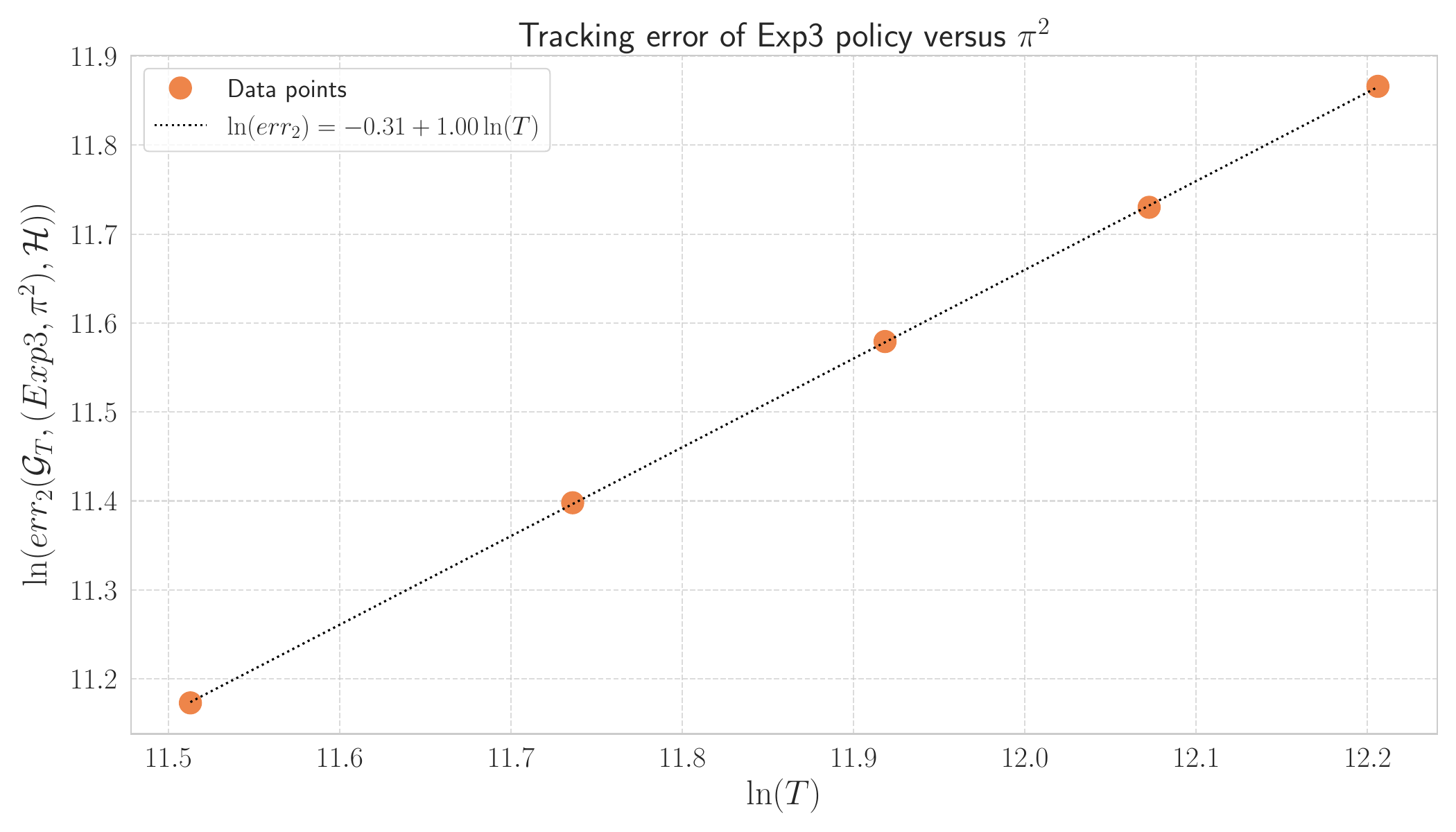}
        \caption*{(b) \footnotesize Tracking error as a function of the horizon length $T$.}
        \label{fig:tracking_Exp3}
    \end{minipage}
    \caption{\footnotesize Pricing game of Example \ref{pricing game example} with $N=1$. While seller two always play his dominant actions, seller one employs the standard Exp3 policy with an exploration parameter $\gamma_T = \sqrt{2 \ln(2)/((e-1)T)}$. The distribution of play is estimated by averaging the results over $750$ simulations. We recall that the Hannan sets are $\mathcal{H}(\Gamma_1) = \{\delta_{p_h,p_h}\}$ and $\mathcal{H}(\Gamma_2) = \{\delta_{p_l,p_l}\}$, where $p_h = 2$ and $p_l = 1$.}
    \label{fig:comparison_Exp3}
\end{figure}

Observe that, in Example \ref{pricing game example}, an increase in the price sensitivity of customers leads to a change in the optimal way to play the game: While with a low price sensitivity ($\beta^t = 3/4$) charging $p_h = 2$ is the dominant action for both sellers, when~$\beta^t$ increase to $7/4$, $p_l = 1$ becomes the dominant price. Thus, a dynamic benchmark that allows one action change can indeed play the optimal sequence of actions that can be selected in hindsight. On the other hand, static prices are sub-optimal in at least one batch in which the stage game is constant, and their performance gap relative to a dynamic benchmark that allows one action change grows linearly with $T$. Notably, the behavior of the Exp3 policy, designed to guarantee only the performance of the best static action, reflects this performance gap (see Figure \ref{fig:comparison_Exp3} (a)). While, on average, seller one's price approaches the dominant $p_h$ in the first batch, as~$p_h$ is already the single best action that can be selected in hindsight, it adapts to the evolved market conditions in the second batch too slowly. As a result, the tracking error is $\Theta(T)$ (see Figure \ref{fig:comparison_Exp3} (b)) despite seller two consistently playing his dominant actions, thereby behaving ``compatibly'' with the unique equilibrium in the Hannan sets.

\subsection{Proofs of \S \ref{subsect: general-sum games}}\label{proofs time-varying general sum-games}

\begin{proof}[Proof of Theorem \ref{thm: restarting = sublinear track error}]

Fix $T \in \mathbb{N}$ and let  

\[ V(\mathcal{G}_T) = \sum_{t=1}^{T-1} \mathbf{1}(\Gamma_t \neq \Gamma_{t+1})\]
\noindent
denoting the variation of the sequence of games $\mathcal{G}_{T} = (\Gamma_t : t \in \llbracket T \rrbracket)$. To ease the exposition, we abuse notation, and write the more compact $V_T$ instead of $V(\mathcal{G}_T)$.

We consider the partition of the horizon of play induced by the intervals $(\mathcal{T}_{(j)})_{j=1}^{V_T+1}$ where the stage game is constant. Formally, let $t_0 := 0$. We  define recursively,~$t_k := \max\{ t \in \llbracket T \rrbracket : \Gamma_t = \Gamma_{t_{k-1}+1}\}$ and $\mathcal{T}_{(k)} := \llbracket t_{k-1}+1,t_{k} \rrbracket$, for $k \in \llbracket V_T+1 \rrbracket$. For each batch $\mathcal{T}_{(k)}$, we denote the regret for player $i$, relative to the single best action in hindsight, over that batch by

\[ \textnormal{R}_k^i := \max_{x \in A_i} \; \mathbb{E}^{\pi} \left[ \sum_{t \in \mathcal{T}_{(k)}} u_{(k)}^{i}(x,a^{-i}_t) - u_{(k)}^{i}(a_t) \right].\]
\noindent
Similarly, we denote by $\textnormal{R}_{k} := \max_{i \in \llbracket N \rrbracket} \textnormal{R}_{k}^i$ the largest regret over the $k$-th batch. Then, we have

\begin{equation}
    \begin{aligned}
    err_p(\mathcal{G}_T,\pi, \mathcal{H}) = \sum_{k = 1}^{V_T+1} |\mathcal{T}_{(k)}| \textnormal{d}_{p}(\bar{\delta}^{\pi}_{(k)}, \mathcal{H}_k) 
    &\overset{(a)}{=} \sum_{k : R_{k} > 0} |\mathcal{T}_{(k)}| \textnormal{d}_{p}(\bar{\delta}^{\pi}_{(k)}, \mathcal{H}_k) \\
    &\overset{(b)}{\leq} \sum_{k : R_{k} > 0} \textnormal{const}(\Gamma_{(k)}) R_{k} \\
    &\leq \textnormal{const}(\mathcal{G}) \sum_{k : R_{k} > 0} \sum_{i \in \llbracket N \rrbracket} R^i_{k} \overset{(c)}{=} o(T),
    \end{aligned}
\end{equation}

where ($a$) holds because if $R_{k} \leq 0$ then $\bar{\delta}^{\pi}_{(k)} \in \mathcal{H}_{k}$, and so $\textnormal{d}_{p}(\bar{\delta}^{\pi}_{(k)}, \mathcal{H}_k)  = 0$. ($b$) holds by applying Lemma \ref{lemma: regret-distance} with $\epsilon = R_k/|\mathcal{T}_{(k)}|$. Finally, ($c$) holds because by assumption, for each player $i$,

 \[\sum_{k=1}^{V(\mathcal{G}_T) + 1} \max_{x \in A^i} \mathbb{E}^{\pi} \left[ \sum_{t \in \mathcal{T}_{(k)}}  \left(u_{(k)}^{i}(x,a^{-i}_t) - u^{i}_{(k)}(a_t^{-i}) \right)\right]_{+} = o(T).\]

We next show that both the restarted policies of Proposition \ref{prop: restarting = sublinear external-internal dynamic regret} and the Exp3S policy of \cite{auer2002nonstochastic} satisfy condition \ref{clipped dbc} in the statement of Theorem \ref{thm: restarting = sublinear track error}. Therefore, suppose that player $i$ is playing one of the following policies.

\paragraph{Restarted Hannan consistent policy.}{
By using Equation \eqref{split regret}, and Inequalities \eqref{split regret 1}, and \eqref{split regret 2} in the proof of Proposition \ref{prop: restarting = sublinear external-internal dynamic regret}, we obtain that for each $\mathcal{T}_{(k)}$, 
\[ \max_{x \in A^i}  \mathbb{E}^{\pi}\left[\sum_{t \in \mathcal{T}_{(k)}} u^i_{(k)}(x_t,a^{-i}_t) - u^i_{(k)}(a_t) \right] \leq |\mathcal{T}_{(k)}| \frac{g^i(\Delta^{i}_T)}{\Delta^{i}_T} + 4M \Delta^{i}_T.\]

In particular, note that in the left-hand-side, we can clip the expected value to be non-negative, and the inequality continues to hold. Finally, by summing over all batches, we obtain 

\[ \sum_{k = 1}^{V_T+1} \max_{x \in A^i}  \mathbb{E}^{\pi}\left[\sum_{t \in \mathcal{T}_{(k)}} u^i_{(k)}(x,a^{-i}_t) - u^i_{(k)}(a_t) \right]_{+} \leq T \frac{g^i(\Delta^{i}_T)}{\Delta^{i}_T} + 4M \Delta^{i}_T (V_T + 1),\]

where the right-hand-side is sub-linear if $\Delta^i_T = o(T/(V_T+1))$, and $\Delta^i_T \to \infty$, as $T \to \infty$.

}

\paragraph{Exp3S.}{The pseudo-code of the Exp3S policy is presented in Appendix \ref{structure of DB policies}. By using Equations (26) and (27) in \cite{auer2002nonstochastic}, we obtain that for each batch $\mathcal{T}_{(k)}$,

\[ \max_{x \in A^i}  \mathbb{E}^{\pi}\left[\sum_{t \in \mathcal{T}_{(k)}} u^i_{(k)}(x_t,a^{-i}_t) - u^i_{(k)}(a_t) \right] \leq (e-1)\gamma_T |\mathcal{T}_{(k)}| + \frac{K^i \ln(K^i / \alpha_T)}{\gamma_T} + \frac{e \alpha_T K^i |\mathcal{T}_{(k)}|}{\gamma_T}.\]

}

Again, note that in the left-hand-side, we can clip the expected value to be non-negative, and the inequality continues to hold. Therefore, by summing over all batches, we obtain 

\[ \sum_{k = 1}^{V_T+1} \max_{x \in A^i}  \mathbb{E}^{\pi}\left[\sum_{t \in \mathcal{T}_{(k)}} u^i_{(k)}(x,a^{-i}_t) - u^i_{(k)}(a_t) \right]_{+} \leq \frac{K^i((V_T+1)\ln(K^i/\alpha_T) + e \alpha_T T)}{\gamma_T} + (e-1)\gamma_T T.\]

Finally, observe that if $\gamma_T = o(1)$, $\alpha_T/\gamma_T = o(1)$, and $\ln(K^i /  \alpha_T)/\gamma_T = o(T/(V_T+1))$, then the right-hand-side is sub-linear.

\end{proof}

\section{Proofs of \S \ref{PoA Section}}\label{proofs of PoA section }

\subsection{Proofs of \S \ref{Welfare in repeated games}}\label{Proofs of welfare in repeated games}

\begin{proof}[Proof of Corollary \ref{Cor:1}]

Fix any profile of DB($C$) consistent policies $(\pi^i)_{i =1}^{N} \in \mathcal{P}^N_{DB}(C)$. Recall that $A$ is endowed with the lexicographical order. We rewrite the average social welfare up to time $T$ as 

\begin{equation*}
    \begin{aligned}
        \frac{1}{T} \sum_{t=1}^{T} \mathbb{E}^{\pi}\left[W(a_t) \right] &= \frac{1}{T} \sum_{t=1}^{T} \sum_{a \in A} W(a)  \mathbb{E}\left[ \delta_{a_t}(a) \right] \\
        &= \sum_{a \in A} W(a)  \mathbb{E}^{\pi} \left[\frac{1}{T} \sum_{t=1}^{T} \delta_{a_t}(a) \right] = w^\top q_T,
    \end{aligned}
\end{equation*}
\noindent
where we are denoting $w := \left(W(a)\right)_{a \in A}$ and $q_{T} := \left( \mathbb{E}^{\pi}\left[ \frac{1}{T} \sum_{t=1}^{T} \delta_{a_t}(a) \right] \right)_{a \in A}$.

\noindent
Define $\Pi(\Gamma) := \{w^\top q \mid q \in \mathcal{H}\}.$ Note that $\Pi(\Gamma)$ is closed and bounded as $\mathcal{H}$ is closed and bounded, and $w^\top q$ is continuous. By applying H\"older's inequality and Theorem \ref{Hannan Set Survives}, we obtain 
\begin{equation}\label{conv}
0 \leq \textnormal{d}_1(w^{\top}q_T,\Pi(\Gamma)) = \underset{q \in \mathcal{H}}{\text{min}} |w^\top q_{T} - w^\top q| \leq \| w\|_{\infty} \underset{q \in \mathcal{H}}{\text{min}} \|q_{T} -  q\|_{1} = \| w\|_{\infty} \text{d}_1(q_{T},\mathcal{H}) \longrightarrow 0.
\end{equation}

\noindent
As $\Pi(\Gamma)$ is compact, \eqref{conv} implies that  any convergent sub-sequence $\left( w^\top q_{T_K} \right)_{K \geq 1}$ converges to $\underset{a \sim q}{\mathbb{E}}[W(a)]$ for some~distribution $q \in \mathcal{H}$. Thus, we have

\begin{equation*}
\underset{T \rightarrow +\infty}{\text{liminf}} \; \frac{1}{T} \sum_{t=1}^{T} \mathbb{E}^{\pi} \left[W(a_t) \right] \geq \min_{q \in \mathcal{H}} \mathbb{E}_{a \sim q} \left[  W(a) \right], 
\end{equation*}

As $\pi$ was arbitrary, we conclude

\begin{equation}\label{easy version}
 \underset{(\pi^i)_{i =1}^{N} \in \mathcal{P}_{DB}^N(C)}{\text{inf}} \; \underset{T \rightarrow +\infty}{\text{liminf}} \; \frac{1}{T} \sum_{t=1}^{T} \mathbb{E}^{\pi} \left[W(a_t) \right] \geq \min_{q \in \mathcal{H}} \mathbb{E}_{a \sim q} \left[  W(a) \right]. 
\end{equation}

\noindent
On the other hand, 
\begin{equation}\label{hard version}
    \begin{aligned}
    \underset{(\pi^i)_{i =1}^{N} \in \mathcal{P}_{DB}^N(C)}{\text{inf}} \; \underset{T \rightarrow +\infty}{\text{liminf}} \; \frac{1}{T} \sum_{t=1}^{T} \mathbb{E}^{\pi} \left[W(a_t) \right] &\overset{(a)}{\leq}   \inf_{q \in \mathcal{H} \cap \Delta(A)_{\mathbb{Q}}} \mathbb{E}_{a \sim q} \left[W(a) \right] \\
    &\overset{(b)}{=} \min_{q \in \mathcal{H}} \mathbb{E}_{a \sim q} \left[W(a) \right],
    \end{aligned}
\end{equation}
\noindent
where ($a$) holds because any rational Hannan equilibrium can be supported as limit of the empirical distribution of play generated by some trigger-type policies, which are DB($C$) consistent by Proposition \ref{Direction 1}. The existence of a rational Hannan equilibrium is guaranteed by our assumption that the payoff functions have rational values. ($b$) holds by using the following argument. First, as the welfare is linear in $q$ and $\mathcal{H}$ is a polytope, the minimum is achieved at some vertex of $\mathcal{H}$. Further, our assumption that the payoff functions have rational values implies that the vertices of $\mathcal{H}$ have rational coordinates. Therefore, by combining equations (\ref{easy version}) and (\ref{hard version}), we obtain
\[ \underset{(\pi^i)_{i =1}^{N} \in \mathcal{P}_{DB}^N(C)}{\text{inf}} \; \underset{T \rightarrow +\infty}{\text{liminf}} \; \frac{1}{T} \sum_{t=1}^{T} \mathbb{E}^{\pi} \left[W(a_t) \right] = \min_{q \in \mathcal{H}} \mathbb{E}_{a \sim q} \left[W(a) \right].\]
\end{proof}

\subsection{Proofs of \S \ref{sect Welfare in time-varying games}}\label{Proofs of welfare in time-varying games}

\begin{proof}[Proof of Corollary \ref{dyn poa without smoothness}]
First of all, observe that
\begin{equation}\label{trivial bound}
    \begin{aligned}
        \sum_{t=1}^{T} \mathbb{E}^{\pi}\left[ W_t (a) \right] &= \sum_{k=1}^{V_T + 1} |\mathcal{T}_{(k)}| \sum_{a \in A} W_{(k)}(a) \bar{\delta}^{\pi}_{(k)}(a) \pm \sum_{k=1}^{V_T + 1} |\mathcal{T}_{(k)}| \sum_{a \in A} W_{(k)}(a) q_{(k)}(a) \\
        &\geq \sum_{k=1}^{V_T + 1} |\mathcal{T}_{(k)}| \min_{q \in \mathcal{H}(\Gamma_{(k)})} \; \underset{a \sim q}{\mathbb{E}}\left[ W_{(k)}(a)\right] + \underbrace{\sum_{k=1}^{V_T + 1} |\mathcal{T}_{(k)}| \sum_{a \in A} W_{(k)}(a) \left( \bar{\delta}^{\pi}_{(k)}(a) - q_{(k)}(a)\right)}_{\textnormal{Werr}_T},
    \end{aligned}
\end{equation}

where, for each $k \in \llbracket V_T + 1\rrbracket$, $q_{(k)} \in \argmin \{ q \in \mathcal{H}(\Gamma_{(k)}) : \|q - \bar{\delta}^{\pi}_{(k)}\|_p \}$. Divide both sides of \eqref{trivial bound} by~$T$. Then, by applying Theorem \ref{thm: restarting = sublinear track error} together with our assumption that $\sup_{T} \max_{W_t \in \mathcal{W}_{T}} \| W_t \|_{p} \leq \bar{W}$, we obtain that $\textnormal{Werr}_T/T \to 0$, as $T \to \infty$. This establishes that 

\[ \liminf_{T \to \infty} \frac{1}{T} \sum_{t=1}^{T} \mathbb{E}^{\pi}\left[ W_t (a) \right] \geq \liminf_{T \to \infty} \frac{1}{T} \sum_{t=1}^{T} \min_{q \in \mathcal{H}(\Gamma_{t})} \; \underset{a \sim q}{\mathbb{E}}\left[ W_{t}(a)\right].\]

After re-arranging terms, the result follows.

\end{proof}

\begin{proof}[Proof of Proposition \ref{DynPoA}]
Let $\pi$ be any profile of DB($C$) consistent policies, and let $(x_t: t \in \llbracket T \rrbracket) \in A^T$ be such that, for $i \in \llbracket N \rrbracket$, $(x_t^i: t \in \llbracket T \rrbracket) \in \mathcal{S}^i(C_T)$. Further, let $g(T) = \max_{i \in \llbracket N \rrbracket}g^i(T)$. We have
\begin{equation}\label{dynamic lambda-mu eq}
    \begin{aligned}
    \sum_{i = 1}^{N} \sum_{t=1}^{T} \mathbb{E}^{\pi}\left[ u^{i}_{t}(a_{t})\right] &\overset{(a)}{\geq} \sum_{i = 1}^{N} \sum_{t=1}^{T} \mathbb{E}^{\pi}\left[ u^{i}_{t}(x^i_t,a_{t}^{-i})\right] - N g(T) \\
    &\overset{(b)}{\geq} \lambda \sum_{t=1}^{T}  W_t(x_t) -\mu  \sum_{t=1}^{T} \mathbb{E}^{\pi}\left[ W_t(a_{t})\right] - N g(T),
    \end{aligned}
\end{equation}
\noindent
where ($a$) holds because $\pi^i$ is DB($C$) consistent for each player $i$, and ($b$) holds because, by assumption, $\Gamma_t$ is ($\lambda,\mu$)-smooth for all $t \in \llbracket T \rrbracket$.

As the sequence $(x_t)_{t=1}^{T}$ was arbitrary, by using \eqref{dynamic lambda-mu eq}, and the assumptions that 
\[ \max_{(x_t)_{t=1}^{T} \in A^{T}: \forall i, \; (x^i_t)_{t=1}^{T} \in \mathcal{S}^i(C_T)} \sum_{t=1}^{T} W_t(x_t) = \beta(C_T) \sum_{t=1}^{T} \max_{x \in A} W_t(x),\]
\noindent
and
\[ W_t(a) \geq \sum_{i=1}^{N}  u^{i}_{t}(a), \quad \textnormal{for all $a$ and $t$},\]

after re-arranging terms we obtain 
\[ \sum_{t=1}^{T} \mathbb{E}^{\pi}\left[ W_t(a_{t})\right] \geq \frac{\lambda \beta(C_T)}{1+\mu} \sum_{t=1}^{T} \max_{x_t \in A} W_t(x_t) - \frac{N}{1+\mu} g(T),\]
\noindent
where $g(T) = o(T)$. 
\end{proof}

\section{Proofs of \S \ref{no-internal regret section}}\label{proofs of no-internal regret section}

\begin{proposition}\label{restarting proposition for internal dbc}
Let $\mathcal{A}^i$ be a IDB(0) policy for all payoff functions. Denote by~$\pi^i$ the policy obtained by restarting $\mathcal{A}^i$ every~$\Delta_T$ periods when the horizon has length $T$. If $C_T = o(T)$, then, for any~$\Delta_T$ such that~$\Delta_T = o(T/(C_T+1))$ and~$\Delta_T \rightarrow \infty$ as $T \rightarrow \infty$, we have that $\pi^i$ is internally dynamic benchmark consistent for all payoff functions with respect to the sequence~$C$.

\end{proposition}

\begin{proof}
Recall that for integers $1 \leq a \leq b$, we use the notation $\llbracket a,b\rrbracket := \{a,\ldots,b\}$ and $\llbracket a \rrbracket := \llbracket 1,a \rrbracket$. For convenience, we also define $\llbracket a,b\rrparen := \llbracket a,b-1 \rrbracket$. Let $\mathcal{A}^i$ be a IDB($0$) consistent policy. Then, there exists a sub-linear function $g$ such that, for all $T \geq 1$ and $\sigma^{-i}: \bigcup_{t \in \llbracket T \rrbracket} A^{t-1} \rightarrow \Delta(A^{-i})$,

\[ \max_{(x,y) \in A^i \times A^i} \mathbb{E}^{\pi^i,\sigma^{-i}} \left[ \sum_{t=1}^{T} \mathbf{1}(a^{i}_t = x) \left(u^{i}_{t}(y,a^{-i}_t) - u^{i}_{t}(x,a_t^{-i}) \right)\right] \leq  g(T). \]
\noindent
Fix $T \in \mathbb{N}$, $\sigma^{-i}: \bigcup_{t \in \llbracket T \rrbracket} A^{t-1} \rightarrow \Delta(A^{-i})$, $\boldsymbol{I} = (I_k)_{k=1}^{c} \in \mathcal{I}_{C_T}$, for some $c \in \llbracket C_T+1 \rrbracket$, and $((x_k,y_k))_{k=1}^{c} \in (A^i \times A^i)^c$. 

The restarting of $\mathcal{A}^i$ takes place over batches of size $\Delta_T$ and it is performed in total $m = \lceil T/\Delta_T \rceil$ times. We denote the batches identified via this restarting procedure by $(\mathcal{B}_i)_{i=1}^{m}$, where $\mathcal{B}_1 := \llbracket 1,\Delta_T \rrbracket$, $\mathcal{B}_2 := \llbracket \Delta_T+1,2\Delta_T \rrbracket$,$\ldots$, $\mathcal{B}_m := \llbracket(m-1)\Delta_T+1,\min \{m \Delta_T,T\} \rrbracket$.

Fix an interval $ I_{k} = \llbracket a_k, b_k \rrbracket \in \boldsymbol{I}$, for some $k \in \llbracket c \rrbracket$, with $a_k < b_k$. We separate the batches that intersect the interval $I_k$ into two groups: the ones that are fully included in $\llbracket a_k, b_k \rrparen$ and the ones that are not fully included in~$\llbracket a_k, b_k \rrparen$. Specifically, define $\mathcal{Q}_k := \{ j \in [m] : \mathcal{B}_{j} \subseteq \llbracket a_k, b_k \rrparen \}$ and $\mathcal{W}_k := \{ j \in [m] : \mathcal{B}_{j} \cap  \llbracket a_k, b_k \rrbracket \neq \emptyset, \textnormal{ but } \mathcal{B}_{j} \nsubseteq  \llbracket a_k, b_k \rrparen \}$. Then, it holds 

\begin{equation}\label{split regret}
    \begin{aligned}
        \mathbb{E}^{\pi^i,\sigma^{-i}} \left[ \sum_{t \in I_k} \mathbf{1}(a^{i}_t = x_k)\left(u^{i}_{t}(y_k,a^{-i}_t) - u^{i}_{t}(a_t) \right)\right] &= \underbrace{\sum_{j \in \mathcal{Q}_k} \mathbb{E}^{\mathcal{A}^i,\sigma^{-i}} \left[ \sum_{t \in \mathcal{B}_j} \mathbf{1}(a^{i}_t = x_k)\left(u^{i}_{t}(y_k,a^{-i}_t) - u^{i}_{t}(a_t) \right)\right]}_{Q_k}\\
        &+ \underbrace{\sum_{j \in \mathcal{W}_k} \mathbb{E}^{\pi^i,\sigma^{-i}} \left[ \sum_{t \in \mathcal{B}_j \cap I_k} \mathbf{1}(a^{i}_t = x_k)\left(u^{i}_{t}(y_k,a^{-i}_t) - u^{i}_{t}(a_t) \right)\right]}_{W_k},
    \end{aligned}
\end{equation}
\noindent
where the equality holds because, over each batch $\mathcal{B}_j$, $\pi^i$ coincides with $\mathcal{A}^i$. We proceed by bounding both terms $Q_k$ and $W_k$ separately.

\begin{equation}\label{split regret 1}
    Q_k \overset{(a)}{\leq} \sum_{j \in \mathcal{Q}_k} g(\Delta_T) \overset{(b)}{\leq}  \frac{|I_k|}{\Delta_T} g(\Delta_T),
\end{equation}

\noindent
where $(a)$ holds because $\mathcal{A}^i$ satisfies no-internal regret, and by observing that, by construction, $m \notin \mathcal{Q}_k$, which, in turn, implies that $|\mathcal{B}_j| = \Delta_T$, for all $j \in \mathcal{Q}_k$. $(b)$ holds because there can be at most~$\lfloor |I_k|/\Delta_T \rfloor$ batches of size $\Delta_T$ fully included in $I_k$. On the other hand, we have

\begin{equation}\label{split regret 2}
    W_k \overset{(a)}{\leq}  \sum_{j \in \mathcal{W}_k} \sum_{t \in \mathcal{B}_j \cap I_k} 2 M \overset{(b)}{\leq}  \sum_{j \in \mathcal{W}_k} 2 M \Delta_T \overset{(c)}{\leq} 4 M \Delta_T,
\end{equation}

where $(a)$ holds because $|u^i(a)| \leq M$, for each $a \in A$; $(b)$ holds because $|\mathcal{B}_j \cap I_k| \leq |\mathcal{B}_j| \leq \Delta_T$, and $(c)$ holds because there can be at most two batches that intersect $I_k$ but that are not fully included in $I_k$, that is $|\mathcal{W}_k| \leq 2$. Therefore, by using Equality $(\ref{split regret})$ and summing over~$k \in [c]$, we obtain 
\begin{equation}
    \begin{aligned}
        \sum_{k=1}^{c} \mathbb{E}^{\pi^i,\sigma^{-i}} \left[ \sum_{t \in I_k} \mathbf{1}(a^{i}_t = x_k)\left(u^{i}_{t}(y_k,a^{-i}_t) - u^{i}_{t}(a_t) \right)\right] 
        &\leq \frac{T}{\Delta_T}g(\Delta_T) + 4 cM \Delta_T  \\
        &\leq \frac{T}{\Delta_T}  g(\Delta_T) + 4 M \Delta_T (C_T+1).
    \end{aligned}
\end{equation}


As $\boldsymbol{I} = (I_k)_{k=1}^{c} \in \mathcal{I}_{C_T}$, for $c \in \llbracket C_T + 1 \rrbracket$, and $((x_k,y_k))_{k=1}^{c} \in (A^i \times A^i)^c$ were arbitrary, we have 

\[ \max_{\boldsymbol{I} \in \mathcal{I}_{C_T}} \sum_{k=1}^{|\boldsymbol{I}|} \max_{(x,y) \in A^{2}_i} \mathbb{E}^{\pi^i,\sigma^{-i}} \left[ \sum_{t \in I_k} \mathbf{1}(a^{i}_t = x)\left(u^{i}_{t}(y,a^{-i}_t) - u^{i}_{t}(a_t) \right)\right] \leq \frac{T}{\Delta_T}  g(\Delta_T) + 4 M \Delta_T (C_T+1).\]


Finally, assume that $g(T) = o(T)$ and $C_T = o(T)$. Let $\Delta_T$ such that $\Delta_T = o(T/(C_T+1))$ and~$\Delta_T \rightarrow \infty$ as $T \rightarrow \infty$. We conclude

\[ \frac{T}{\Delta_T}  g(\Delta_T) + 4 M \Delta_T (C_T+1) = o(T).\]

As $\sigma^{-i}$ was arbitrary, this concludes the proof.

\end{proof}

\begin{lemma}\label{lemma: regret-distance CE}
Let $\Gamma$ be a stage game with corresponding set of correlated equilibria $\mathcal{CE}$. Then, there exists a constant~$\textnormal{const}_1(\Gamma)$ (depending on the underlying stage game) such that for any~$\epsilon \geq 0$, $ T \in \mathbb{N}$, and any profile profile of policies $\pi \in \mathcal{P}$,
\[ \max_{i \in \llbracket N \rrbracket} \frac{1}{T} \sum_{t=1}^{T} \max_{(x,y) \in A^i \times A^i} \mathbb{E}^{\pi} \left[ \sum_{t=1}^{T} \mathbf{1}(a^{i}_t = x) \left(u^i(y,a^{-i}_t) - u^i(a_t) \right)\right] \leq \epsilon \implies \textnormal{d}_p(\bar{\delta}^{\pi}_{T},\mathcal{CE}) \leq \textnormal{const}_1(\Gamma) \epsilon .\]
\end{lemma}

\begin{proof}
The proof is analogous to the one of Lemma \ref{lemma: regret-distance} (as also the set of correlated equilbiria is a polytope) and therefore omitted.
\end{proof}

\begin{proof}[Proof of Theorem \ref{Correlated equilibria Survive}]
The proof follows the same strategy we used to establish Theorems \ref{Direction 1} and \ref{Hannan Set Survives}. First, we can show that any distribution of play generated by a profile of IDB($\tilde{C}$)-consistent policies can be approximated by some trigger-type policies that are IDB($C$)-consistent. Then, by using these trigger-type policies to approximate distributions in the set of correlated equilibria and using the results in \cite{HMSC1}, the result follows. We omit the details.
\end{proof}

\begin{proof}[Proof of Theorem \ref{thm: restarting = sublinear track error for correlated equilibria}]
The proof follows the same steps we used to establish Theorem \ref{thm: restarting = sublinear track error}, with the only difference that it builds on Lemma \ref{lemma: regret-distance CE}. Therefore, we omit it.
\end{proof}

\section{Discussion on the structure of DB-consistent policies}\label{structure of DB policies}

Broadly speaking, in order to guarantee the performance of a dynamic sequence of actions selected in hindsight, a policy has to identify effective actions but also keep certain flexibility that would allow transition among them. There are two algorithmic approached that can be leveraged. 

The first approach relies on ideas developed in the tracking the best expert setting (e.g., \citealt{HERBSTER1995286}), and then extended to the bandit setting by \cite{auer2002nonstochastic}. Classical Hannan consistent policies maintain a sampling distribution over actions based on weights that can be interpreted as an expression of the decision maker's confidence about the identity of the best action. By following this structure, they tend to concentrate significant probability mass on a single action, thereby limiting the selection of other actions. A decision maker can increase the flexibility of such policies by sharing part of the information gathered with each new sample to all the weights, and by that, guaranteeing that the probability of selecting each action is bounded away from zero. The prototypical policy that follows this structure is Exp3S.

\medskip
{\small
\begin{algorithm}[H]
\SetKwInput{Init}{Inputs}
 \Init{Exploration parameter $\gamma \in (0,1]$ and sharing factor $\alpha > 0$.}
 Initialize to one the weight for each action, i.e., $\omega_1(k) = 1$, for all $k \in \llbracket K^i \rrbracket$.\\
 \For{t = 1,2,...,T}{
 \vspace{0.5cm}
    \begin{enumerate}
    \item[1] Set, for all $k \in \llbracket K^i \rrbracket$,
    \[ p_t(k) = (1-\gamma) \frac{w_t(k)}{\sum_{j \in \llbracket K^i \rrbracket} w_t(j)} + \frac{\gamma}{K^i}\]
    \item[2] Draw $a^i_t \sim p_t$ and receive reward $\textnormal{Rew}_t(a^i_t) := u^i_t(a^i_t, a^{-i}_t)$, for some $a^{-i}_t \in A^{-i}$
    \item[3] For $k \in \llbracket K^i \rrbracket$,
    \[ \widehat{\textnormal{Rew}_t}(k) = \begin{cases}
        \textnormal{Rew}_t(k)/p_t(k) &\textnormal{if } k = a^i_t\\
        0 &\textnormal{otherwise},
    \end{cases}\]
    \item[4]  For $k \in \llbracket K^i \rrbracket$, update the weights as $w_{t+1}(k) = w_{t}(k) \exp\left(\frac{\gamma \widehat{\textnormal{Rew}_t}(k)}{K^i}\right) + \frac{e \alpha}{K^i} \sum_{j \in \llbracket K^i \rrbracket} w_t(j)$ 
    \end{enumerate}}\caption{Exp3S \citep{auer2002nonstochastic}}
\end{algorithm}
}

This algorithm has been proposed by \cite{auer2002nonstochastic} and analyzed in a context where the sequence of action changes $C$ does not scale with the horizon, that is, 
$C_T = \bar{C}$, for all $T \geq 1$. In particular, Corollary 8.3. in \cite{auer2002nonstochastic}, shows that the Exp3S policy with parametrization
\[ \gamma_T = \min\left\{1,\sqrt{\frac{K^i((\bar{C}+1) \ln(K^i T) + e)}{(e-1)T}}\right\} \textnormal{ and } \alpha_T = \frac{1}{T}\]
\noindent
has expected regret (adapted to our notation) that is at most 
\begin{equation}\label{auer bound}
 \max_{(x_t) \in \mathcal{S}^i(\bar{C})} \mathbb{E}^{\pi^i,\sigma^{-i}}\left[ \mathrm{Reg}^i(\mathcal{U}_T^i,(x_t),(a_t)) \right] \leq 2 \sqrt{e-1} \sqrt{K^i T((\bar{C}+1) \ln(K T) + e)},    
\end{equation}
\noindent
for all payoff functions $\mathcal{U}^i_T$ and all strategies of the other players $\sigma^{-i}: \cup_{t=1}^{T} A^{t-1} \to \Delta(A^{-i})$. Notably, the bound in \eqref{auer bound} becomes vacuous, by allowing $C$ to depend on $T$, whenever $C_T = \Omega(T/\ln(T))$. However, we can still achieve dynamic benchmark consistency relative to all sub-linear sequences of action changes by modifying the tuning of Corollary 8.3 in \cite{auer2002nonstochastic}. In particular, the following lemma holds.
\begin{lemma}
Let $C$ be such that $C_T = o(T)$, and assume that, for each $T \geq 1$, the Exp3S policy is run with sharing factor $\alpha_T = (C_T+1)/T$ and exploration parameter
\[ \gamma_T = \min \left\{1, \sqrt{ K^i \frac{(C_T+1)}{T} \ln \left(\frac{K^i T}{C_T+1}\right)} \right\}.\]
Then, for all payoff functions $\mathcal{U}^i$, horizon lengths $T$, and all (potentially correlated) strategies of the other players $\sigma^{-i}: \cup_{t=1}^{T} A^{t-1} \to \Delta(A^{-i})$, it holds
\[ \max_{(x_t) \in \mathcal{S}^i(C_T)} \mathbb{E}^{\pi^i,\sigma^{-i}}\left[ \mathrm{Reg}^i(\mathcal{U}_T^i,(x_t),(a_t)) \right] = O\left(\sqrt{K^i T (C_T+1) \ln\left(\frac{K^i T}{C_T+1}\right)}\right).\]
In particular, observe that $\frac{C_T+1}{T} \ln(\frac{T}{C_T + 1}) = o(1)$, so that Exp3S is DB($C$)-consistent.
\end{lemma}
\begin{proof}
By Theorem 8.1 in \cite{auer2002nonstochastic}, for any $T \geq 1$, $\gamma_T \in (0,1]$ and $\alpha_T > 0$, we have
\[ \max_{(x_t) \in \mathcal{S}^i(C_T)} \mathbb{E}^{\pi^i,\sigma^{-i}}\left[ \mathrm{Reg}^i(\mathcal{U}_T^i,(x_t),(a_t)) \right] \leq \frac{K^i\left((C_T+1) \ln(K^i/\alpha_T) + e \alpha_T T \right)}{\gamma_T} + (e-1) \gamma_T T.\]

Suppose that $T$ is large enough, so that $\gamma_T < 1$. Then, by using the parametrization in the statement of the lemma, we obtain 
\begin{equation}
\begin{aligned}
    \max_{(x_t) \in \mathcal{S}^i(C_T)} \mathbb{E}^{\pi^i,\sigma^{-i}}\left[ \mathrm{Reg}^i(\mathcal{U}_T^i,(x_t),(a_t)) \right] &\leq e\sqrt{K^i (C_T+1) T \ln\left(\frac{ K^i T}{(C_T+1)}\right)} + e \sqrt{ \frac{K^i (C_T+1) T}{\ln(K^i T/(C_T+1))}} \\
    &= O\left(\sqrt{K^i T (C_T+1) \ln\left(\frac{K^i T}{(C_T+1)}\right)}\right).
\end{aligned}
\end{equation}
Therefore, the Exp3S policy is DB($(C_T: T \geq 1)$)-consistent.
\end{proof}

We next propose a second approach to generate DB-consistent policies that leverages the structure of Hannan consistent policies, but increases the rate at which they perform exploration. In the non-stationary stochastic optimization literature (e.g., \citealt{besbes2015,besbes2019}), it has been shown that designing a schedule of time periods in which the policy will be restarted and all the observations that have been collected so far discarded is minimax rate optimal in a slowly changing environment, relative to a benchmark that selects the best action in hindsight in every period. In Proposition \ref{prop: restarting = sublinear external-internal dynamic regret}, we show that this restarting idea turns out to be effective in the complementary setting of this paper, where the changes in the environment (that is, the total number of changes in the payoffs an agent receive because of changes in the stage game and changes in the actions of the other players) are unconstrained, but the number of actions changes allowed in the benchmark is constrained. Our procedure can be applied to many Hannan consistent policies (for example, ones that are based on potential functions; see \citealt{HMSC2} and \citealt{cesa2006prediction}), which implies that the corresponding class of dynamic benchmark consistent policies is large, containing as many policies as the traditional no-regret class.
\begin{proposition}\label{prop: restarting = sublinear external-internal dynamic regret}
Let $\mathcal{A}^i$ be a Hannan consistent policy for all payoff functions. Denote by~$\pi^i$ the policy obtained by restarting $\mathcal{A}^i$ every~$\Delta_T$ periods when the horizon has length $T$. If $C_T = o(T)$, then, for any~$\Delta_T$ such that~$\Delta_T = o(T/(C_T+1))$ and~$\Delta_T \rightarrow \infty$ as $T \rightarrow \infty$, we have that $\pi^i$ is dynamic benchmark consistent for all payoff functions with respect to the sequence~$C$.

\end{proposition}

\begin{proof}
The proof is analogous to the one for Proposition \ref{restarting proposition for internal dbc} in Appendix \ref{proofs of no-internal regret section}, and therefore~omitted.
\end{proof}

\section{Discussion on dynamic benchmark consistency and tracking error}\label{dbc and tracking error}

The following example illustrates that not all dynamic benchmark consistent policies can lead to diminishing tracking error. In particular, we show for a sequence of games $\mathcal{G}_T$ that changes only once that the tracking error is linear even though all players are using DB-consistent policies.

Let $T \geq 8$, and consider the following two stage games

\captionsetup[table]{labelformat=empty}

\vspace{-0.1cm}
\begin{table}[H]
\centering
\begin{minipage}{.45\textwidth}
    \centering
    \setlength{\extrarowheight}{1pt}
    \begin{tabular}{cc|c|c|}
      & \multicolumn{1}{c}{} & \multicolumn{1}{c}{{c}}  & \multicolumn{1}{c}{d} \\\cline{3-4}
      & u & $1\hspace{0.1cm};\hspace{0.1cm}\epsilon$ & $-1\hspace{0.1cm};\hspace{0.1cm}\epsilon$ \\\cline{3-4}
      & b & $-1\hspace{0.1cm};\hspace{0.1cm}\epsilon$ & $1\hspace{0.1cm};\hspace{0.1cm} \epsilon$ \\\cline{3-4}
    \end{tabular}
    \caption{$\quad\quad\quad\Gamma_1$}
\end{minipage}%
\hspace{0.1cm}
\begin{minipage}{.45\textwidth}
    \centering
    \setlength{\extrarowheight}{1pt}
    \begin{tabular}{cc|c|c|}
      & \multicolumn{1}{c}{} & \multicolumn{1}{c}{{c}}  & \multicolumn{1}{c}{d} \\\cline{3-4}
      & u & $1/4\hspace{0.1cm};\hspace{0.1cm}\epsilon$ & $1/4\hspace{0.1cm};\hspace{0.1cm}\epsilon$ \\\cline{3-4}
      & b & $-1/4\hspace{0.1cm};\hspace{0.1cm}\epsilon$ & $-1/4\hspace{0.1cm};\hspace{0.1cm}\epsilon$ \\\cline{3-4}
    \end{tabular}
    \caption{$\quad\quad\quad\Gamma_2$}
\end{minipage}
\end{table}
\noindent
with player one being the row player, and player two being the column player. Suppose that $\Gamma_1$ is played for the first $2 \lceil T/4 \rceil$ periods and $\Gamma_2$ for the remaining periods. Observe that for the column player any policy is dynamic benchmark consistent, even relative to the best dynamic benchmark that can be selected in hindsight. Suppose that~$\pi_{2}$ (when the horizon has length $T$) plays actions accordingly
\begin{equation}
    a_{t}^2 \sim 
    \begin{cases}
    \delta_{\textnormal{c}}, & \text{if } t \leq \lceil T/4 \rceil \\
    \delta_{\textnormal{d}}, & \text{otherwise.}
    \end{cases}
\end{equation}

Let $u^{0}_{1} = 0$. Suppose that policy $\pi_{1}$ (when the horizon has length $T$) takes the form:
\begin{equation}
    a^{t}_1 \sim 
    \begin{cases}
    \delta_{\textnormal{u}}, & \text{if } t \leq \lceil T/4 \rceil \text{ and } u^{t-1}_1 \geq 0 \\
    \delta_{\textnormal{b}}, & \text{if } \lceil T/4 \rceil < t \leq 2 \lceil T/4 \rceil \text{ and } u^{t-1}_1 \geq 0 \\
    \delta_{\textnormal{b}}, & \text{if } t = 2 \lceil T/4 \rceil 
 +1 \text{ and } u^{t-1}_1 \geq 0 \\
 \delta_{\textnormal{b}}, & \text{if } t > 2 \lceil T/4 \rceil
 +1 \text{ and } u^{t-1}_1 \leq 0 \\
 \textnormal{Exp3S}(T-t+1), & \text{otherwise}, 
    \end{cases}
\end{equation}
\noindent
where $u^{t-1}_1$ is the realized payoff obtained by the row player at time $t-1$, and we use the notation Exp3S($T$) to indicate the Exp3S algorithm (see \citealt{auer2002nonstochastic}) specified over a horizon of length~$T$. In particular, \cite{auer2002nonstochastic} show that Exp3S($T$) is a DB($1$) consistent policy when tuned with parameters
\[ \alpha_T = \frac{1}{T} \textnormal{ and } \gamma_T = \left \{ 1, \sqrt{\frac{ 4 \log(2T) + 2 e}{(e-1)T}} \right \}.\]








We first show that $\pi_1$ is DB($1$) consistent. Let $\sigma_{2}$ be any strategy for player two, and denote by~$D_T$ the random time at which $\pi_{1}$ deviates to Exp3S when the horizon has length $T$. We understand that if $D_T = \infty$, then no deviation occurred. Let $(x_t) \in \mathcal{S}^{i}(1)$. There are two cases to be discussed. First, define 
\[ \mathbb{E}^{\pi}_{\infty}\left[ \sum_{t=1}^{T}u_t^1(x_t,a_{t}^2) - u_t^1(a_t) \right] := \mathbb{E}^{\pi}\left[ \sum_{t=1}^{T}u_t^1(x_t,a_{t}^2) - u_t^1(a_t) \; \Bigg \lvert \; D_T = \infty \right].\]

Then, no deviation to the Exp3S policy happened. Therefore, we have
\begin{equation*}
    \begin{aligned}
    \mathbb{E}^{\pi}_{\infty}\left[ \sum_{t=1}^{T}u_t^1(x_t,a_{t}^2) - u_t^1(a_t) \right] 
    &\overset{(a)}{=} \mathbb{E}^{\pi_1,\pi_2}\left[ \sum_{t=1}^{T}u_t^1(x_t,a_{t}^2)\right] - \left(2 \lceil T/4 \rceil - \frac{1}{4}(T - 2\lceil T/4 \rceil) \right) \\
    &\overset{(b)}{\leq} \left(2 \lceil T/4 \rceil - \frac{1}{4}(T - 2\lceil T/4 \rceil) \right) - \left(2 \lceil T/4 \rceil - \frac{1}{4}(T - 2\lceil T/4 \rceil) \right) = 0,
    \end{aligned}
\end{equation*}

where ($a$) holds because when $\pi_1$ and $\pi_2$ interact, player one obtains a payoff of $2 \lceil T/4 \rceil$ for the first $2 \lceil T/4 \rceil$ time periods, and of $-(1/4)(T-2 \lceil T/4 \rceil)$ for the remaining $T-2 \lceil T/4 \rceil$ periods. ($b$) holds by observing that, given the actions selected by $\pi_2$, the best dynamic sequence of actions with at most one action change that can be selected by the benchmark is to play $\textnormal{u}$ for the first $\lceil T/4 \rceil$ time periods, and then switching to $\textnormal{b}$ for the remaining ones.\footnote{Indeed, by playing $\textnormal{u}$ for the first $\lceil T/4 \rceil$ time periods, and then switching to $\textnormal{b}$ for the remaining ones, the benchmark accrues a cumulative payoff of $2 \lceil T/4 \rceil - \frac{1}{4}(T - 2\lceil T/4 \rceil) \geq 3T/8$. On the other hand, the second most promising sequence (with at most one action change) is to play $\textnormal{u}$ throughout the horizon, yielding a cumulative payoff of $(T-2\lceil T/4 \rceil)/4 \leq T/8$.   } Further, we have

\begin{equation*}
    \begin{aligned}
   \sum_{d = 1}^{T} \mathbb{E}^{\pi}\left[\sum_{t=1}^{T}u_t^1(x_t,a_{t}^2) - u_t^1(a_t) \mid D_T = d\right] \mathbb{P}^{\pi}(D_T = d)  &= \sum_{d = 1}^{T} \mathbb{E}^{\pi}\left[\sum_{t=1}^{d}u_t^1(x_t,a_{t}^2) - u_t^1(a_t) \mid D_T = d\right] \mathbb{P}^{\pi}(D_T = d)  \\
    &+ \sum_{d = 1}^{T} \mathbb{E}^{\pi}\left[\sum_{t=d+1}^{T}u_t^1(x_t,a_{t}^2) - u_t^1(a_t) \mid D_T = d\right] \mathbb{P}^{\pi}(D_T = d) \\
    &\overset{(a)}{\leq} \sum_{d = 1}^{T} 2 \cdot \mathbf{1}(d \leq 2 \lceil T/4 \rceil) \mathbb{P}^{\pi}(D_T = d)  \\
    &+ \sum_{d = 1}^{T} \mathbb{E}^{\pi}\left[\sum_{t=d+1}^{T}u_t^1(x_t,a_{t}^2) - u_t^1(a_t) \mid D_T = d\right] \mathbb{P}^{\pi}(D_T = d) \\
    &\leq 2 +  \sum_{d=1}^{T} \mathbb{E}^{\pi}\left[\sum_{t=d+1}^{T}u_t^1(x_t,a_{t}^2) - u_t^1(a_t) \mid D_T = d \right]\mathbb{P}^{\pi}(D_T = d) \\
    &\overset{(b)}{\leq} 2 + \sum_{d = 1}^{T} 2\sqrt{e-1} \sqrt{2(T-d)(2 \ln(2 (T-d)) + e)} \mathbb{P}^{\pi}(D_T = d) \\
    &\leq 2 + 2\sqrt{e-1} \sqrt{2T(2 \ln(2 T) + e)} = o(T),
    \end{aligned}
\end{equation*}
\noindent
where ($a$) holds by the following argument. Suppose $d \geq 2\lceil T / 4 \rceil + 1$. Then, $u^1(x,\textnormal{c}) = u^1(x,\textnormal{d})$ for~$x \in \{\textnormal{u},\textnormal{b}\}$, so that, by applying the same argument outlined above for the case $D_T = \infty$, we conclude that no sequence with at most one action change can accrue a strictly larger payoff than~$\pi_1$. On the other hand, if $d \leq 2\lceil T / 4 \rceil$, then the regret is at most $2$. Indeed, up to period $d-1$, the cumulative payoff achieved by $\pi^1$ is equal to $d-1$, which is the maximum that can be obtained. Finally, note that the regret incurred at period $d$ is at most equal to $2$. ($b$) holds because over $\llbracket d+1,T \rrbracket$, $\pi_{1}$ switches to Exp3S, 
and, over a horizon of length~$T-d$, the expected regret of Exp3S relative to the best dynamic sequence of actions with at most one action change can be bounded (see \citealt{auer2002nonstochastic}) by 
\[ 2\sqrt{e-1} \sqrt{2(T-d)(2 \ln(2 (T-d)) + e)}.\]

This shows that $\pi_1$ is DB($1$) consistent. To conclude the proof, denote the distribution of play (induced by $\pi$) restricted to the $k$-th batch by $\bar{\delta}^{\pi}_{(k)}$. We have 
\[ err_2(\mathcal{G}_T,\pi,\mathcal{H}) \geq (T-2\lceil T/4\rceil) \textnormal{d}_2(\bar{\delta}^{\pi}_2,\mathcal{H}(\Gamma_2)) \geq \left(\frac{T}{2} - 2 \right)\textnormal{d}_2(\bar{\delta}^{\pi}_2,\mathcal{H}(\Gamma_2)) = \sqrt{6} \frac{T-4}{4}, \]
\noindent
where the last step follows by observing that $\bar{\delta}^{\pi}_{(2)} = \delta_{\textnormal{b,d}}$, and 
$$\mathcal{H}(\Gamma_2) = \{ \alpha \delta_{\textnormal{u,c}} + (1-\alpha) \delta_{\textnormal{u,d}} : \alpha \in [0,1]\}.$$

\section{Almost-surely guarantees}\label{a.s. guarantess}

In this section, we adapt the definition of dynamic benchmark consistency to require almost surely performance guarantees and provide necessary and sufficient conditions for the existence of such a policy.

\begin{definition}
For some payoff functions $\mathcal{U}^i = (u^i_t : t \in \mathbb{N})$, we say that a policy $\pi^i = \cup_{t=1}^{\infty} (A^i \times [-M,M])^{t-1} \to \Delta(A^i)$ is dynamic benchmark consistent with respect to the sequence~$C$ if for all (potentially correlated) strategies of the other players $\sigma^{-i}: \cup_{t=1}^{\infty} A^{t-1} \to \Delta(A^{-i})$, 

    \[ 
\mathbb{P}^{\pi^i,\sigma^{-i}} \left(  \underset{T \rightarrow +\infty} \limsup \; \frac{1}{T} \max_{(x_t) \in \mathcal{S}^{i}(C_T)}\mathrm{Reg}^i (\mathcal{U}^{i}_T,(a_t),(x_t)) \leq 
    0  \right) = 1.
\]
\end{definition}

Then, the following proposition holds.

\begin{proposition}[Existence of DB-consistent strategies] \label{V Class is non-empty}
Let $C$ be any sequence of non-negative and non-decreasing real numbers. We have the following:\vspace{-0.1cm}
\begin{itemize}
	\item[(i)] If $C_T = o(T)$ then the set  $\mathcal{P}^{DB}_i(C)$ of dynamic benchmark consistent policies subject to the sequence~$C$ is non-empty. \vspace{-0.1cm}
	\item[(ii)] If $\underset{T \rightarrow \infty}{\limsup} \frac{C_T}{T} > 0$ then the set  $\mathcal{P}^{DB}_i(C)$ of dynamic benchmark consistent policies subject to the sequence $C$  is non-empty if and only if there exists a common best reply to all profiles of actions of the other players.
 
\end{itemize}
\end{proposition}

\renewcommand{\thetheorem}{A.\arabic{theorem}}
\renewcommand{\thecorollary}{A.\arabic{corollary}}
\renewcommand{\thelemma}{A.\arabic{lemma}}
\setcounter{theorem}{0}
\setcounter{corollary}{0}
\setcounter{lemma}{0}

\SetAlgorithmName{Algorithm}{}{}
To simplify the notation, we denote by $K$ the cardinality of the action set of player $i$, i.e. $K = |A_i|$. At high level, the proof modifies the analysis of the Exp3P policy in \cite{bubeck2010} to cover sequences of action changes that scale with the horizon length.

\medskip
{\small
\begin{algorithm}[H]
\SetKwInput{Init}{Inputs}
 \Init{$\eta \in \mathbb{R}_{+}$, $\gamma,\beta \in [0,1]$, length of horizon $T \in \mathbb{N}$.}
 Let $p_1$ be the uniform distribution over $A_i$.\\
 \For{t = 1,2,...,T}{
    \begin{enumerate}
    \item[1] Draw $a_t \sim p_t$ and receive reward $g_{a_{t},t} \in [-M,M]$
    \item[2] Compute $\tilde{g}_{k,t} = \frac{1}{2 M} \left( \frac{g_{k,t} \mathbf{1}(a_t = k) + \beta}{p_{k,t}} + M \right)$ for all $k \in A_i$
    \item[3] Update $\tilde{G}_{k,t} = \sum_{\tau = 1}^{t} \tilde{g}_{k,\tau}$ for all $k \in A_i$
    \item[4] Compute mixed action $p_{t+1}$, where 
    \[p_{k,t+1} = (1-\gamma)\frac{\exp (\eta \tilde{G}_{k,t})}{\sum_{j \in A_i} \exp ( \eta \tilde{G}_{j,t} )} + \frac{\gamma}{K} \quad \text{for all} \; k \in A_i\] 
    \end{enumerate}}\caption{Exp3P}
\end{algorithm}
}
\medskip 
\begin{theorem}[Theorem 22 in \citealt{bubeck2010}]\label{Bubeck Theorem}
Fix $T > 0$, $S \leq T-1$ and arbitrary $\delta > 0$. Set $s = S \log \left( \frac{3 T K}{S} \right) + 2\log(K)$ with the natural convention $S \log \left( \frac{3 T K}{S} \right) = 0$ for $S = 0$. Let $\beta = 3 \sqrt{\frac{s}{T K}}$, $\gamma = \min \left\{ \frac{1}{2} ,\sqrt{\frac{K s}{2 T}} \right\}$ and $\eta = \frac{1}{5} \sqrt{\frac{s}{T K}}$. Then, for all $\sigma^{-i}: \cup_{t=1}^{T} A^{t-1} \rightarrow \Delta(A^{-i})$,

\[
\mathbb{P}^{\mathrm{Exp3P},\sigma^{-i}}_T \left( \max_{(x_t) \in \mathcal{S}^{i}(S)}\mathrm{Reg}^i(\mathcal{U}_T^i,(x_t),(a_t)) \leq  7 \sqrt{T K s} + \sqrt{\frac{T K}{s}} \ln \left(\frac{1}{\delta} \right) \right) \geq 1-\delta,
\]
\noindent
where $\mathbb{P}^{\mathrm{Exp3P},\sigma^{-i}}_T$ is the probability measure over the space of sequences of length $T$ generated by the interaction of the finite-time \emph{Exp3P} policy and $\sigma$.

\end{theorem}

The above finite time high probability guarantee constitutes the building block for our successive Proposition \ref{V Class is non-empty}. We first establish Part $(i)$ of the theorem and then continue to prove Part $(ii)$.
\SetAlgorithmName{Algorithm}{}{}
\vspace{-2mm}
\noindent
\paragraph{Part $\boldsymbol{(i)}$.}{\SetAlgorithmName{Algorithm}{}{}

We show that in order to devise a dynamic benchmark consistent policy, it suffices to combine Exp3P with the so called \emph{doubling-trick}, by adopting an online worst-case approach. Specifically, we run Exp3P over an exponentially increasing pulls of periods~$(2^{r-1})_{r=1}^{+\infty}$. Each re-start identifies a batch, whose size is $2^{r-1}$. Further, Exp3P is tuned as if we are competing, over~$2^{r-1}$ periods, against a benchmark allowed to switch action at most $\min \{ C_{\sum_{j=1}^{r} 2^{j-1}} + 1, 2^{r-1}-1\}$ times, regardless of the (actual) number of changes spent in the previous batches. We refer to the corresponding policy as Rexp3P whose pseudocode is the following:
\medskip

\begin{algorithm}[H]\label{Rexp3P}
{\small
\SetKwInput{Init}{Inputs}
\SetKwInput{Par}{Parameters}
 \Init{A sequence $C$ of non-negative and non-decreasing real numbers.}
 \For{r = 1,2,...}{
    \begin{enumerate}
    \item[1] Set: $C^r = \min \{ C_{\sum_{j=1}^{r} 2^{j-1}}+1, 2^{r-1}-1 \} = \min \{ C_{2^{r}-1}+1, 2^{r-1}-1 \}$
    \item[2] Set: $c^r = C^r \ln \left( \frac{3 \times 2^{r-1} K}{C^r} \right) + 2\ln(K) $
    \end{enumerate}
    \If{ r = 1 }{ $a^{\mathrm{Rexp3P}}_1 \sim \mathrm{Uniform}(K)$}
    \Else{
    \For{$t \in \{ 2^{r-1},..., 2^{r}-1 \}$ }{
    run Exp3P with inputs $S = C^{r}$, $s = c^{r}$, $\eta = \frac{1}{5} \sqrt{\frac{c^r}{2^{r-1} K}}$, $\gamma = \min \left\{ \frac{1}{2}, \sqrt{\frac{K c^r}{2^{r}}} \right\}$,$\beta = 3\sqrt{\frac{c^r}{2^{r-1} K }}$ and $T = 2^{r-1}$
    }} 
 }
 \caption{Rexp3P}
 }
\end{algorithm}

\medskip
Without loss of generality,\footnote{If $C_T = o(T)$, extend $C$ by setting $C_0 = 0$. Further, let $\tilde{C}$ be the ``discrete'' \emph{concavification} of $C$, i.e. $\tilde{C}_T = \text{sup} \left\{ c \mid (T,c) \in \text{Conv}(C) \right\}$ for all $T \in \mathbb{N}$, where $\text{Conv}(C)$ is the convex hull in $\mathbb{R}_{+}^{2}$ of the \emph{hypograph} $\{(T,c) \mid c \leq C_{T}\}$.
Concavity of $\tilde{C}$ together with $\tilde{C}_0 = 0$ and sublinearity of the original $C$ imply that $\frac{\tilde{C}_T}{T} \searrow 0$ 
as $T \rightarrow +\infty$.} we assume that $C$ is monotonically decreasing to zero, i.e. ~$\frac{C_T}{T}  \searrow 0$ as $T \rightarrow +\infty$. 
Further, we let $\sigma^{-i}: \cup_{t=1}^{+\infty} A^{t-1} \rightarrow \Delta(A^{-i})$ be a contingent plan of actions by the other players.
The interaction of Rexp3P and $\sigma^{-i}$ induces the probability measure~$\mathbb{P}^{\mathrm{Rexp3P},\sigma^{-i}}$ over the space of infinite sequences of outcomes $A^{\mathbb{N}}$. Fix $T$ such that $C_{2^{\lceil \log_2(T+1) \rceil} -1 } + 1 \leq 2^{\lceil \log_2(T+1) \rceil - 1}-1$ and denote by 
\[
R := \min \left\{ n \in \mathbb{N} \;\Bigg\vert\;  T \leq \sum_{r=1}^{n} 2^{r-1} \right\} = \lceil \log_2(T+1) \rceil 
\]
the index of the first pull that reaches or goes above $T$. Let~$\delta \in (0,1)$ to be specified later. Let~$L_1 = U_1 = 1$, $L_r = 2^{r-1}$ and $U_r = 2^{r}-1$, for $r = 2,...,R$ be the starting and ending points of each pull. 

For a given sequence of outcomes $(a_t)_{t=1}^{+\infty} \in A^{\mathbb{N}}$, define, for each pull $r = 1,...,R$

\begin{equation}
    \mathrm{Reg}_r((a_t)_{t=1}^{+\infty},C^r) := \underset{(y_t)_{t} \in A_i^{2^{r-1}}}{\max} \left\{ \sum_{t=L_r}^{U_r} u^i_t(y_t,a^{-i}_{t}) - u^i_t(a^{i}_{t},a^{-i}_{t}) \; \Bigg\lvert \; \sum_{t = L_r}^{U_r-1} \mathbf{1}(y_{t+1} \neq y_t ) \leq C^{r} \right \},
\end{equation}
\noindent
where we understand, for the first pull, $\sum_{t = 1}^{0} \mathbf{1}(y_{t+1} \neq y_t ) = 0$ and similarly

\begin{equation}
    \mathrm{Reg}((a_t)_{t=1}^{+\infty},T,C_T) := \underset{(x_t)_{t} \in A_i^{T}}{\max} \left\{ \sum_{t=1}^{T} u^i_t(x_t,a^{-i}_{t}) - u^i_t(a^{i}_{t},a^{-i}_{t}) \; \Bigg\lvert \; \sum_{t =1}^{T-1} \mathbf{1}(x_{t+1} \neq x_t ) \leq C_T \right \}.    
\end{equation}


\noindent
 Further, let
 
\begin{equation}\label{Er set}
 E_r := \left\{ (a_t)_{t=1}^{+\infty} \in A^{\mathbb{N}} \; \Bigg \lvert \; \mathrm{Reg}_r((a_t)_{t=1}^{\infty},C^r) >  \underbrace{7 \sqrt{2^{r-1} K c^r} + \sqrt{\frac{2^{r-1} K }{c^r}}\ln\left( \frac{1}{\delta}\right)}_{:= b_r}  \right\}   
\end{equation}

\noindent
be the collection of all sequences of outcomes for which the regret over pull $r$, relative to a dynamic benchmark endowed with $C^{r}$ action changes, is at least $b_r$.\footnote{Recall from the pseudocode of Rexp3P that $c^r = C^r \ln \left( \frac{3 \times 2^{r-1} K}{C^r} \right) + 2\ln(K) $ and $C^r = \min \{ C_{\sum_{j=1}^{r} 2^{j-1}} + 1, 2^{r-1}-1 \}.$
   } Theorem \ref{Bubeck Theorem} with a union bound over the pulls imply that,
   
\begin{equation}\label{union bound}
    \mathbb{P}^{\mathrm{Rexp3P},\sigma^{-i}}\left( \cup_{r=1}^{R} E_r\right) \leq \sum_{r=1}^{R} \mathbb{P}^{\mathrm{Rexp3P},\sigma^{-i}}(E_r) \leq R \delta.
\end{equation}
Two cases may take place.

\medskip
\noindent
{\emph{Case 1}.}{ Suppose first that $T = \sum_{r=1}^{R} 2^{r-1}$. Then, 
\begin{equation*}
    \begin{aligned}
    \mathbb{P}^{\mathrm{Rexp3P},\sigma^{-i}} \left( \mathrm{Reg}((a_t)_{t=1}^{+\infty},T,C_T) \leq 
    \sum_{r=1}^{R} b_r\right) &\overset{(a)}{\geq} \mathbb{P}^{\mathrm{Rexp3P},\sigma^{-i}} \left( \sum_{r=1}^{R} \mathrm{Reg}_r((a_t)_{t=1}^{+\infty},C^r) {\leq} 
    \sum_{r=1}^{R} b_r \right) \\
    &\overset{(b)}\geq  \mathbb{P}^{\mathrm{Rexp3P},\sigma^{-i}} \left( \cap_{r=1}^{R} {E}^{c}_r \right) \overset{(c)}{=} 1 - \mathbb{P}^{\mathrm{Rexp3P},\sigma^{-i}} \left( \cup_{r=1}^{R} E_r \right) \overset{(d)}{\geq}   1-R\delta,
    \end{aligned}
\end{equation*}
\noindent
where ($a$) holds because, by construction, for each pull $r = 1,...,R$, sequence of actions $(x_t)_{t=1}^{T} \in \mathcal{S}_i(C_{T})$ and sequence of outcomes $(a_t)_{t=1}^{\infty} \in A^{\mathbb{N}}$, one has 

\begin{equation}\label{Theorem 1 inequality}
\sum_{t=L_r}^{U_r} u^i_t(x_t,a^{-i}_{t}) - u^i_t(a^{i}_{t},a^{-i}_{t}) \leq \mathrm{Reg}_r((a_t)_{t=1}^{\infty},C^{r}),
\end{equation}
which in turn, by summing over the pulls and taking the maximum with respect to $(x_t)_{t=1}^{T} \in \mathcal{S}_i(C_{T})$, imply that
\[ \left\{\sum_{r=1}^{R} \mathrm{Reg}_r((a_t)_{t=1}^{+\infty},C^r) \leq x  \right\} \subseteq \left\{\mathrm{Reg}((a_t)_{t=1}^{+\infty},T,C_{T}) \leq x \right\}, \quad \text{for all x $\in \mathbb{R}$};\]
\noindent
in addition, ($b$) holds immediately from Definition (\ref{Er set}), ($c$) holds by applying De Morgan's law and ($d$) holds because of Inequality (\ref{union bound}).
}

\medskip
\noindent
{\emph{Case 2}.}{ Now, suppose $  \sum_{r=1}^{R-1} 2^{r-1}<  T <  \sum_{r=1}^{R} 2^{r-1}$. 
Construct a new contingent plan  $\bar\sigma^{-i}$ that coincides with $\sigma^{-i}$ until period~$T$, and then it plays a constant profile of actions $\bar{a}^{-i}$ forever after, i.e., $a^{\bar\sigma^{-i}}_t = \bar{a}^{-i} \in A^{-i}$ for $t \geq T+1$. We denote by 

\[ G := \left\{ (a_t)_{t=1}^{\infty} \in A^{\mathbb{N}} \; \Bigg \lvert \; \underset{y \in A_{i}}{\max} \sum_{t=T+1}^{2^{R}-1} \left( u^i_t(y,a^{-i}_{t}) - u^i_t(a^{i}_{t},a^{-i}_{t}) \right) \geq 0 \right \}\]
\noindent
the collection of all outcomes such that their regret, relative to a static benchmark, over the segment $\{T+1,...,2^{R}-1\}$ is non-negative. Note that, by construction, from period $T+1$ up to~$2^{R}-1$, $\bar\sigma^{-i}$ is constant. This implies that
$\mathbb{P}^{\mathrm{Rexp3P},\bar\sigma^{-i}}(G) = 1.$ Therefore, the following holds
\vspace{-3mm}
\begin{equation*}
    \begin{aligned}
    \mathbb{P}^{\mathrm{Rexp3P},\sigma^{-i}} \left( \mathrm{Reg}((a_t)_{t=1}^{+\infty},T,C_T) \leq
    \sum_{r=1}^{R} b_r \right) &\overset{(a)}{=} \mathbb{P}^{\mathrm{Rexp3P},\bar\sigma^{-i}} \left( \mathrm{Reg}((a_t)_{t=1}^{+\infty},T,C_T) \leq 
    \sum_{r=1}^{R} b_r \right) \\
    &\overset{(b)}{=} \mathbb{P}^{\mathrm{Rexp3P},\bar\sigma^{-i}} \left( \left\{\mathrm{Reg}((a_t)_{t=1}^{+\infty},T,C_T) \leq 
   \sum_{r=1}^{R} b_r \right\} \cap \; G \right) \\
    &\overset{(c)}{\geq} \mathbb{P}^{\mathrm{Rexp3P},\bar\sigma^{-i}} \left( \left \{\sum_{r=1}^{R} \mathrm{Reg}_r((a_t)_{t=1}^{+\infty},C^r) \leq 
    \sum_{r=1}^{R} b_r \right \} \cap G \right) \\
    &\overset{(d)}{=} \mathbb{P}^{\mathrm{Rexp3P},\bar\sigma^{-i}} \left( \sum_{r=1}^{R} \mathrm{Reg}_r((a_t)_{t=1}^{+\infty},C^r) \leq
    \sum_{r=1}^{R} b_r  \right) \\
    &\overset{(e)}{\geq} \mathbb{P}^{\mathrm{Rexp3P},\bar\sigma^{-i}} \left( \cap_{r=1}^{R} {E}^{c}_r \right) \overset{(f)}{\geq} 1-R\delta,
    \end{aligned}
\end{equation*}
\noindent
where ($a$) holds because the contingent plans of actions $\sigma^{-i}$ and $\bar\sigma^{-i}$ coincide up to time $T$, and ($b$) and ($d$) hold because, for any measurable set $F$, $\mathbb{P}^{\mathrm{Rexp3P},\bar\sigma^{-i}}(F \cap G) = \mathbb{P}^{\mathrm{Rexp3P},\bar\sigma^{-i}}(F)$, given that  $\mathbb{P}^{\mathrm{Rexp3P},\bar\sigma^{-i}}(G) = 1$, 
\vspace{-1mm}
and ($c$) holds because, for any sequences of outcomes $(a_{t})_{t=1}^{+\infty} \in G$ and actions $(x_t)_{t=1}^{T} \in  \mathcal{S}_i(C_T)$, one has that
\begin{equation}\label{inequality G}
\begin{aligned}
 \sum_{t = 2^{R-1}}^{T} \left(u^i_t(x_t,a^{-i}_{t}) - u^i_t(a^{i}_{t},a^{-i}_{t}) \right) &\leq \sum_{t = 2^{R-1}}^{T} \left( u^i_t(x_t,a^{-i}_{t}) - u^i_t(a^{i}_{t},a^{-i}_{t}) \right)  + \underset{y \in A_i}{\max} \sum_{t = T+1}^{2^{R}-1} \left( u^i_t(y,a^{-i}_{t}) - u^i_t(a^{i}_{t},a^{-i}_{t}) \right) \\
 &\leq \mathrm{Reg}_R((a_{t})_{t=1}^{\infty},C^{R}),
\end{aligned}
\end{equation}
where the first inequality holds because $(a_t)_{t=1}^{\infty} \in G$ and the second one holds because, as by assumption $C_{2^{\lceil \log_2(T+1) \rceil} -1 } + 1 \leq 2^{\lceil \log_2(T+1) \rceil - 1}-1$, one has that $C^{R} = C_{2^{R}-1}+1$ and so $C^{R} = C_{2^{R}-1}+~1 = C_{2^{\lceil \log_2(T+1) \rceil}-1}+1 \geq C_T +~1$, while the sequence $((x_t)_{t=2^{R-1}}^{T},\underbrace{y,...,y}_{2^{R}-T-1} )$ has at most $C_T+1$ action changes, for any $y \in \text{argmax} \left\{ \sum_{t = T+1}^{2^{R}-1} u^i_t(y,a^{-i}_{t}) - u^i_t(a^{i}_{t},a^{-i}_{t}) \right \}$. That is, 
\[
((x_t)_{t=2^{R-1}}^{T},\underbrace{y,...,y}_{2^{R}-T-1} ) \in \left \{ (y_t)_{t=2^{R-1}}^{2^{R}-1} \in A_i^{2^{R-1}} \; \Bigg \lvert \; \sum_{t=2^{R-1}}^{2^{R}-1} \mathbf{1}(y_{t+1} \neq y_t) \leq C^{R} \right\},
\]
and so such a sequence is feasible for a dynamic benchmark with $C^{R}$ action changes over the last pull, which implies Inequality (\ref{inequality G}). Thus, by using (\ref{inequality G}), one obtains
\[ \left\{\sum_{r=1}^{R} \mathrm{Reg}_r((a_t)_{t=1}^{T},C^r) \leq x \right\}  \cap G  \subseteq \{\mathrm{Reg}((a_t)_{t=1}^{T},T,C_{T}) \leq x \} \cap G , \quad \text{for all x $\in \mathbb{R}$}.\]
}
Finally, ($e$) and ($f$) hold by \emph{Case} 1. Together the two cases imply that, for any $T$, one has 
\begin{equation}\label{Main inequality}
  \mathbb{P}^{\mathrm{Rexp3P},\sigma^{-i}} \left( \mathrm{Reg}((a_t)_{t=1}^{+\infty},T,C_T) \leq
    \sum_{r=1}^{R} b_r \right) \geq 1-R\delta.  
\end{equation}

We conclude the argument by applying the Borel-Cantelli lemma and by showing that $\sum_{r=1}^{
R} b_r = o(T)$. Note that for pulls from second to second to last, i.e. $r=2,...,R-1$, one has
\begin{equation}\label{2T+1 bound}
\frac{C^{r}}{2^{r-1}} = \min \left\{ \frac{2^{r-1}-1}{2^{r-1}},\frac{C_{\sum_{j=1}^{r} 2^{j-1}}+1}{2^{r-1}} \right\} \geq \min \left\{ \frac{1}{2}, \frac{C_{\sum_{j=1}^{r} 2^{j-1}}}{\sum_{j=1}^{r} 2^{j-1}} \right \} \overset{(a)}{\geq} \min \left\{ \frac{1}{2}, \frac{C_{T}}{T} \right \} \overset{(b)}{\geq} \min \left\{ \frac{1}{2}, \frac{C_{2T+1}}{2T+1} \right \},    
\end{equation}
\noindent
where both ($a$) and ($b$) hold because, by assumption, $C_T/T$ is monotonically decreasing to zero. Moreover, it holds that $\frac{C^{1}}{1} = 0$ and $\frac{C^{R}}{2^{R-1}} \geq \min \left\{ \frac{1}{2}, \frac{C_{2T+1}}{2T+1} \right \}$, as the end of the last pull is $2^{R}-1 = 2^{\lceil \log_2(T+1) \rceil} - 1 \leq 2T + 1$ and $C_T/T$ is monotone decreasing.
Therefore, under the natural convention $C^r \ln \left( \frac{3 \times 2^{r-1} K}{C^{r}} \right) = 0$, whenever $C^{r} = 0$, one obtains
\begin{equation*}
\begin{aligned}
\sum_{r=1}^{R} b_r &= \sum_{r=1}^{R} 7 \sqrt{2^{r-1} K \left( C^r \ln \left( \frac{3 \times 2^{r-1} K}{C^r}\right) + 2 \ln(K)\right)} + \sqrt{\frac{2^{r-1} K }{C^r \ln \left( \frac{3 \times 2^{r-1} K}{C^r}\right) + 2 \ln(K)}}\ln\left( \frac{1}{\delta}\right) \\
&\overset{(a)}{\leq} \sum_{r=1}^{R} 7 \sqrt{2^{r-1} K \left( -C^r \ln \left( \frac{1}{3K} \frac{C^r}{2^{r-1}} \right) + 2 \ln(K)\right)} + \sqrt{\frac{2^{r-1} K }{2 \ln(K)}}\ln\left( \frac{1}{\delta}\right) \\
&\overset{(b)}{\leq} 7 \sqrt{K \left( -C_{2T+1} \ln \left( \frac{1}{3 K} \min\left\{\frac{1}{2},\frac{C_{2T+1}}{2T+1}\right\} \right) + 2\ln(K) \right)} \sum_{r=1}^{R} \sqrt{2^{r-1}} + \sqrt{\frac{K}{2 \ln(K)}} \ln \left(\frac{1}{\delta} \right) \sum_{r=1}^{R} \sqrt{2^{r-1}}  \\
&\leq \underbrace{14\left(1+\sqrt{2}\right) \left( \sqrt{K T \left( -C_{2T+1} \ln \left( \frac{1}{3 K} \min\left\{\frac{1}{2},\frac{C_{2T+1}}{2T+1}\right\} \right) + 2\ln(K) \right)} + \sqrt{\frac{K}{2 \ln(K)}} \ln \left(\frac{1}{\delta} \right) \sqrt{T} \right)}_{:= U(T,C_{2T+1},\delta)},
\end{aligned}
\end{equation*}
\noindent
where ($a$) holds because  $C^r \ln \left( \frac{3 \times 2^{r-1} K}{C^{r}} \right) \geq 0$ for all pulls r, and ($b$) holds by applying Inequality (\ref{2T+1 bound}).

\noindent
Finally, set $\delta = \frac{1}{T^2}$. Given that $\sum_{r=1}^{R} b_r \leq U(T,C_{2T+1},\frac{1}{T^2})$, $U(T,C_{2T+1},\frac{1}{T^2}) = o(T)$ and $\sum_{T=1}^{\infty} \frac{\lceil \log_2(T+1)\rceil }{T^2} < \infty$, by applying the Borel-Cantelli lemma and Inequality (\ref{Main inequality}), one concludes that
\[ 
\mathbb{P}^{\mathrm{Rexp3P},\sigma^{-i}} \left(  \underset{T \rightarrow +\infty}{\text{limsup}} \; \frac{1}{T} \max_{(x_t) \in \mathcal{S}^{i}(C_T)}\mathrm{Reg}^i (\mathcal{U}^{i}_T,(a_t),(x_t)) \leq 
    0  \right) = 1,
\]
that is Rexp3P is DB($C$) consistent.
}
\vspace{-5mm}
\noindent
\paragraph{Part $\boldsymbol{(ii)}$.}{We first show the result by assuming that 
$C_T = \alpha T$, for some $\alpha \in (0,1]$, and then argue how to extend it to the case $\text{limsup}_{T \rightarrow \infty} C_T/T >0$. We take the perspective of player $1$. Suppose there exist $a^{-1}_{(1)} \neq a^{-1}_{(2)} \in (A^{-1})^2$ such that 
\[ \underset{x \in A^1}{\text{argmax}} \; u^1(x,a^{-1}_{(1)}) \cap \underset{y \in A^1}{\text{argmax}} \; u^1(y,a^{-1}_{(2)}) = \emptyset.\]
\noindent
Let $a_{(1)}^{1} \in \underset{x \in A^1}{\text{argmax}} \; u^1(x,a^{-1}_{(1)})$ and $a_{(2)}^{1} \in \underset{y \in A^1}{\text{argmax}} \; u^1(y,a^{-1}_{(2)})$.

We can represent the payoff of player $1$ under $a^{-1}_{(1)}$ and $a^{-1}_{(2)}$ in the following~table:
\begin{table}[H]
\centering
\begin{tabular}{|l|l|l|}
\hline
\emph{row} player' payoff  & $a^{-1}_{(1)}$ & $a^{-1}_{(2)}$ \\ \hline
$a^{1}_{(1)}$ & $u^1(a^{1}_{(1)},a^{-1}_{(1)})$ & $u^1(a^{1}_{(1)},a^{-1}_{(2)})$ \\ \hline
$a^{1}_{(2)}$ & $u^1(a^{1}_{(2)},a^{-1}_{(1)})$ & $u^1(a^{1}_{(2)},a^{-1}_{(2)})$ \\ \hline
$\vdots$ & $\vdots$ & $\vdots$ \\ \hline
$a_{(K_i)}^{1}$ & $u^1(a_{(K_i)}^{1},a^{-1}_{(1)})$ & $u^1(a_{(K_i)}^{1},a^{-1}_{(2)})$ \\ \hline
\end{tabular}
\end{table} 

Let $d \in \mathbb{N}$ such that $d \geq \frac{3}{\alpha}$ and $p \in (0,1)$. Without loss of generality suppose that $T \geq d(\alpha d - 1)$. We construct the sequence of distributions by the other players in the following way: 
\[\forall t \in \mathbb{N}, \quad \sigma_t = \begin{cases}
        \delta_{a^{-1}_{(1)}}  &\text{if \;$t \in \{1,...,d-\lfloor \alpha d \rfloor+2, d+1,...,2d - \lfloor \alpha d \rfloor+2,... \}$,}
        \\
        p \delta_{a^{-1}_{(1)}} \otimes (1-p) \delta_{a^{-1}_{(2)}}   &\text{otherwise}.
        \end{cases}\]
Namely, over segments of constant length $d$, one has that $\sigma$ plays for the first $d - \lfloor \alpha d \rfloor + 2$ periods~$a^{-1}_{(1)}$ and then randomizes between $a^{-1}_{(1)}$ and $a^{-1}_{(2)}$ (with respective probabilities $p$ and $1-p$) for the remaining $\lfloor \alpha d \rfloor -2$ periods. Note that $\sigma$ can change action at most $\left \lceil \frac{T}{d} \right \rceil (\lfloor \alpha d\rfloor - 1)$ times throughout the horizon of length $T$. Assumptions over~$T$ and $d$ ensure that $\left \lceil \frac{T}{d} \right \rceil (\lfloor \alpha d\rfloor - 1) \leq \lfloor \alpha T \rfloor $. Indeed, 

\[ \left \lceil \frac{T}{d} \right \rceil (\lfloor \alpha d\rfloor - 1) \overset{(a)}{<} \left( \frac{T}{d} +1 \right) (\lfloor \alpha d\rfloor - 1) \overset{(b)}{\leq} \left( \frac{T}{d} +1 \right)( \alpha d - 1) = \alpha T + \left(\alpha d - 1 - \frac{T}{d} \right) \overset{(c)}{\leq} \alpha T, \]

where $(a)$ holds because $\lceil x \rceil < x + 1$ and, by assumption, $\alpha d \geq 3$, $(b)$ holds because $\lfloor x \rfloor \leq x$ and~$(c)$ holds because, by assumption, $T \geq d(\alpha d - 1)$. We first lower bound the expected regret, i.e.,
\[
\mathbb{E}^{\pi,\sigma} \left[ \underbrace{\underset{x \in \mathcal{S}\left( \alpha T \right)}{\text{max}} \sum_{t=1}^{T} \left(u^1(x_t,a^{-i}_t) - u^1(a^{\pi}_t,a^{-i}_t) \right)}_{:= R_T}  \right].
\]

Let $a_{*}^1 \in \text{argmax}_{x \in A^1} \{pu^1(x,a^{-1}_{(1)}) + (1-p)u^1(x,a^{-1}_{(2)})\}$.
Define $\pi^{*}$ for player 1 in the following way:
\[\forall t \in \mathbb{N}, \quad \pi^{*}_t = \begin{cases}
        \delta_{a_{1}^{1}}  &\text{if \;$t \in \{1,...,d-\lfloor \alpha d \rfloor + 2, d+1,...,2d - \lfloor \alpha d \rfloor + 2,... \}$,}
        \\
        \delta_{a^1_{*}}   &\text{otherwise}.
        \end{cases}\]

By construction, $\pi^{*}$ minimizes the expected regret of the player 1. Define the regrets for playing action $a_{*}^1$ rather than the best replies to $a_{(1)}^{-1}$ and $a_{(2)}^{-1}$ as
\begin{equation*}
    \begin{aligned}
    \Delta_1 &:= u^1(a_{(1)}^{1},a^{-1}_{(1)}) - u^1(a^1_{*},a^{-1}_{(1)}) \geq 0\\
    \Delta_2 &:= u^1(a_{(2)}^{1},a^{-1}_{(2)}) - u^1(a^1_{*},a^{-1}_{(2)}) \geq 0.
    \end{aligned}
\end{equation*}
Note that the expected regret of $\pi^{*}$ over a time period in which $\sigma$ randomizes is~$\Delta = p\Delta_1 + (1-p)\Delta_2$. In particular, observe that $\underset{x \in A^1}{\text{argmax}} \; u^1(x,a^{-1}_{(1)}) \cap \underset{y \in A^1}{\text{argmax}} \; u^1(y,a^{-1}_{(2)}) = \emptyset$ and $p \in (0,1)$ imply that $\Delta > 0$. Indeed, as $p \in (0,1)$, $\Delta = 0$ if and only if $\Delta_1 = \Delta_2 = 0$. However, this implies that $a^1_{*}$ is a best reply to both $a{1}_{-1}$ and $a_{(2)}^{-1}$, thus violating disjointness.
Therefore, for any policy ${\pi}$ of player 1, one has
\begin{equation*}
\begin{aligned}
\frac{1}{T}\mathbb{E}^{{\pi},\sigma}[R_T] \geq \underset{\pi \in \mathcal{P}_{1}^{DB}(C)}{\text{min}} \frac{1}{T} \; \mathbb{E}^{\pi,\sigma} \left[ R_T \right] =\frac{1}{T}\mathbb{E}^{{\pi^*},\sigma}[R_T] &\overset{(a)}{\geq}  \frac{1}{T} \left \lfloor \frac{T}{d} \right \rfloor \left( \lfloor \alpha d \rfloor -2 \right) \Delta \overset{(b)}{\geq} \frac{ \left \lfloor \alpha d  \right \rfloor -2}{2 d} \Delta \overset{(c)}{>} 0,
\end{aligned}
\end{equation*}
where ($a$) holds because in expectation, $\pi^{*}$ incurs positive regret in periods in which $\sigma$ randomizes, and there are at least $\lfloor T/d \rfloor$ such periods over a horizon of length~$T$; ($b$) holds because, by assumption, $T \geq d(\alpha d -2)$ and $\alpha d \geq 3$; ($c$) holds because $\Delta > 0$ and , by assumption, $d \geq \frac{3}{\alpha}$. We denote $\theta := \frac{ \left \lfloor \alpha d  \right \rfloor -2}{2d} \Delta$, and, to ease notation, we neglect the dependence of the probability $\mathbb{P}$, the expectation $\mathbb{E}$, and the variance~$\mathbb{V}$ on $\sigma$ and the policy $\pi$ selected by player 1. One has that
\begin{equation}\label{Paley-Zigmund}
    \begin{aligned}
    \mathbb{P} \left( \frac{R_T}{T} \geq \frac{\theta}{2} \right) &\overset{(a)}{\geq} \mathbb{P} \left( \frac{R_T}{T} \geq \frac{\mathbb{E}[R_T]}{2 T } \right)  \\
    &\overset{(b)}{\geq} \frac{1}{4} \times \frac{ \mathbb{E}^{2}[R_T / T]}{\mathbb{E}[R^2_T/T^2]}  \\
    &\overset{(c)}{=} \frac{1}{4} \times \frac{\mathbb{E}^2[R_T/T]}{\mathbb{V}[R_T/T] + \mathbb{E}^2[R_T/T]}  \\
    &\overset{(d)}{\geq} \frac{1}{4} \times \frac{\mathbb{E}^2[R_T/T]}{1/4 + \mathbb{E}^2[R_T/T]}   \\
    &\overset{(e)}{\geq}  \frac{\theta^2}{1+4\theta^2},
    \end{aligned}
\end{equation}
where ($a$) holds because $\frac{E[R_T]}{T} \geq \theta$ implies $\left \{ \frac{R_T}{T} \geq \frac{\mathbb{E}[R_T]}{2T} \right \} \subseteq \left \{ \frac{R_T}{T} \geq \frac{\theta}{2} \right \}$; ($b$) holds by applying the Paley-Zygmund Inequality, which can be applied since the dynamic benchmark is able to best-reply at every period of the game implying that the (realized) regret $R_T$ is non-negative; ($c$) holds because~$\mathbb{V}(X) = \mathbb{E}(X^2) - \mathbb{E}^2(X)$; ($d$) holds because $\mathbb{V}(X) \leq \frac{1}{4}(M-m)^2,\; \text{if} \; X \in [m,M]$ where $m=0$ and $M = 1$; and ($e$) holds because $\mathbb{E}[R_T/T] \geq \theta,\;\text{and}\; \frac{x^2}{1+4x^2}$ is monotone increasing for $x \geq 0 $. 

\[ \mathbb{P} \left( \underset{T \rightarrow \infty}{\text{limsup}} \; \frac{R_T}{T} > 0 \right) \overset{(a)}{\geq} \mathbb{P} \left( \underset{T \rightarrow \infty}{\text{limsup}} \; \frac{R_T}{T} \geq \frac{\theta}{2} \right) \geq \mathbb{P} \left( \bigcap_{s=1}^{+\infty} \bigcup_{T \geq s} \left\{ \frac{R_T}{T} {\geq} \frac{\theta}{2}  \right\}\right ) \overset{(b)}{\geq}  \frac{\theta^2}{1+4\theta^2} \overset{(c)}{>} 0\]
\noindent
where $(a)$ and $(c)$ hold because $\theta > 0$ and $(b)$ holds by applying the reversed Fatou's lemma in (\ref{Paley-Zigmund}).

Finally, observe that $\text{limsup}_T C_T/T >0$ implies the existence of a subsequence $(C_{T_k}: k \geq 1)$ such that, for all $k \in \mathbb{N}$, $C_{T_k} \geq \alpha T_k$ with $\alpha \in (0,1]$. As the above construction is independent of $T$, we can repeat the same procedure over the subsequence of horizons $(T_k)_{k=1}^{\infty}$ and obtain 

\begin{equation}\label{result for subsequence}
\mathbb{P} \left( \underset{k \rightarrow \infty}{\text{limsup}} \; \frac{R_{T_k}}{T_{k}} > 0 \right) > 0.
\end{equation}

By definition of limsup, (\ref{result for subsequence}) implies $\mathbb{P} \left( \underset{T \rightarrow \infty}{\text{limsup}} \; \frac{R_T}{T} > 0 \right) > 0$. This establishes that $\mathcal{P}^{DB}_{1}(C) = \emptyset$, whenever $\text{limsup}_T C_T/T >0$ and there exist $a^{-1}_{(1)} \neq a^{-1}_{(2)} \in (A^{-1})^2$ such that $\underset{x \in A^1}{\text{argmax}} \; u^1(x,a^{-1}_{(1)}) \cap \underset{y \in A^1}{\text{argmax}} \; u^1(y,a^{-1}_{(2)}) = \emptyset$. On the other hand, suppose that, for any $a^{-i}_{(1)}, A^{-i}_{(2)} \in (A^{-i})^2$, it holds 
\vspace{-0.2cm}
\[ \underset{x \in A_i}{\mathrm{argmax}} \; u^1(x,a^{-i}_{(1)}) \cap \underset{y \in A_i}{\mathrm{argmax}} \; u^1(y,a^{-i}_{(2)}) \neq \emptyset.\]

Then, the static benchmark performs as well as the unconstrained benchmark. Therefore,~$(i)$ implies the possibility to compete relative to it.

}

\end{document}